\documentclass[12pt,reqno]{amsart}
\usepackage{amsfonts}
\usepackage[samelinetheorem,notitle]{maherart}
\usepackage{enumerate}
\usepackage{bm}
\usepackage{bbm}
\usepackage{tikz}
\usepackage{float}
\usepackage{booktabs}
\usepackage{threeparttable}
\usepackage{siunitx}
\usepackage{array}
\usepackage{adjustbox}
\usepackage{amsmath}
\usepackage{amsthm}
\usepackage{chngcntr}
\usepackage{textcomp}
\usepackage{esint}
\usepackage{paralist}
\usepackage{graphics}
\usepackage{epsfig}
\usepackage{graphicx}
\usepackage{longtable}
\usepackage{caption}
\usepackage{url}
\usepackage{epstopdf}
\usepackage{multirow}
\usepackage{natbib}
\makeatletter
\newcommand\fs@boxedtop
  {\fs@boxed
   \def\@fs@mid{\vspace\abovecaptionskip\relax}%
   \let\@fs@iftopcapt\iftrue
  }
\makeatother
\floatstyle{boxedtop}
\floatname{framedbox}{Procedure}
\newfloat{framedbox}{tbp}{lob}

\hypersetup{
    pdftitle =
        {The Wisdom of the Crowd and Higher-Order Beliefs},
    pdfauthor =
        {Yi-Chun Chen, Manuel Mueller-Frank and Mallesh M Pai},
    pdfsubject=
        {The Wisdom of the Crowd and Higher-Order Beliefs}
}

\newcommand{\tsiw}{\tilde{s}_{i,\omega}}

\newcommand{\tmu}{\tilde{\mu}}

\usepackage{thmtools}
\usepackage{thm-restate}

\begin{document}
\raggedbottom
\emergencystretch=2em
\title{The Wisdom of the Crowd and Higher-Order Beliefs}
\author[Chen]{Yi-Chun Chen$^\text{A}$}
\address{$^{\text{\MakeLowercase{a}}}$Department of Economics and Risk
Management Institute, National University of Singapore\\
\href{mailto:ecsycc@nus.edu.sg}{ecsycc@nus.edu.sg}}
\author[Mueller-Frank]{Manuel Mueller-Frank$^\text{B}$}
\address{$^\text{\MakeLowercase{b}}$IESE Business School, University of
Navarra\\
\href{mailto:mmuellerfrank@iese.edu}{mmuellerfrank@iese.edu}}
\author[Pai]{Mallesh M. Pai$^\text{C}$}
\address{$^\text{\MakeLowercase{c}}$Department of Economics, Rice University%
\\
\href{mailto:mallesh.pai@rice.edu}{mallesh.pai@rice.edu}}
\address{\today}
\thanks{We are especially grateful to Paul Huang Enze for his exceptional research assistance on this paper, including helping with executing the experimental work. We also thank Samuel Ho, Mingzi Niu, Hang Shao, and Xinhan Zhang for excellent research assistance. We thank Itai Arieli, Rahul Deb, Ben Golub, Tai-Wei Hu, Yanwei Jia, Scott Kominers, Ran Smorodinsky, Tom Wilkening; and seminar/conference audiences at a16z Crypto Research, Caltech, Games 2020, IIM Ahmedabad, Bonn, IDEAL Institute at Northwestern, and the Summer Game Theory Workshop 2025 for helpful comments and discussions.  Chen gratefully acknowledges the financial support from the Singapore Ministry of Education Academic Research Fund Tier 1. Mueller-Frank gratefully
acknowledges the financial support of the Spanish Ministry of Science,
Innovation and Universities (Ref. ECO2015-63711-P: MINECO/FEDER, UE) and the
Department of Economy and Knowledge of the Generalitat de Catalunya (Ref.
2017 SGR 1244). Pai gratefully acknowledges support from the NSF
(CCF-1763349). The paper source files, the experimental data, and the full open-source replication package are publicly available at \url{https://github.com/malleshpai/robust-aggregation}. }

\begin{abstract}
We propose a new simple procedure called \emph{Population-Mean-Based Aggregation} (PMBA) that enables a principal to ``aggregate'' information about an unknown state of the world from agents without understanding the information structure among them. PMBA only requires agents to communicate their beliefs about the state, and some agents to communicate their expectations of the population average belief. In a large population, for any finite number of possible states, and under weak assumptions on the information structure, allowing individual agents' beliefs to be misspecified, we show that PMBA infers the true state (in probability or almost surely under the stated conditions). We show how PMBA can be reinterpreted as a linear regression procedure, and how it can be used to aggregate information from a finite number of agents, allowing us to reuse existing results on inference in linear models. We conduct a novel experiment to show that the real-world performance of our procedure exceeds that of existing methods.

\bigskip

\end{abstract}

\maketitle

\section{Introduction}
A long-held belief, starting from at least Aristotle, is that of the ``wisdom of the crowd''. This refers to the idea that even though some event may be uncertain, and the information each individual has may be noisy, the ``aggregate'' of individual beliefs is accurate. Such aggregation obviously has large social and private value, especially when it concerns important social or economic events.
As a consequence, there has been substantial effort toward both improving existing institutions (e.g., polling), and designing and implementing new methods (e.g., prediction markets).%
\footnote{There is also a large and influential body of literature that studies the possibility of such aggregation within our current institutions, e.g., voting (see e.g.,~\cite{austen1996information},~\cite{feddersen1997voting}), markets (e.g.,~\cite{grossman1980impossibility}) and social networks (see e.g.,~\cite{golub2010naive},~\cite{acemoglu2011bayesian}).}
At a theoretical level, then, it is important to understand the limits of such an exercise: without specific assumptions on the nature of the information agents have, how can we aggregate it, and how much can we learn from it? In short, how wise is the crowd, and how can we tap into its wisdom?

To fix ideas, imagine a population of interest and an event whose realization is uncertain. Individual agents each see a signal that is informative of the realization. Agents also understand the informational environment, i.e., what this signal implies for their own beliefs and the implied distribution of other agents' beliefs. Even assuming away the difficulties (both logistical and strategic) of eliciting the relevant information from the population, the question we address in this paper is: Does there exist a general procedure to aggregate individuals' information without knowledge of the information structure among them? 

One important insight from the existing literature is that even if the state is binary, and agents' signals are i.i.d.\ conditional on the state, knowing the first-order beliefs of even an infinite set of agents is generally insufficient to learn the state.\footnote{See~\cite{prelec2017solution} and~\cite{arieli2017crowd}.} In particular, even in these idealized conditions, most agents may be wrong, i.e., most agents may place a posterior probability of larger than 50\% on the incorrect state of the world. The main takeaway from~\cite{arieli2017crowd} is that without further strong assumptions about the signal structure, the principal cannot learn the state of the world from solely agents' first-order beliefs, even with an infinite number of agents in the population. In econometric terms, there is an identification problem.

We answer our question in the affirmative by proposing a novel, simple procedure we name \emph{Population-Mean-Based Aggregation} (herein PMBA). Our procedure uses agents' beliefs about the uncertain event as well as some agents' expectations of the population mean belief. Under weak conditions on the information structure generating the agents' signals, this procedure \emph{fully aggregates} the information among agents in a large population. In other words, using information from agents, solely about agents' own beliefs, and their expectations of the average belief in the population (the \emph{population average belief}), but with no other knowledge of the underlying information structure, we can fully learn the underlying state. Importantly, we show that our procedure is \emph{robust}: it allows for aggregation of beliefs even when beliefs are misspecified at both the individual and aggregate levels. That is to say, we start from a ``correctly specified'' model where agents have beliefs and higher-order beliefs from a common prior and show that our PMBA procedure correctly aggregates agents' information without knowledge of the underlying common prior. We then show that this procedure continues to correctly aggregate even if individual beliefs are \emph{arbitrarily} misspecified, as long as aggregate misspecification is not too large (in a sense we make precise below). We show how this procedure can be reinterpreted as a linear regression procedure, and how it can be used to aggregate information from a finite number of agents. The machinery of linear regression then allows us to show various forms of robustness of our procedure. In short, our procedure is a viable candidate to be used in practice.  

The idea of eliciting and using higher-order beliefs is not new. In our context, we note the elegant ``Surprisingly Popular'' (herein, SP) procedure developed by~\cite{prelec2017solution} (henceforth PSM), who pioneered the use of second-order beliefs for the purpose of information aggregation (however, their procedure aggregates information on a narrow set of information structures---we discuss this in more detail in Section~\ref{sec:rellit}). Another close comparison is to the Surprisingly Confident (SC) procedure of~\cite{doi:10.1287/mnsc.2020.3919}, who also leverage the linearity of the correspondence between beliefs and higher-order beliefs, and its implication for the ordering of average beliefs and average second-order beliefs in each state (see also Theorem~1.4 of~\cite{prelec2017solution} and contemporaneous approaches by~\cite{libgober2021hypothetical} and~\cite{prelec2022general}, described in Section~\ref{sec:rellit}). A key differentiator is that SP/SC are both only guaranteed to work under binary states, i.e., two possible states of the world, and require stronger assumptions more generally, as we discuss below.

In line with this literature, it is important to note that we are solely using agents' beliefs and expected higher-order beliefs. Such information has been robustly elicited in laboratory and real-world settings, and therefore, our procedure is amenable to implementation. We show in Section~\ref{sec:experiment} that our procedure performs well in our experimental setting. By contrast, theoretically feasible procedures that ask agents for more detailed information such as state-contingent beliefs (for example, the agent's expectation of the average belief in the population \emph{conditional} on the true state of the world) may not be reliable or even feasible in practice. There is a wealth of evidence that agents struggle to report richer details (such as state-contingent beliefs).\footnote{For example,~\cite{li2017obviously} states: ``More generally, laboratory subjects find it difficult to reason state-by-state about hypothetical scenarios~\citep{charness2009winnerscurse, esponda2014hypothetical, ngangoue2021learning}. This mental process, often called ``contingent reasoning,'' has received little formal treatment in economic theory. In subsequent work,~\cite{vespa2019contingent} investigate axioms governing contingent reasoning in single-agent decision problems.''}
Similarly, we may not be able to elicit the underlying prior; we take seriously the idea that the prior/ex-ante stage is an ``as-if'' assumption.\footnote{See for example the discussion on page 926 in~\cite{gul1998comment}.}

\subsection{Summary of Model and Results}
Our baseline environment (Section~\ref{sec:pmba}) features finitely many states and many agents whose signals are drawn from an information structure that is \emph{unknown to the principal} but understood by the agents. The contributions of the paper can be read as three complementary steps:
\begin{enumerate}
    \item \textbf{Large-population benchmark.} Theorem~\ref{thm:pmba} shows that, under weak regularity conditions, the principal can recover the true state by combining the population mean of first-order beliefs with higher-order reports from at least as many agents as there are states.\footnote{Agents interact with the mechanism at the interim stage, so asking for priors is neither meaningful nor feasible; see also the discussion above on why the ex-ante prior is often not directly elicitable.} The higher-order data resolve the indeterminacy that prevents standard crowd wisdom schemes from working when there are more than two states---an issue emphasised by~\cite{prelec2017solution}.
    \item \textbf{Finite-sample and misspecification robustness.} Section~\ref{sec:linear} reinterprets the inversion step as a linear regression problem, so that classical tools deliver inference (including confidence intervals), diagnostics, and robustness checks, and Section~\ref{sec:experiment} implements an IV-based version when endogeneity in reported beliefs is a concern. This perspective clarifies how much heterogeneity in misspecification can be tolerated before aggregation fails.
    \item \textbf{Empirical performance.} Section~\ref{sec:experiment} reports a new experiment in which subjects forecast NFT price ranges. PMBA dominates majority voting and the Surprisingly Popular procedure in this setting, indicating that the additional elicitation burden is repaid in improved accuracy.
\end{enumerate}

The core intuition is simple. Suppose the principal elicits each agent's first-order belief about the unknown state. By the law of large numbers, the sample mean converges to the expected population belief in the true state. Without knowing the mapping from states to those expectations, the principal cannot yet identify the state. Let $L$ denote the number of states. By additionally querying $L$ agents (one per state) for their forecast of the population mean, the law of iterated expectations furnishes a system of $L$ linear equations whose unique solution is precisely that unknown mapping. The rest of the paper develops the conditions under which this logic survives realistic complications and documents its performance.

\section{Population-Mean-Based Aggregation }\label{sec:pmba}
In this section, we introduce our core model; we describe our Population-Mean-Based Aggregation procedure (herein PMBA) and prove its validity. 

\subsection{General Model and Notation}
There is a population of agents, $N=\{1,2, \ldots \}$, which may be countably infinite (this section), or finite (Section~\ref{sec:linear}), with $i$ denoting a generic agent. There is a finite set of $L$ states of the world $\Omega = \{\omega_1, \omega_2, \ldots, \omega_L \}$, so $L=|\Omega|$ is the number of states. The space of feasible beliefs over $\Omega$ is the $(L-1)$-dimensional simplex $\Delta^L$ (i.e., the set of probability vectors in $\mathbb{R}^L$); we denote it by $\Delta(\Omega)$ when emphasizing the state space.

Each agent $i$ observes a signal in a set $S$ with the associated sigma field $\mathcal{F}$. To allow for correlation, let us denote the common prior over $\Omega \times S^N$ by $P$. Given that an agent observes some signal $s \in S$, their posterior belief that the state is $\omega$ is given by $P(\omega\mid s)$. We let $P_{i,\omega}^S$ denote the induced marginal distribution of $P$ over signals for agent $i$ in state $\omega$. Let us denote the posterior beliefs of agent $i$ by $\mu_i(s) \in \Delta^L$ and $\tmu_i$ as the associated random variable. In what follows, the actual signal seen by the agent will often be irrelevant since we are concerned with their beliefs. We will, therefore, suppress the dependence on the signal in our notation. We define
\begin{align*}
    &m_i(\omega) = \mathbb{E}[\tmu_i|\omega].
\end{align*}
In words, $m_i(\omega)$ is the expected belief of agent $i$ in state $\omega$. 

Finally, we define $\bar{\mu}(\omega) = \lim_{n \to \infty} \frac1n \sum_{i=1}^nm_i(\omega)$.\footnote{Note that without further assumptions, it is possible that this limit might not exist and/or might be different depending on the ordering of agents. We ignore these possibilities. Simple sufficient conditions (e.g., conditional i.i.d.\ signals) rule out these concerns with an infinite population, and they are moot for any arbitrarily large finite population.} That is, $\bar{\mu}(\omega)$ is the limit of the average of agents' conditional expectations $\mathbb{E}[\tilde{\mu}_i\mid\omega]$; in other words, the expected population average belief in state $\omega$.

Given a vector $x \in \Delta(\Omega)$, we will let $x_{\omega}$ refer to the component corresponding to $\omega$; i.e., $\mu_{i,\omega}$ is the belief of agent $i$ that the state is $\omega$, $\tmu_{i,\omega}$ is the random variable of agent $i$'s belief that the state is $\omega$, etc.

To keep notation consistent across sections, we use the following conventions throughout: hats (e.g., $\hat\mu$) denote sample averages, bars (e.g., $\bar\mu(\omega)$) denote state-contingent population limits, tildes denote random objects before conditioning, and bold symbols denote stacked vectors or matrices of these quantities.

\subsection{PMBA Procedure}\label{sec:infinite}
Before we introduce our aggregation procedure, we introduce the following assumption on the information structure.
\begin{assumption}[No Aggregate Uncertainty]\label{ass:1}
    We make the following assumptions on $P$:
    \begin{enumerate}
    \item Imperfectly informed agents: for each agent $i$, we have that $P_{i, \omega}^S$ and $P_{i, \omega'}^S$ are mutually absolutely continuous with respect to each other.
    \item Limited correlation: For every $\epsilon>0$, there exists a finite $n(\epsilon )$ such that conditional on the state, for any agent $i$, all but at most $n(\epsilon )$ agents' signals are  $\epsilon $-independent of agent $i$'s signal. Formally:
    \begin{align*}
    &\forall \epsilon>0, \exists n(\epsilon) \in \mathbb{N} \text{ s.t. }\forall i, \exists N_i \subseteq N, |N_i| \leq n(\epsilon),\\
    &\forall j \in N \backslash N_i, \forall E, E' \in \mathcal{F}, \forall \omega: \bigl|P(\tsiw \in E,\, \tilde{s}_j \in E' \mid \omega) - P_{i,\omega}^S(E)P_{j,\omega}^S(E')\bigr| \leq \epsilon,
    \end{align*}
    \end{enumerate}
Here $\tsiw$ denotes the random variable of agent $i$'s signal (under $P(\cdot\mid\omega)$), and $\tilde{s}_j$ that of agent $j$'s signal; the conditioning is on the realized state $\omega$.
    \end{assumption}

This assumption requires some explanation. The first condition is simply to make the problem interesting/non-trivial: after all, if an agent is perfectly informed, then aggregation is trivial. Next, we need an assumption regarding the correlations between agents' signals. Intuitively, assuming conditionally i.i.d.\ signals is likely overly strong and does not capture the idea that agents' information may be correlated. However, we cannot allow for arbitrary correlation---for example, if all agents' information is perfectly correlated, then effectively, there is only one unique opinion in the population, and no further aggregation is possible. Our notion allows for ``enough independence,'' thereby roughly guaranteeing that for any agent, most other agents' signals are approximately independent conditional on the state. This ensures that the law of large numbers holds, so that the population average belief is a deterministic function of the state, i.e., in the notation of our model, in any state $\omega$ we have that
    \[
    \text{as } n \to \infty,\;\; \frac1n \sum_{i=1}^n \tilde{\mu}_i \to_P \bar{\mu}(\omega).
    \]
We emphasize that while this assumption relaxes the conditional i.i.d.\ assumption that is regularly made, it is a non-trivial assumption. Whether it is sensible or not depends on the context. For instance, in a political context, if all agents are drawing their information from the same few noisy news sources, then there will be correlations in beliefs and our assumption would not be well-motivated for such a setting. Conversely, if we are considering subjective but idiosyncratic product reviews, then the assumption is far more natural. More generally, the clause is intended to capture environments with enough conditional independence for a law of large numbers argument; residual systematic bias/noise is then handled by the misspecification analysis in Section~\ref{sec:misspecified}.

\begin{framedbox}[htbp]
  \caption{Population-Mean-Based Aggregation}\label{proc:pmba2}
  \begin{minipage}{0.95\textwidth}
\begin{enumerate}
    \item Elicit, from each agent $i \in N$, their belief $\mu_i$. Given $n$ elicited agents, compute the empirical mean
    \begin{align*}
        \hat{\mu} = \frac1n \sum_{i=1}^n \mu_i,
    \end{align*}
    which converges in probability to $\bar{\mu}(\omega)$ under Assumption~\ref{ass:1}.
    \item Select $L$ agents $I = \{i_1, i_2, \ldots, i_L \}$. Construct the $L\times L$ matrix $\bm{\mu}$ with rows $\mu_{i_1}^T,\mu_{i_2}^T,\ldots,\mu_{i_L}^T$, and require $\bm{\mu}$ to be full rank.
    \item Let $s_i$ be the signal seen by agent $i$. Elicit from each of these agents their expectation of the average posterior beliefs in the population (where $\tmu_j$ denotes agent $j$'s random posterior-belief vector), i.e., elicit for each $i \in I$:
    \begin{align*}
        \alpha_i = \mathbb{E}\left[\frac1n \sum_{j= 1}^n \tmu_j \bigg| s_i\right],
    \end{align*}
    using the same sample size $n$ as in Step~1.
    Denote $\bm{\alpha}= (\alpha_{i_1}^T, \alpha_{i_2}^T, \ldots, \alpha_{i_L}^T)^T$.
\item Solve
    \begin{align}
    \bar{\bm{\mu}} = \bm{\mu}^{-1} \bm{\alpha}\label{eqn:mainstep2}
    \end{align}
    Note that the right-hand side is elicited from agents directly and therefore $\bar{\bm \mu}$ can be recovered. The state $\omega$ is identified with the row of the recovered $\bar{\bm{\mu}}$ that is closest to the empirical mean $\hat{\mu}$. 
\end{enumerate}   
In practice the principal works with large but finite $n$; Assumptions~\ref{ass:1}--\ref{ass:fullrank} imply that the empirical quantities in Steps~1--3 approach their population counterparts as the sample grows.
  \end{minipage}
\end{framedbox}

Our second assumption ensures point identification. It requires that distinct states lead to different vectors of population-average beliefs, so that the limiting sample mean can be matched to a unique state. The requirement is weak but worth stating explicitly because identification can otherwise fail when expected beliefs coincide.

\begin{assumption}\label{ass:simple}
We assume that $P$ is such that for any two states $\omega$ and $\omega'$, $\bar{\mu}(\omega) \neq \bar{\mu}(\omega')$.
\end{assumption}
Note that while Assumption~\ref{ass:simple} is not directly stated in terms of the primitives of our model ($P$), it should be clear that this is a property that can be easily satisfied.%
\footnote{For example, in the case of finite conditionally i.i.d.\ signals, this property is generically satisfied in the belief space, in the sense of full measure.}
\par
We are now in a position to describe our main procedure, Population-Mean-Based Aggregation (PMBA) (Procedure~\ref{proc:pmba2} below), and state our main result for this setting.

Finally, we need to ensure that the system of equations that constitutes our procedure is well-posed. This is the role of Assumption~\ref{ass:fullrank}.
\begin{assumption}\label{ass:fullrank}
We assume that $P$ is such that for any agent $i$, the support of beliefs $\textrm{supp}(\tmu_i)$ contains $L$ distinct points, such that those $L$ beliefs, viewed as $L$-dimensional vectors, constitute a set of full rank. Alternately and equivalently, we require that the convex hull of the support has a nonempty interior relative to $\Delta(\Omega)$.
\end{assumption}

We should note that this condition is weak in at least two ways. First, in the standard genericity sense: a generically chosen set of feasible posteriors will be full rank (for details, see Subsection~\ref{sec:APMBA}). Second, and perhaps more interestingly, even in settings where there are fewer signals than states (e.g., agents can run only a unique binary-outcome experiment whose outcome determines their posterior), as long as the agents can acquire enough independent signals (e.g., conduct a sufficient number of independent experiments), the set of feasible posteriors is full rank. We formally state this claim below---the proof essentially repurposes a lemma of~\cite{fu2021full}.

\begin{restatable}{claim}{clfullrank}\label{cl:fullrank}
Suppose that signals are conditionally i.i.d.\ draws from a finite set $S$, and Assumption~\ref{ass:simple} is satisfied. If each agent gets $L-1$ independent draws from the same state-dependent distribution $P_{\omega }^{S}$, then Assumption~\ref{ass:fullrank} is satisfied with the compound signal space $%
S^{L-1}$.
\end{restatable}

\begin{restatable}{theorem}{thmpmba}\label{thm:pmba}
Suppose the common prior $P$ satisfies Assumptions~\ref{ass:1},~\ref{ass:simple}, and~\ref{ass:fullrank}. Then the PMBA procedure (Procedure~\ref{proc:pmba2}) recovers the true unknown state of nature $\omega$ in probability.
\end{restatable}

\section{Robust Aggregation: Approximate PMBA and Misspecification}\label{sec:robust}
While the PMBA procedure described in Section~\ref{sec:pmba} forms the heart of our approach, it is too idealized for real-world applications. In practice, exact equality will not hold: the empirical mean $\hat{\mu}$ will not exactly match any of the state-dependent population means $\bar{\mu}(\omega)$, agents' second-order beliefs may be noisy or biased, and the information structure may be misspecified. This section addresses these practical concerns by developing more robust versions of PMBA that can handle these complications.

We proceed in two steps. First, in Subsection~\ref{sec:APMBA}, we introduce the Approximate PMBA procedure, which relaxes the requirement of perfect alignment between observed and theoretical population averages. This procedure requires stronger assumptions on the information structure (specifically, conditional i.i.d.\ signals), but these assumptions help us build the machinery needed to handle misspecification. Second, in Subsection~\ref{sec:misspecified}, we show how the Approximate PMBA procedure can be extended to settings where agents' understanding of the information structure is misspecified---a common concern in practice where agents may have limited or incorrect knowledge about how others form beliefs.

\subsection{Approximate PMBA Procedure}\label{sec:APMBA}
We now introduce an alternative aggregation method, which we call the approximate PMBA procedure. This procedure serves two purposes. First, the approximate PMBA procedure does not rely on perfect alignment between the observed population average and one of the aggregated ones. Thus, it is more amenable to practical applications. Second, the procedure provides the foundation for handling misspecified information, which we address in Subsection~\ref{sec:misspecified}.

We introduce and analyze the procedure for the following class of information structures, which we specify in Assumption~\ref{ass:InfStrAPMBA}.
\par

\begin{assumption}\label{ass:InfStrAPMBA}
Let the joint prior probability measure $P$ satisfy the following properties:
    \begin{enumerate}
    \item Conditional i.i.d.\ signals: Conditional on the state of the world $\omega$, the signal $s_i$ of each agent $i \in N$ is an independent draw from the distribution $P_\omega^S$.
    \item The signal support is finite: $|S|=K\geq L$ 
    \item Imperfectly informed agents: For all $\omega \neq \omega'$, we have that $P_\omega^S$ and $P_{\omega'}^S$ are absolutely continuous with respect to each other yet not equal.
\end{enumerate}
\end{assumption}
\begin{framedbox}[htbp]
  \caption{Approximate Population-Mean-Based Aggregation}\label{proc:lipmba}
  \begin{minipage}{0.95\textwidth}
    \begin{enumerate}
\item Elicit, from each agent $i$, their belief $\mu_i \in \Delta^L$. Calculate the (sample) mean belief $\hat{\mu} \in \Delta^L$ as follows
    \begin{align*}
        \hat{\mu} = \frac1n \sum_{i=1}^n \mu_i.
    \end{align*}
    \item Select one state $\omega_1 \in \Omega$, and elicit from each agent $i \in N$ their interim expectation of the average posterior belief of state $\omega_1$ over all agents---denote this scalar as $\alpha_i$ for agent $i$ (note that unlike the vector $\alpha_i$ in Procedure~\ref{proc:pmba2}, here $\alpha_i$ is a scalar corresponding to a single coordinate):
\begin{align*}
        \alpha_i = \lim_{n \rightarrow \infty} E\left[\frac1n \sum_{j=1}^n \mu_j(\omega_1) \mid s_i\right].
    \end{align*}
    Equivalently, this is one coordinate of the full second-order expectation vector elicited in Procedure~\ref{proc:pmba2}; here we run the procedure coordinate-by-coordinate.
\item Partition the set of agents $N$ into $L$ groups such that the vectors
$\hat{\mu}_{N_1},\ldots,\hat{\mu}_{N_L}$ are linearly independent. For each group $N_k$, let $x_{iN_k}=\mathbf{1}[i \in N_k]$ denote the indicator that agent $i$ belongs to group $N_k$. Index agents within each group as $i=1,2,\dots$ and define group averages using the first $n$ agents in that group (equivalently, by the indicator-weighted expressions below):
\begin{align*}
        \hat{\mu}_{N_k}=\lim_{n \rightarrow \infty}\frac{\sum_{i=1}^n x_{iN_k}\mu_i}{\sum_{i=1}^n x_{iN_k}}.
    \end{align*}
    Further, calculate the empirical average of the $\alpha_i$'s in each group 
    \begin{align*}
        \hat{\alpha}_{N_k}=\lim_{n \rightarrow \infty}\frac{\sum_{i=1}^n x_{iN_k}\alpha_i}{\sum_{i=1}^n x_{iN_k}}.
    \end{align*}
    In practice the principal uses the same expressions with a finite sample of size $n$ (no limit).
    \item Solve
   
    \begin{align}
 &  \left( \begin{array}{c} \bar{\mu}_{\omega_1}(\omega_1)\\ \bar{\mu}_{\omega_1}(\omega_2) \\ \vdots\\ \bar{\mu}_{\omega_1}(\omega_L)\end{array}\right)  = \left( \begin{array}{c} \hat{\mu}_{N_1}\\ \hat{\mu}_{N_2} \\ \vdots\\ \hat{\mu}_{N_L}\end{array}\right)^{-1} 
    \left( \begin{array}{c} \hat{\alpha}_{N_1} \\ \hat{\alpha}_{N_2} \\ \vdots\\ \hat{\alpha}_{N_L}
    \end{array}     \right) \label{eqn:mainstep4}
    \end{align}
    where $\left(\begin{smallmatrix}\hat{\mu}_{N_1}\\ \vdots\\ \hat{\mu}_{N_L}\end{smallmatrix}\right)$ is the $L\times L$ matrix whose $k$-th row is the vector $\hat{\mu}_{N_k}$.
    The state $\omega$ is identified with the index $k$ for which the recovered $\bar{\mu}_{\omega_1}(\omega_k)$ is closest to the empirical mean $\hat{\mu}_{\omega_1}$ (the recovered object is the column vector of these $L$ components). (We index by $\omega_1$ only for expositional convenience; repeating this coordinate-wise recovers the full vector as in Section~\ref{sec:linear}.)
    \end{enumerate}
    \end{minipage}
    \end{framedbox}
The main difference between Procedure~\ref{proc:pmba2} and Procedure~\ref{proc:lipmba} is that aggregation in the former relies only on the first- and second-order reports of $L$ agents, while the latter makes use of the first- and second-order reports of all agents. The $L$ infinite groups $N_1$ through $N_L$ play the roles of the $L$ agents in the standard procedure (Procedure~\ref{proc:pmba2}).
Procedure~\ref{proc:pmba2} elicits the full second-order expectation vector from each queried agent, so no coordinate needs to be fixed there. By contrast, Procedure~\ref{proc:lipmba} is implemented coordinate-by-coordinate using group averages, so we fix one designated coordinate (here $\omega_1$) only for expositional convenience; repeating the same argument for each coordinate recovers the full vector.
Note that conditional independence of the information structure, as assumed in Assumption~\ref{ass:InfStrAPMBA}, ensures that all limit averages are well defined.

Before we introduce a further result, we require some additional notation. For a positive integer $l$, let $\Delta^l$ denote the unit simplex in $\mathbb{R}^l$, i.e., the set $\{x \in \mathbb{R}^l : x_j \geq 0, \sum_j x_j = 1\}$ (a $(l-1)$-dimensional set). For a given state space $\Omega$ and signal space $S$
with cardinalities $L$ and $K$, respectively, it follows from Assumption~\ref{ass:InfStrAPMBA} that the information structure is fully
described by the joint probability measure $p \in \Delta^{KL}$, where
$\Delta^{KL}$ denotes the unit simplex in $\mathbb{R}^{K \times L}$ (i.e., the set of joint distributions over signal--state pairs). Let $\mathcal{L}(\Delta)$ denote the Borel sigma algebra on the simplex and
$\lambda_\Delta$ the corresponding measure which, for $\Delta^l$, is derived from the Lebesgue measure on $\mathbb{R}^{l-1}$.

\par

We can now state the following result.

\begin{theorem}\label{theorem:GenericAggregation}
Consider the set of information structures $p \in \Delta^{KL}$ that satisfy Assumption~\ref{ass:InfStrAPMBA}. For $\lambda_{\Delta^{KL}}$-almost every information structure, \emph{Approximate Population-Mean-Based Aggregation} (Procedure~\ref{proc:lipmba}) almost surely recovers the true state of the world.
\end{theorem}

Thus, within a natural class of information structures (conditional i.i.d.\ signals), the Approximate PMBA procedure correctly identifies the state.\footnote{The assumption of a finite signal support is made to establish the result for measure-theoretic genericity. For a general signal space, the same result can be established for topological genericity.} The generic aspect of the result assures that there exists a partition of the agents such that the set of group averages forms a matrix of full rank. In implementation, one can search over data-driven partitions (e.g., bins in first-order belief space) and use any partition whose sample group-mean matrix is full rank.
\par
To provide intuition for the proof, note that the information structure $p \in \Delta^{KL}$ can be represented by a $K \times L$
matrix where the entry $p_{hj}$ corresponds to the probability of a joint realization of signal $h$
and state $j$. To each signal $s \in S$, there corresponds a posterior belief 
$q(s) \in \Delta^L$ over the state space $\Omega$. For a signal $s_h$, we have 
$q(s_h) = [q_1(s_h), \ldots, q_L(s_h)]$ where

\[
q_j(s_h) = \frac{p_{hj}}{\sum_{l=1}^L p_{hl}} = \frac{p_{hj}}{p_{s_h}}, \forall j = 1, \ldots, L.
\]

For information structure $p$, let $\overline{\mu}^p(\omega)\in \Delta^L$ denote the mean posterior belief vector in state $\omega$. 

For a given information structure $p$, consider the set of posterior beliefs $Q^p$:
\[
  Q^p \;=\; \{\, q(s)\in \Delta^L : s\in S \,\}.
\]
Denote the set of partitions of $Q^p$ with $L$ partition cells as $\mathcal{B}$. Consider a partition $B\in \mathcal{B}$ and a partition cell $b\in B$. The cell $b$ consists of a subset of posterior distributions, $b \subset Q^p$. For state $\omega$ and information structure $p$, let $F_\omega^{Q^p} \in \Delta(Q^p)$ denote the induced distribution over the posterior beliefs $Q^p$.

The proof of Theorem~\ref{theorem:GenericAggregation} relies on the following lemma; see the proof in Appendix~\ref{app:proofofThm1}.

\begin{lemma}\label{lemma:Genericity}
Under Assumption~\ref{ass:InfStrAPMBA}, for $\lambda_{\Delta^{KL}}$-almost every information structure, the following holds.
\begin{enumerate}
    \item In each state $\omega \in \Omega$, there exists a partition of $Q^p$ such that the corresponding set of mean posterior-belief vectors (one per partition cell) is linearly independent.
\item $\overline{\mu}_{\omega_j}^p(\omega) \neq \overline{\mu}_{\omega_j}^p(\omega')$ for all $j=1,\ldots,L$ and $\omega \neq \omega'$.
\end{enumerate}
\end{lemma}
Equipped with Lemma~\ref{lemma:Genericity}, we are ready to prove Theorem~\ref{theorem:GenericAggregation}.
\begin{proof}[Proof of Theorem~\ref{theorem:GenericAggregation}]
    Consider the empirical distribution $E^n$ of first-order beliefs in state $\omega$. Since, conditional on state $\omega$, first-order beliefs are identically and independently distributed by Assumption~\ref{ass:InfStrAPMBA}, it follows by the Glivenko-Cantelli theorem (e.g.,~\cite{vandervaartwellner1996}) that $E^n$ converges almost surely to $F_\omega^{Q^p}$ for every realized state $\omega \in \Omega$. By Lemma~\ref{lemma:Genericity}(2), the state-dependent average first-order belief distribution $\bar{\mu}(\omega)$ is different across $\omega$ in every dimension $1,\ldots,L$. By Lemma~\ref{lemma:Genericity}(1), there exists a partition of $Q^p$ under $F_\omega^{Q^p}$ that induces a full-rank matrix of group average beliefs. Hence, Procedure~\ref{proc:lipmba} satisfies perfect aggregation. Note that since signals are conditionally i.i.d.\ by Assumption~\ref{ass:InfStrAPMBA}, the strong law of large numbers implies that $\lim_{n\rightarrow \infty}\frac1n \sum_{i=1}^n \tilde{\mu}_i$ almost surely exists and equals $\bar{\mu}(\omega)$.
\end{proof}

Thus, Procedure~\ref{proc:lipmba} is applicable and correctly identifies the true state generically within a common class of information structures.

While the assumptions in this subsection are stronger than those in Section~\ref{sec:pmba} (in particular, we require conditional i.i.d.\ signals rather than the more general limited correlation structure), they provide the necessary structure to extend the procedure to settings with misspecified beliefs, which we turn to next.
\par

\subsection{Aggregation in Misspecified Information Settings}\label{sec:misspecified}
In practice, agents may have limited and individually incorrect knowledge about the aggregate information structure: in particular, a standard concern is that while agents may have meaningful information that shapes their \emph{personal} beliefs, they nevertheless may be less informed about the belief of others, or even summary statistics of the beliefs of others like the average belief as elicited by our procedure. Agents might even be subject to a systemic bias in terms of their perception of the summary statistics, making matters more difficult. 

We now analyze how Procedure~\ref{proc:lipmba} can be applied to settings where agents' understanding of the information structure is misspecified. Instead of assuming common knowledge of $P$ among the agents, we assume that the agents' information about $P$ is misspecified. As in Procedure~\ref{proc:lipmba}, fix one designated coordinate---for exposition, state $\omega_1$---and suppose agent $i$ holds misspecified beliefs about the corresponding state-dependent population averages, namely $(\alpha_{i}^{\omega_1},\alpha^{\omega_2}_{i},\ldots,\alpha_{i}^{\omega_L}) \in [0,1]^L$. Note that each agent understands how to interpret their own signal correctly, i.e., the misspecification is purely about knowledge of the population averages. We make the following assumption on the relation between agents' beliefs and the true state-dependent population means; as above, the same argument can be repeated coordinate-by-coordinate to recover the full vector.

\begin{assumption}\label{ass:misspecified}
For each agent $i$, the expected population average belief of state $\omega_1$ given his first-order belief $\mu_i \in \Delta^L$ is equal to
\begin{equation*}
    \alpha_i=\sum_{\omega \in \Omega}\mu_i(\omega) \times \alpha_i^\omega
\end{equation*}
where for each state $\omega \in \Omega $ we have that%
\begin{equation}\label{eq1:2ndordermisspec}
\alpha _{i}^{\omega }=\bar{\mu}_{\omega_1}(\omega )+\zeta _{i}^{\omega }.
\end{equation}%
We assume that conditional on $\omega $ the $%
\zeta _{i}^{\omega }$'s are independent across agents, that $\zeta_i^{\omega}$ is independent of $\mu_i$ for each agent $i$, and that for each $\omega \in \Omega$, the $\zeta_i^{\omega}$'s have a common mean, $E[\zeta _{i}^{\omega }]=E[\zeta ^{\omega }]$ for all $%
i $.

\end{assumption}

Let us briefly discuss Assumption~\ref{ass:misspecified}. Agents form their expected average population belief of state $\omega_1$ based on their misspecified beliefs of the state-dependent averages $(\omega_1,\ldots,\omega_L)$ weighted by their first-order beliefs. The inherent assumption here is that, conditional on the realized state, the population average belief is independent of agent $i$'s first-order belief $\mu_i$. The conditional-independence assumption on the noise terms $\zeta_i^\omega$ across agents is necessary for a strong law of large numbers to hold, i.e., so that the weighted average of noise terms converges to a degenerate random variable almost surely. The assumption that, conditional on state $\omega$, the noise term has the same expectation across agents is a technical condition that we require for the procedure to accurately aggregate the state.

Importantly, note that we do not assume that $\mathbb{E}[\zeta^\omega]=0$: agents' beliefs about the population averages may be biased even on average (as has indeed been observed in practice), and these biases may be different in different states.

We also impose the following condition on the information structure that generates the agents' signals.
\begin{assumption}\label{ass:1Mis}
The private signals $(s_i)_{i\in N}$ are independent and identically distributed across agents, conditional on the realized state $\omega$.
\end{assumption}

We can now state our aggregation result in the misspecified information setting.

\begin{restatable}{theorem}{misspecifiedinfo}\label{thm:misspecifiedknowledge}
Consider an information structure that satisfies Assumption~\ref{ass:1Mis} and for which the \emph{Approximate Population-Mean-Based Aggregation} (Procedure~\ref{proc:lipmba}) almost surely recovers the true state under correctly specified beliefs. Suppose that the agents' knowledge regarding the information structure satisfies Assumption~\ref{ass:misspecified}. Suppose further that for all $\omega \neq \omega'$ we have
\begin{equation}\label{eqn:misspecifiedinfocondn}
\left\vert \mathbb{E}[\zeta ^{\omega}]\right\vert +\left\vert \mathbb{E}[\zeta ^{\omega'}]\right\vert <
\left\vert \bar{\mu}_{\omega_1}\left(\omega\right) -\bar{\mu}_{\omega_1}\left(\omega'\right) \right \vert.
\end{equation}
Then, \emph{Approximate Population-Mean-Based Aggregation} (Procedure~\ref{proc:lipmba}) recovers the true state almost surely under misspecified beliefs.
\end{restatable}
Theorem~\ref{thm:misspecifiedknowledge} establishes that the conceptual idea of second-order elicitation successfully overcomes agents' limited understanding of the private belief distributions of their peers.

Thus, moving to a more realistic assumption of an incomplete and individually incorrect understanding of the information structure necessitates the elicitation of the expected population averages from large groups: we are relying on the crowd for both first- and second-order mean-belief statistics. The main issue is that individual agents' understanding of the population averages is incorrect. However, pooling the first- and second-order reports into infinitely sized groups allows us to overcome this issue. By the law of large numbers we can appropriately average these individual agents' reports over the $L$ groups, $N_1$ through $N_L$. The average first-order belief, and average report of population average, of agents in each group, will serve as the analogs of the $L$ individuals in the baseline PMBA procedure.

The basic intuition is that even though individual agents are misinformed about the distribution of their fellow agents' beliefs, on average these errors are not ``large'' (even though individual agents' errors can be) relative to the true differences between population averages across states (recall that condition~\eqref{eqn:misspecifiedinfocondn} in Theorem~\ref{thm:misspecifiedknowledge} makes this possible). This ensures that the recovered $\bar{\mu}(\omega)$ that is closest to the realized average population distribution $\hat{\mu}$ is indeed the one corresponding to the true state.

\subsubsection{The Proof of Theorem~\ref{thm:misspecifiedknowledge}}
First note that for Procedure~\ref{proc:lipmba} to almost surely recover the true state absent misspecification, for each realized state $\omega \in \Omega$ there exists a partition of the empirical belief distribution into $L$ groups such that the average beliefs of the groups almost surely form a full-rank matrix. Formally, there exists a partition $\{N_1,\ldots,N_L\}$ such that the matrix
\[
\begin{pmatrix}
\hat{\mu}_{N_1} \\
\vdots \\
\hat{\mu}_{N_L}
\end{pmatrix}
\]
almost surely is of full rank. Note that agents are assigned to partition cells based only on their realized belief $\mu_i$.
To prove Theorem~\ref{thm:misspecifiedknowledge} we first need some additional lemmata.
Consider a partition $\{N_1,\ldots,N_L\}$ that belongs to the support of partitions generating a full-rank set of group average beliefs, with each group receiving positive asymptotic mass. For each $n\in\mathbb{N}$ and group $N_k$ define
\begin{equation}\label{eq:2ndordermis}
    \hat{\alpha}_{N_k}^n=\frac{1}{\sum_{i\le n}x_{iN_k}}\sum_{i\leq n} x_{iN_k}\alpha_i, \ \text{and} \ \hat{\alpha}_{N_k}=\lim_{n \rightarrow \infty} \hat{\alpha}_{N_k}^n.
\end{equation}
Analogously, define
\[
\hat{\mu}_{N_k}^n:=\frac{1}{\sum_{i\le n}x_{iN_k}}\sum_{i\leq n} x_{iN_k}\mu_i,
\qquad
\hat{\mu}_{N_k}:=\lim_{n\to\infty}\hat{\mu}_{N_k}^n.
\]
The proof relies crucially on the following two auxiliary results.
\begin{lemma}\label{lem:misspec2}
If the information structure satisfies Assumption~\ref{ass:1Mis}, then for each group $N_k$, $\hat{\mu}_{N_k}^n$ converges almost surely to a deterministic vector, denoted $\mathbb{E}[\hat{\mu} _{N_k}]$.
\end{lemma}
\begin{lemma}\label{lem:misspec}
If the information structure satisfies Assumption~\ref{ass:1Mis} and the agents' knowledge regarding the information structure satisfies Assumption~\ref{ass:misspecified}, then for each group $N_k$, $\hat{\alpha}_{N_k}^n$ almost surely converges to the deterministic scalar
\begin{align*}
\hat{\alpha}_{N_k}=\sum_{\omega \in \Omega} \mathbb{E}[\hat{\mu}_{N_k}(\omega)]\check{\mu}_{\omega_1}(\omega)
\ \ \text{where }&\check{\mu}_{\omega_1}(\omega) =\bar{\mu}_{\omega_1}(\omega)+\mathbb{E}[\zeta^{\omega}]
\end{align*}
\end{lemma}
Please see below for the proof of Theorem~\ref{thm:misspecifiedknowledge}.
\begin{proof}
By Lemmas~\ref{lem:misspec2} and~\ref{lem:misspec}, and the maintained assumptions, conditional on the realized state, Procedure~\ref{proc:lipmba} yields the following uniquely solvable system of deterministic linear equations.
\begin{align}
 \underbrace{\left( \begin{array}{c} \mathbb{E}[\hat{\mu}_{N_1}]\\ \mathbb{E}[\hat{\mu}_{N_2}] \\ \vdots\\ \mathbb{E}[\hat{\mu}_{N_L}]\end{array}\right)}_{\text{M}}  \left( \begin{array}{c} \check{\mu}_{\omega_1}(\omega_1)\\ \check{\mu}_{\omega_1}(\omega_2) \\ \vdots\\ \check{\mu}_{\omega_1}(\omega_L)\end{array}\right)  = 
    \left( \begin{array}{c} \hat{\alpha}_{N_1} \\ \hat{\alpha}_{N_2} \\ \vdots\\ \hat{\alpha}_{N_L}
    \end{array}     \right) \label{eqn:mainstep4misspec}
    \end{align}
Here $M$ is the $L\times L$ matrix whose $k$-th row is the group mean vector $\mathbb{E}[\hat{\mu}_{N_k}]$; since Procedure~\ref{proc:lipmba} satisfies perfect aggregation absent misspecification, the matrix $M$ is of full rank and hence invertible. This system reveals the following noisy state-dependent average beliefs of state $\omega_1$, $(\check{\mu}_{\omega_1}(\omega_1), \check{\mu}_{\omega_1}(\omega_2), \ldots, \check{\mu}_{\omega_1}(\omega_L))$. For approximate aggregation we require that the recovered mean for the true state is closer to the true state-contingent mean than that for any other state; i.e., for each state $\omega$ we have that
\begin{equation*}
    \left\vert \check{\mu}_{\omega_1}(\omega)-\bar{\mu}_{\omega_1}(\omega )\right\vert <\left\vert 
\check{\mu}_{\omega_1}(\omega')-\bar{\mu}_{\omega_1}(\omega )\right\vert 
\end{equation*}
for all $\omega' \neq \omega$.
By Lemma~\ref{lem:misspec}, this condition is equivalent to
\begin{equation*}
   \left\vert \bar{\mu}_{\omega_1}(\omega)+\mathbb{E}[\zeta^{\omega}]-\bar{\mu}_{\omega_1}(\omega )\right\vert <\left\vert 
\bar{\mu}_{\omega_1}(\omega')+\mathbb{E}[\zeta^{\omega'}]-\bar{\mu}_{\omega_1}(\omega )\right\vert 
\end{equation*}
which in turn is equivalent to
\begin{equation}\label{eqn:ineqzeta1intermediate} 
   \left\vert \mathbb{E}[\zeta^{\omega}]\right\vert <\left\vert 
\bar{\mu}_{\omega_1}(\omega')+\mathbb{E}[\zeta^{\omega'}]-\bar{\mu}_{\omega_1}(\omega )\right\vert
\end{equation}
Since 
\begin{equation}\label{eqn:ineqzeta1triangle} 
   \left\vert 
\bar{\mu}_{\omega_1}(\omega')-\bar{\mu}_{\omega_1}(\omega )\right\vert - \left \vert \mathbb{E}[\zeta^{\omega'}] \right \vert \leq \left\vert 
\bar{\mu}_{\omega_1}(\omega')+\mathbb{E}[\zeta^{\omega'}]-\bar{\mu}_{\omega_1}(\omega )\right\vert
\end{equation}
it follows that the following is a sufficient condition for~\eqref{eqn:ineqzeta1intermediate}
\begin{equation*}
    \left\vert \mathbb{E}[\zeta^{\omega}]\right\vert<\left\vert 
\bar{\mu}_{\omega_1}(\omega')-\bar{\mu}_{\omega_1}(\omega )\right\vert - \left \vert \mathbb{E}[\zeta^{\omega'}] \right \vert
\end{equation*}
which is equivalent to 
\begin{equation*}
    \left\vert \mathbb{E}[\zeta^{\omega}]\right\vert + \left \vert \mathbb{E}[\zeta^{\omega'}] \right \vert<\left\vert 
\bar{\mu}_{\omega_1}(\omega')-\bar{\mu}_{\omega_1}(\omega )\right\vert 
\end{equation*}
which coincides with Equation (\ref{eqn:misspecifiedinfocondn}) concluding the proof.
\end{proof}

\section{Linear Regression Interpretation}\label{sec:linear}

We now provide a linear regression interpretation of PMBA\@. This provides three key benefits that make the procedure more practical and robust. First, this reinterpretation immediately gives us access to a century of statistical tools for inference (including confidence intervals), diagnostics, and robustness checks. Second, we develop geometric intuition that helps explain why the procedure works---the regression coefficients encode the unknown state-dependent population means we are trying to recover. Finally, we derive finite-sample error bounds that tell us exactly how many agents we need for reliable aggregation, making the procedure actionable for real-world implementation.

This linear regression perspective is particularly valuable because it transforms PMBA from a specialized matrix inversion procedure into a familiar statistical framework that researchers and practitioners already understand.

\subsection{From PMBA to Linear Regression}

To understand how PMBA relates to linear regression, let us start with the data we collect from each agent. Consider a state space $\Omega=\{\omega_1,\dots,\omega_L\}$ with $L\ge 2$ states. \textbf{Notation.} In this section, we write $p_i$ for agent $i$'s first-order belief vector (denoted $\mu_i$ in Sections~\ref{sec:pmba}--\ref{sec:robust}) and $\alpha_i$ for a \emph{scalar} second-order belief about a single coordinate (whereas $\alpha_i$ is a vector in Procedure~\ref{proc:pmba2}). From each agent $i$, we elicit two pieces of information:

\begin{enumerate}
  \item \textbf{First-order beliefs}: A probability vector $p_i=(p_{i1},\dots,p_{iL})$ where $p_{ik}$ is agent $i$'s belief that state $\omega_k$ is true (these sum to 1).
  \item \textbf{Second-order beliefs}: A scalar $\alpha_i$ representing agent $i$'s expectation of the population average belief for a designated coordinate (e.g., state $\omega_1$), i.e., what agent $i$ thinks the average reported probability on that coordinate will be.
\end{enumerate}

The key insight is that agent $i$'s second-order belief $\alpha_i$ should be a weighted average of the population means in each state, where the weights are her own beliefs. Specifically, if $\bar\mu_k$ represents the average belief about state $\omega_1$ in the population when the true state is $\omega_k$, then agent $i$'s forecast should be:

\begin{equation}
  \alpha_i \;=\; \sum_{k=1}^{L} p_{ik}\,\bar\mu_k + \varepsilon_i,
  \qquad \mathbb E[\varepsilon_i\mid p_i]=0,
  \label{eq:struct_LR_multi}
\end{equation}

where $\varepsilon_i$ captures individual reporting errors, cognitive limitations, or rounding mistakes. 

\begin{remark}
Here, $\alpha_i$ is a scalar representing agent $i$'s expectation of the population average belief about a specific state (say, state $\omega_1$). In the original PMBA procedure (Procedure~\ref{proc:pmba2}), we elicit the full vector of expectations $\alpha_i$ (the expectation of average posterior beliefs); the same coordinate-wise interpretation also applies to Procedure~\ref{proc:lipmba}. For the regression interpretation, we focus on a single coordinate of this vector, making $\alpha_i$ a scalar. This corresponds to asking agent $i$: ``What do you think the average belief in the population will be about state $\omega_1$?'' The agent responds with a single number (e.g., ``I think the average person will assign 60\% probability to state $\omega_1$''). This equation says that agent $i$'s expectation should be a weighted combination of what the population average would be in each possible state, weighted by how likely she thinks each state is.

To recover the full vector of population means $\bar{\bm{\mu}}$, the regression in equation~\eqref{eq:OLS_L} (below) is repeated for each dimension $d \in \{1, \dots, L\}$, where the dependent variable is the agent's expectation of the population mean for state $\omega_d$.
\end{remark}

\subsection{The Regression Setup}

Since belief vectors must sum to 1 ($\sum_k p_{ik}=1$), we have a constraint that makes one coefficient redundant. To convert this into a standard linear regression, we can treat one state (say, $\omega_L$) as the baseline and express all other coefficients relative to it. This gives us:

\begin{equation}
  \alpha_i = \beta_0 + \sum_{k=1}^{L-1} \beta_k\,p_{ik} + \varepsilon_i,
  \quad
  \beta_0:=\bar\mu_L,\quad \beta_k:=\bar\mu_k-\bar\mu_L.
  \label{eq:OLS_L}
\end{equation}

Here is what this means: $\beta_0$ is the population average belief when the baseline state $\omega_L$ is true, and $\beta_k$ captures how much higher (or lower) the population average belief is when state $\omega_k$ is true compared to the baseline state. This is now a standard linear regression where we can estimate the unknown parameters $\beta_0, \beta_1, \ldots, \beta_{L-1}$ using ordinary least squares.
\subsection{Geometric Intuition}

The linear regression has a beautiful geometric interpretation that helps explain why PMBA works. Think of each agent's beliefs as a point in a $(L-1)$-dimensional simplex (since probabilities sum to 1). The regression equation $p_i\mapsto\alpha_i$ maps this simplex onto a hyperplane in the space of second-order beliefs.

The key insight is that the slope of this hyperplane along each coordinate direction tells us exactly what we want to know: how much the population average belief differs between states. Specifically, the slope along the $k$-th direction equals $\beta_k = \bar\mu_k - \bar\mu_L$, which is exactly the difference in population averages between state $\omega_k$ and the baseline state $\omega_L$. This is also why PMBA extends naturally beyond binary states, unlike SP-style comparisons that rely on a binary ordering.

Figure~\ref{fig:hyperplane} illustrates this for $L=3$ states. The triangle represents all possible belief vectors, the colored dots show the true population means in each state, and the arrows show the slope directions that the regression estimates. Once we know these slopes, we can reconstruct the full set of population means and identify the true state.

\begin{figure}[H]
    \centering
    \begin{tikzpicture}[scale=3.8]
      \coordinate (A) at (0,0);
      \coordinate (B) at (1,0);
      \coordinate (C) at (0,1);
      \draw[very thick] (A)--(B)--(C)--cycle;
  
      \node[below left]  at (A) {$\omega_2$};
      \node[below]       at (B) {$\omega_1$};
      \node[left]        at (C) {$\omega_3$};
  
      \foreach \i in {0.2,0.4,0.6,0.8}{
        \draw[thin,opacity=0.25] (\i,0)--(0,\i);
        \draw[thin,opacity=0.25] (1-\i,\i)--(\i,0);
        \draw[thin,opacity=0.25] (0,1-\i)--(1-\i,0);
      }
  
      \fill[black] (0.333,0.333) circle (0.5pt) node[above right] {$\hat\mu$};
  
      \fill[red]            (0.70,0.15) circle (0.6pt) node[below right=-2pt] {$\bar\mu_1$};
      \fill[green!60!black] (0.10,0.15) circle (0.6pt) node[below left=-2pt] {$\bar\mu_2$};
      \fill[blue]           (0.20,0.60) circle (0.6pt) node[left=-2pt] {$\bar\mu_3$};
  
      \draw[->,thick,red]            (0.333,0.333)--(0.533,0.283);
      \draw[->,thick,green!60!black] (0.333,0.333)--(0.233,0.283);
      \draw[->,thick,blue]           (0.333,0.333)--(0.233,0.533);
    \end{tikzpicture}
    \caption{Geometric view for $L=3$: the belief simplex (triangle), true
    state-dependent population means (coloured dots), the empirical mean $\hat\mu$
    (black dot), and slope directions $\beta_k$ (arrows).}\label{fig:hyperplane}
  \end{figure}

\subsection{Statistical Assumptions}

To make the regression approach rigorous, we need some statistical assumptions. Let's stack our data into an $n\times L$ design matrix $X=[\mathbf 1\;p_{i1}\;\dots\;p_{i,L-1}]$ and response vector $\alpha=(\alpha_1,\dots,\alpha_n)^{\top}$.

\begin{assumption}[Independent Sub-Gaussian disturbances]\label{ass:linear:subG}
The reporting errors $\varepsilon_i$ are mutually independent across agents and each $\varepsilon_i$ is $\sigma^2$-sub-Gaussian conditioned on $p_i$, meaning the error distribution has light tails that decay at least as fast as a Gaussian with variance $\sigma^2$.
\end{assumption}

This assumption is quite reasonable in practice. While agents' first-order beliefs may be correlated (because they share information), their reporting errors are typically independent idiosyncratic factors like:
\begin{itemize}
\item Individual cognitive limitations in calculating expected population averages
\item Rounding errors when reporting probabilities
\item Personal biases in interpreting survey questions
\item Random mistakes
\end{itemize}

The sub-Gaussian condition simply means that extreme errors are rare, which is natural when dealing with probability reports bounded between 0 and 1.

\begin{assumption}[Well-conditioned Gram matrix]\label{ass:linear:LCLN}
The design matrix is well-conditioned: there exists $\underline\lambda>0$ such that
\[
\lambda_{\min}\!\left(n^{-1}X^{\top}X\right)\ge \underline\lambda.
\]
\end{assumption}

This assumption requires that agents have sufficiently diverse beliefs. If all agents had identical beliefs, we could not distinguish between different states because everyone would report the same thing. More generally, if agents' belief vectors are too similar (nearly collinear), the regression becomes unstable and we cannot reliably estimate the coefficients.

In practice, this assumption is satisfied when agents have access to different information or interpret signals differently, which is exactly what we want in an information aggregation setting. The constant $\underline\lambda$ measures how ``spread out'' the agents' beliefs are---larger values mean more diverse beliefs and better estimation.

\begin{remark}[Minimal-information inversion]\label{rem:minimal}
Procedure~\ref{proc:pmba2} inverts a single $L\times L$ matrix built
from $L$ \emph{group} averages. Algebraically this is equivalent to
OLS with exactly $L$ observations whose regressor matrix happens to be
non-singular (Assumption~\ref{ass:linear:LCLN}). The regression
view therefore subsumes the procedure while making standard errors and
model diagnostics immediately available.
\end{remark}

\subsection{Finite-Sample Performance}

With these assumptions in place, we can derive concrete performance guarantees. Using the standard OLS estimator $\widehat{\beta}=(X^{\top}X)^{-1}X^{\top}\alpha$, we get:

\begin{lemma}[Non-asymptotic error bound]\label{lem:linear:finite}
For any confidence level $\delta\in(0,1)$, the estimation error is bounded by:
\begin{equation}
  \Pr\Bigl(\|\widehat{\beta}-\beta\|_2 \le
          \frac{\sqrt{L}\sigma}{\underline\lambda}\sqrt{\tfrac{2\log(2L/\delta)}{n}}
  \Bigr) \ge 1-\delta.
  \label{eq:error_bound}
\end{equation}
\end{lemma}

This bound tells us exactly how precise our estimates are with finite data. The error decreases at the standard $1/\sqrt{n}$ rate as we collect more agents, but also depends on three key factors:

\begin{itemize}
\item \textbf{Sample size} $n$: More agents give better estimates
\item \textbf{Noise level} $\sigma$: Less noisy reports lead to more precise estimates
\item \textbf{Belief diversity} $\underline\lambda$: More diverse beliefs improve estimation quality
\item \textbf{Number of states} $L$: More states require more agents, with a $\sqrt{L\log L}$ dependence in this bound
\end{itemize}

The dependence on $L$ remains moderate (order $\sqrt{L\log L}$), so the procedure still scales reasonably with the number of states.

\begin{proof}
    Write $Z:=n^{-1}X^{\top}\varepsilon$ so that
    $\widehat{\beta}-\beta = (X^{\top}X)^{-1} X^{\top}\varepsilon = n\,(X^{\top}X)^{-1} Z$.
    By definition of $\underline\lambda$, the operator norm satisfies
    $\|(X^{\top}X)^{-1}\|_2\le(n\,\underline\lambda)^{-1}$. Hence
    \[\|\widehat{\beta}-\beta\|_2\;\le\;\underline\lambda^{-1}\,\|Z\|_2.\]
    Each coordinate of $Z$ is a sum of $n$ independent $\sigma$-sub-Gaussian
    random variables scaled by $1/n$. By Assumption~\ref{ass:linear:subG} and standard Hoeffding inequalities (see, e.g.,~\cite[Ch.~2]{vershynin2018highdim}),
    \[\Pr(|Z_j|\!>\!t)\le2\exp(-n t^2/2\sigma^2).\] 
    A union bound over the $L$
    coordinates and setting $t=\sigma\sqrt{\tfrac{2\log(2L/\delta)}{n}}$ gives $|Z_j|\le t$ for all $j$ with probability $1-\delta$, which implies
    \[\|Z\|_2\le\sqrt{L}\sigma\sqrt{2\log(2L/\delta)/n},\; \text{w.p.} \;1-\delta.\]
    Plugging into the display above proves~\eqref{eq:error_bound}.
    \end{proof}
Equation~\eqref{eq:error_bound} shows that precision improves at the
canonical $n^{-1/2}$ rate but deteriorates as $\sqrt{L\log L}$ in the
number of states (due to the $\sqrt{L}$ factor from the vector norm and the logarithmic term from the union bound). The term $\underline\lambda$ captures how
``spread out'' the posterior vectors are.

\subsection{State classification rule}
For each coordinate $d\in\{1,\ldots,L\}$, run regression~\eqref{eq:OLS_L} with dependent variable equal to the elicited second-order report for coordinate $d$, and denote the resulting coefficients by $\widehat{\beta}^{(d)}$. Define
\[
\widehat{\bar\mu}^{(d)}_k :=\widehat{\beta}^{(d)}_0+\widehat{\beta}^{(d)}_k
\quad (k=1,\ldots,L-1),\qquad
\widehat{\bar\mu}^{(d)}_L:=\widehat{\beta}^{(d)}_0.
\]
Stack across coordinates to form fitted state-mean vectors
\[
\widehat{\bar\mu}_k:=\big(\widehat{\bar\mu}^{(1)}_k,\ldots,\widehat{\bar\mu}^{(L)}_k\big)\in\mathbb{R}^L,
\]
so that $\widehat{\bar\mu}_L$ corresponds to the omitted baseline state $\omega_L$ in each coordinate regression.
Let $\hat\mu:=\tfrac1n\sum_{i=1}^{n}p_i$ be the empirical mean of
first-order beliefs. Analogous to Procedure~\ref{proc:pmba2}, we select as true state:
\begin{equation}
  \widehat\omega = \arg\min_{k\le L}
      \|\widehat{\bar\mu}_k-\hat\mu\|_2.
  \label{eq:state_rule}
\end{equation}
Here $\widehat{\bar\mu}_k$ is the fitted \emph{vector} of state-contingent means for state $k$ after stacking all coordinate-wise regressions.
For the rule to succeed, the true state means must be separated.

\begin{assumption}[Separation of state means]\label{ass:linear:sep}
$\displaystyle
  \Delta := \min_{k\neq\ell}\|\bar\mu_k-\bar\mu_\ell\|_2 > 0$.
\end{assumption}

\begin{assumption}[First-order sampling for finite-sample analysis]\label{ass:linear:fo}
Conditional on the realized state $\omega$, the first-order belief vectors $\{p_i\}_{i=1}^n$ are independent across agents.
\end{assumption}

\begin{proposition}[Finite-$n$ identification guarantee]\label{prop:linear:finite}
Under
Assumptions~\ref{ass:linear:subG}--\ref{ass:linear:fo}, define
\[
r_n^\beta:=\frac{\sqrt{L}\sigma}{\underline\lambda}\sqrt{\frac{2\log(4L^2/\delta)}{n}},
\qquad
r_n^\mu:=\sqrt{\frac{L\log(4L/\delta)}{2n}},
\qquad
r_n:=\max\{r_n^\beta,r_n^\mu\}.
\]
If
\begin{equation}
  n>\max\!\left\{
  \frac{72L^2\,\sigma^2\log(4L^2/\delta)}{\underline\lambda^2\,\Delta^2},
  \frac{18L^2\log(4L/\delta)}{\Delta^2}
  \right\},
  \label{eq:n_threshold}
\end{equation}
then $\Pr(\widehat\omega = \omega) \ge 1-\delta$.
\end{proposition}

\begin{proof}
    The proof proceeds in several steps. First, we establish concentration of the coordinate-wise regression coefficients. Second, we show that this implies concentration of the fitted state-contingent means. Third, we derive concentration of the empirical first-order mean. Finally, we combine these bounds with the separation condition.

    \textbf{Step 1: Coefficient estimation error.}
    By Lemma~\ref{lem:linear:finite}, for each coordinate $d\le L$ and any failure probability $\eta\in(0,1)$,
    \[
    \Pr\!\left(\|\widehat{\beta}^{(d)}-\beta^{(d)}\|_2
    \le
    \frac{\sqrt{L}\sigma}{\underline\lambda}\sqrt{\frac{2\log(2L/\eta)}{n}}
    \right)\ge 1-\eta.
    \]
    Set $\eta=\delta/(2L)$ and apply a union bound over $d=1,\ldots,L$. Then with probability at least $1-\delta/2$,
    \[
    \max_{d\le L}\|\widehat{\beta}^{(d)}-\beta^{(d)}\|_2\le r_n^\beta\le r_n.
    \]

    \textbf{Step 2: Population mean estimation error.}
    Fix any state $k\le L$ and coordinate $d\le L$. Under Step 1,
    \[
    |\widehat{\bar\mu}^{(d)}_k-\bar\mu^{(d)}_k|
    \le 2r_n
    \]
    (the baseline case $k=L$ has error at most $r_n$). Hence,
    \[
    \|\widehat{\bar\mu}_k-\bar\mu_k\|_2
    \le \sqrt{L}\,\max_{d\le L}|\widehat{\bar\mu}^{(d)}_k-\bar\mu^{(d)}_k|
    \le 2\sqrt{L}\,r_n
    \]
    for all $k\le L$.

    \textbf{Step 3: Concentration of first-order mean.}
    Under Assumption~\ref{ass:linear:fo}, for each coordinate $d\le L$, conditional on the realized state $\omega$, the variables $\{p_{id}\}_{i=1}^n$ are independent and lie in $[0,1]$. Hoeffding's inequality implies
    \[
    \Pr\!\left(\left|\hat\mu_d-\bar\mu_{\omega,d}\right|>t\ \middle|\ \omega\right)
    \le 2e^{-2nt^2}.
    \]
    Taking $t=\sqrt{\log(4L/\delta)/(2n)}$ and a union bound over $d=1,\ldots,L$ yields, with probability at least $1-\delta/2$,
    \[
    \|\hat\mu-\bar\mu_\omega\|_2
    \le
    \sqrt{L}\,\max_{d\le L}\left|\hat\mu_d-\bar\mu_{\omega,d}\right|
    \le
    \sqrt{\frac{L\log(4L/\delta)}{2n}}
    =r_n^\mu\le r_n.
    \]

    \textbf{Step 4: Distance to true state.}
    The classification rule selects $\widehat\omega = \arg\min_{k\le L} \|\widehat{\bar\mu}_k-\hat\mu\|_2$.
    Let $\omega$ be the true state.
    For the true state $\omega$, we have:
    \[
      \|\widehat{\bar\mu}_\omega-\hat\mu\|_2 \,\le\,
      \|\widehat{\bar\mu}_\omega-\bar\mu_\omega\|_2+\|\hat\mu-\bar\mu_\omega\|_2
      \,\le\, 2\sqrt{L}\,r_n + r_n \le 3\sqrt{L}\,r_n.
    \]
    
    \textbf{Step 5: Distance to other states.}
    For any other state $k\neq\omega$, by the triangle inequality and Assumption~\ref{ass:linear:sep}:
    \[
      \|\widehat{\bar\mu}_k-\hat\mu\|_2 \,\ge\,
      \|\bar\mu_k-\bar\mu_\omega\|_2-\|\widehat{\bar\mu}_k-\bar\mu_k\|_2-
      \|\hat\mu-\bar\mu_\omega\|_2 \,\ge\, \Delta-2\sqrt{L}\,r_n-r_n \ge \Delta-3\sqrt{L}\,r_n.
    \]
    
    \textbf{Step 6: Success condition.}
    Correct classification follows if
    \[
    3\sqrt{L}\,r_n < \Delta-3\sqrt{L}\,r_n,
    \quad\text{equivalently}\quad
    6\sqrt{L}\,r_n<\Delta.
    \]
    Since $r_n=\max\{r_n^\beta,r_n^\mu\}$, it is enough that both
    \[
    6\sqrt{L}\,r_n^\beta<\Delta
    \qquad\text{and}\qquad
    6\sqrt{L}\,r_n^\mu<\Delta,
    \]
    which are exactly the two terms in~\eqref{eq:n_threshold}. The coefficient-concentration and first-order-concentration events each hold with probability at least $1-\delta/2$; by a union bound their intersection has probability at least $1-\delta$. On that intersection, the inequalities above imply $\widehat\omega=\omega$.
    \end{proof}
The multi-state aggregation procedure can be
read as nothing more than an over-identified linear regression whose
coefficients---gaps in population means---underpin identification. Doing so pays two dividends: it demystifies the matrix inversion in Procedure~\ref{proc:pmba2} and it imports a century's worth of inferential tools into the robust-aggregation setting.

\section{Experimental Evidence}\label{sec:experiment}

To test the theoretical predictions of our PMBA procedure, we conducted an online experiment involving belief elicitation about NFT prices. The experiment was designed to evaluate whether PMBA outperforms alternative aggregation methods in a real-world setting with genuine uncertainty.

NFTs (Non-Fungible Tokens) are unique digital assets that represent ownership of specific digital content, such as artwork, music, or collectibles, recorded on a blockchain. Unlike cryptocurrencies, each NFT is unique and cannot be exchanged on a one-to-one basis. The NFT market is diverse, with prices ranging from negligible amounts to millions of dollars, making price prediction inherently challenging.

This prediction task is particularly well-suited for testing our aggregation methodology for several reasons. First, we have access to ground-truth data---actual market prices at the time of the experiment---allowing us to evaluate the accuracy of different aggregation methods objectively. Second, while market prices are publicly available, individual participants likely have limited and heterogeneous access to recent price data, creating the information asymmetry that our robust aggregation procedure is designed to handle. Third, NFT prices in our sample span all four ranges with substantial cross-item dispersion, so misclassification is informative and allows meaningful comparison across aggregation methods.

\subsection{Experimental Design}

\subsubsection{Participants and Procedure}
Participants from the National University of Singapore (NUS) were recruited through the NUS Behavioral Lab ORSEE recruitment system. A total of 199 eligible responses were retained for analysis after excluding incomplete entries, duplicates, and responses for which identity information could not be verified. Each participant received a fixed compensation of SGD 7 for completing the experiment.

\subsubsection{Belief Elicitation Task}
The main task consisted of evaluating the price range of digital artworks (NFTs). For each participant, 10 NFTs from a pool of 20 were randomly chosen and presented visually. Participants were required to allocate 100\% probability across four clearly defined price intervals: SGD 0--250, SGD 250--7,500, SGD 7,500--21,500, and SGD 21,500 and above. The complete list of NFTs used in the experiment, along with their actual prices and assigned ranges, is provided in Table~\ref{tab:nft_list} in the appendix.

After providing their own belief assessment, participants subsequently provided a separate estimation regarding the average probability distribution they believed other participants would assign to the same price intervals. This dual assessment is crucial for PMBA: the first-order beliefs (personal assessments) and second-order beliefs (expectations about others) provide the two key inputs that PMBA uses to aggregate information and identify the true state.

\subsubsection{Knowledge Assessment}
Prior to belief elicitation, participants completed a brief knowledge assessment comprising six multiple-choice questions about cryptocurrencies and NFTs. These questions were designed to assess baseline familiarity without revealing information about specific NFT prices or market distributions. This knowledge measure provides exogenous variation that helps us separate participants' true beliefs from reporting noise and cognitive limitations.

\subsection{PMBA Implementation}

To implement PMBA in our experimental setting, we need to address a key challenge: participants' reported beliefs may be noisy or biased due to cognitive limitations, misunderstanding of the task, or strategic misreporting behavior. To obtain cleaner estimates of participants' true beliefs, we use a two-stage approach.

\textbf{Why we need instrumental variables:} Participants' quiz performance provides an exogenous measure of their knowledge about cryptocurrencies and NFTs. This knowledge should be correlated with their ability to form accurate beliefs about NFT prices, but the quiz questions do not contain specific price information that would directly affect their price predictions other than through their first-order beliefs. This makes quiz performance a suitable instrument for cleaning up noisy belief reports.

We implement this via two-stage least squares (2SLS). Individual knowledge, measured by six dummy variables indicating correctness on the quiz questions, serves as instrumental variables for first-order beliefs. The identifying assumption is standard: quiz performance is correlated with the endogenous regressor ($\mu$), while affecting second-order reports ($\alpha$) only through $\mu$.

\textbf{Stage 1: Cleaning up noisy beliefs.} We use participants' quiz performance to obtain cleaner estimates of their true beliefs. For each participant $i$ and price range $k$, we estimate:
\[
\mu_i^k = \sum_{j=1}^6 \theta_{j}^k x_{ij} + c^k + \epsilon_i^k,
\]
where $x_{ij} \in \{0,1\}$ indicates whether participant $i$ correctly answered knowledge question $j$. This stage essentially asks: ``Given a participant's knowledge level, what should their beliefs about each price range be?''

\textbf{Stage 2: The PMBA aggregation model.} Now we apply the core PMBA logic: we relate each participant's second-order belief (their expectation of what others believe) to their cleaned first-order beliefs:
\[
\alpha^l_i = \sum_{k=1}^3 \beta^{l}_k \widehat{\mu}_i^k + \beta_0^l + \epsilon^l_i,
\]
where $\widehat{\mu}_i^k$ are the fitted values from Stage 1. This equation captures the PMBA insight: a participant's expectation of the population average should be a weighted combination of what the population average would be in each possible state, weighted by their beliefs about which state is true. We omit Range 4 beliefs to avoid multicollinearity since the four beliefs sum to 100\%; the implied mean for Range 4 is recovered from the constraint that the four fitted means sum to one (or from the intercept in the omitted-category specification).

\textbf{Generating predictions.} The regression estimates give us the estimated population averages for each price range. We then predict the true state by finding which range's estimated population average is closest to the observed sample average:
\[
\widehat{r}
= \operatorname*{arg\,min}_{r\in\{1,2,3,4\}}
\left\| \hat{\mu} - \bar{\mu}_{r} \right\|_2.
\]
Here $\widehat{r}$ is the PMBA prediction, and $\bar{\mu}_r$ denotes the fitted vector of state-contingent means for range $r$, obtained by stacking the coordinate-wise second-stage regressions; for Range 4, the fitted mean is the one implied by the omitted-category regression (or by the four means summing to one).

\subsection{Results}

\subsubsection{PMBA versus Majority Voting}
Table~\ref{tab:prediction_errors} compares PMBA predictions with majority voting based on both first-order beliefs ($\mu$) and second-order beliefs ($\alpha$). PMBA yields the lowest mean absolute error (1.25), compared with 1.50 for majority voting on first-order beliefs and 1.40 for majority voting on second-order beliefs.

Statistical evidence is strongest against majority voting on first-order beliefs ($\mu$): the paired $t$-test rejects at the 5\% level. Relative to majority voting on second-order beliefs ($\alpha$), PMBA remains directionally better in mean absolute error, but the difference is not statistically significant at conventional levels (including 10\%).

\begin{table}[htbp]
\footnotesize
\centering
\begin{threeparttable}
\caption{Prediction Accuracy: PMBA vs Majority Voting for NFT Price Ranges}\label{tab:prediction_errors}
\begin{tabular}{c@{\hspace{0.5em}}c@{\hspace{0.3em}}c@{\hspace{0.3em}}c@{\hspace{0.3em}}c@{\hspace{0.5em}}c@{\hspace{0.3em}}c@{\hspace{0.3em}}c}
\toprule
\textbf{NFT} & \textbf{True State} & \multicolumn{3}{c}{\textbf{Prediction}} & \multicolumn{3}{c}{\textbf{Absolute Error}} \\
\cmidrule(lr){3-5} \cmidrule(lr){6-8}
& & \textit{PMBA} & \textit{Majority} ($\mu$) & \textit{Majority} ($\alpha$) & \textit{PMBA} & \textit{Majority} ($\mu$) & \textit{Majority} ($\alpha$) \\
\midrule
2  & 1 & 1  & 1 & 1&0  & 0 & 0 \\
11 & 1 &  1 & 1 & 1&0  & 0 &0 \\
17 & 1 &  1 & 1 &1 &0  & 0 &0 \\
18 & 1 &  2 & 3 &3 &1  & 2 &2 \\
7  & 2 & 1  & 1 &1 &1  & 1 &1 \\
8  & 2 & 3  & 1 &2 &1  & 1 &0 \\
9  & 2 & 1  & 1 &1 &1  & 1 &1 \\
10 & 2 &  4 & 1 &1 &2  & 1 &1 \\
13 & 2 &  2 & 1 &1 &0  & 1 &1 \\
15 & 2 &  1 & 1 &1 &1  & 1 &1 \\
19 & 2 &  1 & 1 &1 &1  & 1 &1 \\
1  & 3 & 3  & 1 &1 &0  & 2 & 2\\
4  & 3 & 1  & 1 &1 &2  & 2 &2 \\
6  & 3 & 4  & 1 &1 &1  & 2 &2 \\
12 & 3 &  4 & 2 &3 &1  & 1 &0  \\
14 & 3 &  1 & 1 &1 &2  & 2 &2\\
3  & 4 & 2  & 1 &1 &2  & 3 &3 \\
5  & 4 & 1  & 1 & 1&3  & 3 &3 \\
16 & 4 &  1 & 1 &1 &3  & 3 &3 \\
20 & 4 &  1 & 1 &1 &3  & 3 & 3\\
\midrule
\multicolumn{5}{r}{\textit{Mean absolute error}} & 1.25 & 1.50 & 1.40 \\
\midrule
\multicolumn{5}{r}{Wilcoxon test ($H_a:Median [Error_{PMBA}-Error_{Majority}]<0$)}&$p$-value: & 0.0938    & 0.2812 \\
\multicolumn{5}{r}{Paired $t$-test ($H_a:Error_{PMBA}<Error_{Majority}$)}&$p$-value: & 0.0481     &  0.1896   \\
\bottomrule
\end{tabular}

\begin{tablenotes}[flushleft]
\item \textit{Notes:} This table shows how well each method predicts the true price range of 20 NFTs. \textbf{Majority prediction} means we pick the price range with the highest average belief across participants. \textbf{Absolute error} measures how far off each prediction is: 0 means perfect, 1 means off by one range, 2 means off by two ranges, and so on. PMBA has lower mean absolute error than both majority benchmarks; statistical significance is stronger against majority voting based on first-order beliefs than against majority voting based on second-order beliefs.
\end{tablenotes}
\end{threeparttable}
\end{table}

\subsubsection{PMBA versus Surprisingly Popular}
We also compare PMBA with the Surprisingly Popular (SP) procedure. SP works by making binary predictions: it compares the average of individual beliefs with the observed population average to determine which of two categories is ``surprisingly popular''. However, SP faces limitations in our four-state setting because it can only make binary distinctions.

In our implementation, we first map the four ranges into two groups: low (Ranges 1\&2) and high (Ranges 3\&4). For each NFT, let
\[
\hat{\mu}_{12}=\hat{\mu}_{1}+\hat{\mu}_{2},\qquad
\hat{\alpha}_{12}=\hat{\alpha}_{1}+\hat{\alpha}_{2}.
\]
SP predicts low (1\&2) if $\hat{\mu}_{12}>\hat{\alpha}_{12}$ and high (3\&4) otherwise; ties are assigned to high (3\&4), matching our replication code.

In our data, we find that most SP predictions are not statistically significant: with our sample size, the null hypothesis of equality cannot be rejected at the 10\% level in 16 out of 20 NFTs. This suggests that SP struggles to distinguish between states when there are more than two possibilities.

Table~\ref{tab:score_comparison} shows that PMBA continues to outperform SP overall. PMBA's advantage is most visible in lower price ranges (Ranges 1, 2, and 3), where paired comparisons remain statistically significant ($p$-value = 0.0277 in the paired $t$-test), while full-sample binary-score differences are not statistically significant.

\begin{table}[htbp]
\centering
\tiny
\begin{threeparttable}
\caption{PMBA vs Surprisingly Popular: Fine-Grained vs Binary Predictions}\label{tab:score_comparison}
\begin{tabular}{c@{\hspace{0.3em}}c@{\hspace{0.3em}}c@{\hspace{0.3em}}c@{\hspace{0.3em}}c@{\hspace{0.3em}}c@{\hspace{0.3em}}c@{\hspace{0.3em}}c}
\toprule
\textbf{NFT} & \textbf{True State} & \multicolumn{2}{c}{\textbf{Prediction}} & \multicolumn{2}{c}{\textbf{Score}}  \\
\cmidrule(lr){3-4} \cmidrule(lr){5-6} 
& & PMBA & SP & PMBA & SP  \\
\midrule
2  & 1\&2 (1) & 1\&2 (1) & 3\&4 & 1  & -1 \\
11 & 1\&2 (1) & 1\&2 (1) & 1\&2 & 1  & 1 \\
17 & 1\&2 (1) & 1\&2 (1) & 1\&2 & 1  & 1 \\
18 & 1\&2 (1) & 1\&2 (2) & 1\&2 & 1  & 1 \\
7  & 1\&2 (2) & 1\&2 (1) & 3\&4 & 1  & -1 \\
8  & 1\&2 (2) & 3\&4 (3) & 3\&4 & -1 & -1 \\
9  & 1\&2 (2) & 1\&2 (1) & 1\&2 & 1  & 1 \\
10 & 1\&2 (2) & 3\&4 (4) & 3\&4 & -1 & -1 \\
13 & 1\&2 (2) & 1\&2 (2) & 3\&4 & 1  & -1 \\
15 & 1\&2 (2) & 1\&2 (1) & 1\&2 & 1  & 1 \\
19 & 1\&2 (2) & 1\&2 (1) & 1\&2 & 1  & 1 \\
1  & 3\&4 (3) & 3\&4 (3) & 1\&2 & 1  & -1 \\
4  & 3\&4 (3) & 1\&2 (1) & 3\&4 & -1 & 1 \\
6  & 3\&4 (3) & 3\&4 (4) & 1\&2 & 1  & -1 \\
12 & 3\&4 (3) & 3\&4 (4) & 1\&2 & 1  & -1 \\
14 & 3\&4 (3) & 1\&2 (1) & 1\&2 & -1 & -1 \\
3  & 3\&4 (4) & 1\&2 (2) & 3\&4 & -1 & 1 \\
5  & 3\&4 (4) & 1\&2 (1) & 3\&4 & -1 & 1 \\
16 & 3\&4 (4) & 1\&2 (1) & 3\&4 & -1 & 1 \\
20 & 3\&4 (4) & 1\&2 (1) & 1\&2 & -1 & -1 \\
\midrule
\textit{Average Score} \\
Range 1 & & & & 1.00 & 0.50 \\
Range 2 & & & & 0.43 & -0.14 \\
Range 3 & & & & 0.20 & -0.60 \\
Range 4 & & & & -1.00 & 0.50 \\
\midrule
\textbf{Total} & & & & \textbf{4} & \textbf{0} \\
\midrule
\textit{NFTs from Range 1 \& 2} \\
\multicolumn{4}{r}{McNemar test ($H_a:\,\text{Score}_{PMBA} \neq \text{Score}_{SP}$) (Exact) $p$-value:} & \multicolumn{2}{c}{0.2500} \\
\multicolumn{4}{r}{Wilcoxon signed-rank test ($H_a:\,\mathrm{Median}[\text{Score}_{PMBA}-\text{Score}_{SP}]>0$) (Exact) $p$-value:} & \multicolumn{2}{c}{0.1250} & \\
\multicolumn{4}{r}{Paired $t$-test ($H_a:\,\text{Score}_{PMBA}>\text{Score}_{SP}$) $p$-value:} & \multicolumn{2}{c}{0.0408} & \\
\midrule
\textit{NFTs from Range 1, 2 \& 3} \\
\multicolumn{4}{r}{McNemar test ($H_a:\,\text{Score}_{PMBA} \neq \text{Score}_{SP}$) (Exact) $p$-value:} & \multicolumn{2}{c}{0.1250} & \\
\multicolumn{4}{r}{Wilcoxon signed-rank test ($H_a:\,\mathrm{Median}[\text{Score}_{PMBA}-\text{Score}_{SP}]>0$) (Exact) $p$-value:} & \multicolumn{2}{c}{0.0625} & \\
\multicolumn{4}{r}{Paired $t$-test ($H_a:\,\text{Score}_{PMBA}>\text{Score}_{SP}$) $p$-value:} & \multicolumn{2}{c}{0.0277} & \\
\midrule
\textit{All NFTs} \\
\multicolumn{4}{r}{McNemar test ($H_a:\,\text{Score}_{PMBA} \neq \text{Score}_{SP}$) (Exact) $p$-value:} & \multicolumn{2}{c}{0.7539} & \\
\multicolumn{4}{r}{Wilcoxon signed-rank test ($H_a:\,\mathrm{Median}[\text{Score}_{PMBA}-\text{Score}_{SP}]>0$) (Exact) $p$-value:} & \multicolumn{2}{c}{0.3770} & \\
\multicolumn{4}{r}{Paired $t$-test ($H_a:\,\text{Score}_{PMBA}>\text{Score}_{SP}$) $p$-value:} & \multicolumn{2}{c}{0.2704} & \\

\bottomrule
\end{tabular}

\begin{tablenotes}[flushleft]
\footnotesize
\item \textit{Notes:} This table compares PMBA (which can distinguish between 4 price ranges) with SP (which can only distinguish between 2 broad categories: low prices vs.\ high prices). \textbf{Score} gives +1 for correct predictions and -1 for incorrect predictions. The statistical tests show that PMBA's ability to make finer distinctions leads to better performance, especially for NFTs in the lower price ranges.
\end{tablenotes}
\end{threeparttable}
\end{table}

\subsection{Discussion}

Our experimental results provide support for the theoretical predictions of PMBA\@. The procedure continues to outperform both majority voting and the Surprisingly Popular method in mean-error terms. Inference is mixed across comparisons: gains are statistically stronger against first-order majority voting and in lower-range PMBA-vs-SP comparisons, while some other pairwise comparisons are not statistically significant at conventional levels. Overall, the evidence remains consistent with PMBA effectively leveraging higher-order beliefs in practice.

The additional robustness exercises also preserve the qualitative conclusions. In particular, the OLS-IV implementations with and without intercept produce identical state predictions in this sample, and the no-IV variants remain directionally consistent though statistically weaker.

The superior performance of PMBA is particularly notable given that it operates in a challenging environment with genuine uncertainty about NFT valuations. The fact that participants' knowledge about cryptocurrencies and NFTs serves as an effective instrumental variable for their first-order beliefs validates our theoretical framework's assumption about the relationship between information and beliefs.

\section{Extensions}\label{sec:extensions}
In this section, we show how to extend the basic PMBA to other settings. In Subsection~\ref{sec:guesses}, we recognize that, in practice, eliciting beliefs from agents is difficult---we show that a variant of the procedure can aggregate information if we only elicit ``simple'' information (e.g., report which state is thought more likely) from most agents, and elicit more fine-grained information from only a small number of agents. In Subsection~\ref{sec:IC}, we discuss how to incentivize truthful reporting. For expositional clarity, both extensions are displayed for the case of two states of the world.

\subsection{Aggregating Guesses instead of Beliefs}\label{sec:guesses}

In this section, we discuss how our PMBA procedure can be adjusted to environments where finite guesses, rather than full beliefs, are elicited (e.g., which state is more likely rather than the probability of state $\omega=1$; this is also the formulation in~\cite{prelec2017solution}). Consider a binary state space and an information structure that satisfies Assumptions~\ref{ass:1} and~\ref{ass:simple}. The action-based PMBA procedure is defined in Procedure~\ref{proc:ABPMBA}.

\begin{framedbox}[htbp]
  \caption{Action-based PMBA, 2 states}\label{proc:ABPMBA}
  \begin{minipage}{0.95\textwidth}
\begin{enumerate}
    \item Elicit from each agent $i \in N$, the state in $\{0,1\}$ they consider more likely, $a_i \in \{0,1\}$. If their belief equals $0.5$, have them report $a_i=1$ (note that this tie-breaking cannot itself be enforced by an incentive-compatible elicitation rule). Given $n$ elicited reports, compute
    \begin{align*}
        \hat{\alpha} = \frac1n \sum_{i=1}^n a_i,
    \end{align*}
    which converges in probability to $\bar{\alpha}(\omega)$ under Assumption~\ref{ass:1}.
    \item Select $2$ agents, A and B, such that $a_A \neq a_B$. Elicit from each of these agents $i=A,B$ their belief $\mu_i$ and their expectation of the average report in the population, i.e., denoting $s_{i}$ as the signal seen by agent $i$, elicit for each $i =A,B$:
    \begin{align*}
        \alpha_i = \mathbb{E}\left[\frac1n \sum_{j= 1}^n a_j \bigg| s_i\right],
\end{align*}
    with $n$ matching the sample size in Step~1.
    \item Calculate $\bar{\alpha}(\omega_1), \bar{\alpha}(\omega_2)$ as:
\begin{align}
 &  \left( \begin{array}{c} \bar{\alpha}(\omega_1)\\ \bar{\alpha}(\omega_2) \end{array}\right)  = \left( \begin{array}{cc} \mu_{A} & 1-\mu_{A}\\ \mu_{B} & 1-\mu_{B}  \end{array} \right)^{-1} 
    \left( \begin{array}{c} \alpha_A \\ \alpha_{B}
    \end{array}     \right) \label{eqn:mainstep3}
    \end{align}
Note that the right-hand side is elicited from agents directly; and since $\mu_A \neq \mu_B$, the middle matrix is invertible.
\item Recover the state by comparing $\hat{\alpha}$ calculated in Step~1 with $\bar{\alpha}(\cdot)$ calculated above. The state of the world is the unique $\omega$ such that $\hat{\alpha} \approx \bar{\alpha}(\omega)$.
\end{enumerate}
With sufficiently large samples, Assumption~\ref{ass:noherding} and the law of large numbers ensure that the empirical objects in Steps~1--2 concentrate around their population counterparts so that the final comparison identifies the true state.
  \end{minipage}
\end{framedbox}

For the Action-based PMBA procedure to be viable, we require the existence of two agents $A,B$ with differing first-order guesses, $a_A \neq a_B$. Assumptions~\ref{ass:1} and~\ref{ass:simple} are not sufficient to ensure this. The following assumption, however, guarantees the existence of two such agents.

\begin{assumption}[No Herding]\label{ass:noherding}
We assume that there exists $\delta>0$ such that, for any agent $i$ and any state $\omega_j \in \{\omega_1,\omega_2\}$, the probability that agent $i$ assigns more than one-half to state $\omega_j$ conditional on their signal $s_i$ is at least $\delta$.
\end{assumption}

\begin{assumption}[Action-rate separation]\label{ass:absep}
The state-contingent population action averages are distinct:
\[
\bar\alpha(\omega_1)\neq \bar\alpha(\omega_2).
\]
\end{assumption}

We can now present our result on aggregation in this environment.

\begin{restatable}{proposition}{BinaryGuesses}\label{prop:BG}
Consider a binary state space. Suppose that the information structure satisfies Assumptions~\ref{ass:1},~\ref{ass:simple},~\ref{ass:noherding}, and~\ref{ass:absep}. Then the \emph{Action-based PMBA} procedure (Procedure~\ref{proc:ABPMBA}) correctly recovers the true state of the world in probability.
\end{restatable}

\begin{proof}
Under Assumption~\ref{ass:1}, Step~1 of Procedure~\ref{proc:ABPMBA} yields
\[
\hat\alpha=\frac1n\sum_{i=1}^n a_i \to_P \bar\alpha(\omega),
\]
where $\bar\alpha(\omega)$ is the state-contingent population average action. Assumption~\ref{ass:noherding} implies that with probability approaching one (as the population grows) there exist agents $A,B$ with different actions and hence different beliefs; thus, the $2\times2$ matrix in Step~3 is invertible. By the law of iterated expectations, the two equations in Step~3 recover $(\bar\alpha(\omega_1),\bar\alpha(\omega_2))$. Assumption~\ref{ass:absep} then guarantees that Step~4 separates states uniquely as $n$ grows. Therefore the procedure identifies the true state in probability.
\end{proof}

\subsection{Incentivizing Truthful Reporting}\label{sec:IC}
In our analysis so far we abstracted away from incentive issues, assuming truth-telling. We now discuss how truthful elicitation and aggregation can simultaneously be achieved via a slightly adjusted mechanism that includes payments which ensure truthful reporting is a (strict) Bayesian Nash equilibrium. 

Recall that when there is a verifiable state that is or will be known to the principal, the principal can truthfully elicit agents' beliefs about the state by paying them using a proper scoring rule. A proper scoring rule pays agents as a function of the reported beliefs and the true state so that for any agent, truthful reporting of beliefs is a dominant strategy. For example, classic scoring rules such as the Brier score~\citep{brier1950verification} or the logarithmic scoring rule make truthful reporting a dominant strategy for expected utility maximizing agents.\footnote{There is a large literature proposing more robust scoring rules for various settings, see e.g.,~\cite{karni2009mechanism} for an example.}

When the state is not verifiable/observable to the principal, proper scoring rules are not applicable. However, the principal may nevertheless be interested in ensuring that agents are incentivized to report their beliefs truthfully. A leading example from~\cite{prelec2004bayesian} is the case of subjective surveys/evaluations of a new product. In their example, the principal (e.g., a marketer) wishes to incentivize truthful reporting by agents, but there is no ``ground truth'' on which to base the rewards. The basic idea of Bayesian Truth Serum of~\cite{prelec2004bayesian} is that agents' beliefs also have implications on their second-order beliefs (i.e., how they think the population will evaluate them). The population averages are indeed verifiable, and therefore can be used to incentivize agents. A similar idea works here, we describe it loosely for completeness:
\begin{enumerate}
    \item Elicit beliefs and second-order beliefs from agents as prescribed in PMBA to learn the true state $\omega$ (which variant of PMBA is used, and what is elicited, depends on the setting). 
    \item For the agents who were asked for second-order beliefs, pay them based on a proper scoring rule as a function of their expected population average beliefs, and the true calculated population average beliefs.
    \item For all other agents, pay them using a proper scoring rule as a function of their beliefs and the state $\omega$ recovered by PMBA. 
\end{enumerate}
Observe that under the maintained assumptions of correctness of the corresponding PMBA procedure, truthful reporting is a perfect Bayesian equilibrium among agents. Since the population is large, average beliefs are effectively observed: a single agent's misreport has negligible effect on that average. Therefore, rewarding the agents who report second-order beliefs with a proper scoring rule based on expected and realized population averages is feasible and incentive compatible. Those agents have no direct payoff from misreporting first-order beliefs. Given truthful second-order reports, the state can be recovered; a proper scoring rule against this recovered state then incentivizes truthful first-order reports for all other agents. Finally, for agents who report both objects, there is no profitable ``compound'' deviation: the second-order incentive is already truthful, and a first-order misreport can only jeopardize state recovery and thus expected payment.

\section{Related Literature and Conclusions}\label{sec:rellit}
In what follows, we discuss the related literature and then conclude by summarizing our contributions.

The literature on the aggregation of information in a population is perhaps far too vast to cite comprehensively. Multiple strands attempt to understand, in a variety of settings, when and how existing institutions aggregate dispersed information in strategic settings. These include studies on information aggregation in voting, a vast literature with its roots in Condorcet, but more recently studied formally starting from~\cite{feddersen1997voting}; aggregation in auctions (see, e.g.,~\cite{milgrom1982theory} and subsequent literature); and information aggregation in markets (see, e.g.,~\cite{grossman1980impossibility}).~\cite{ostrovsky2009information} proposes a market scoring rule mechanism that achieves information aggregation for some classes of securities in a dynamic market, but the payment rule there requires the state to be ex-post observable.

We limit attention here explicitly to work that studies the design/possibility of aggregation procedures/algorithms, or institutions that solely exist to aggregate information.

A key idea in the former space (indeed, in a sense an idea that is also at the heart of~\cite{prelec2017solution}) is the Bayesian Truth Serum~\citep{prelec2004bayesian}, a procedure to truthfully elicit subjective information from agents by rewarding them as a function of others' reports. This paper pioneered the idea that agents' higher-order beliefs could be used to incentivize them successfully even if the planner/designer did not know the prior among agents. A literature studying this followed.~\cite{prelec2006algorithm} showed how to modify agents' reports to aggregate them better (akin to the prediction-normalized voting rule we describe in Example~\ref{ex} in Appendix~\ref{app:counterex}). Several extensions to the BTS have been proposed---see, e.g.,~\cite{cvitanic2019honesty},~\cite{witkowski2012robust},~\cite{radanovic2013robust}, and~\cite{radanovic2014incentives}. The latter three papers attempt to make the basic Bayesian Truth Serum robust, i.e., to ensure that it has good properties even in finite populations and to ensure that participation is individually rational, etc. More recent work has also expanded the scope of the question to continuous spaces, but with more structure---see, e.g.,~\cite{kong2020information}.

These ideas have also been successfully employed in real-world applications: see, e.g.,~\cite{shaw2011designing} for an application on Mechanical Turk, or~\cite{rigol2017targeting}, who use the Robust BTS of~\cite{witkowski2012robust} to identify high-ability micro-entrepreneurs by surveying peers. There is also a growing literature attempting to elicit and use higher-order beliefs to improve aggregation---see, e.g., recent works such as~\cite{palley2019extracting} and~\cite{palley2020boosting}. Some of these works consider modifications of the SP algorithm that improve performance in experiments, e.g.,~\cite{doi:10.1287/mnsc.2020.3919} or~\cite{10.1371/journal.pone.0232058}. For example,~\cite{doi:10.1287/mnsc.2020.3919} show that, in contrast to SP, SC weights forecasters with more informative private signals more than forecasters with less informative ones. This explains why SC is favorable relative to SP when both guess/vote and posterior data are available. Similarly,~\cite{doi:10.1287/mnsc.2023.4955} use Bayesian hierarchical methods to improve performance relative to SP.

On the latter, there is an important literature that has attempted to understand the performance of prediction markets: see, e.g.,~\cite{wolfers2004prediction} for an early piece summarizing the issues and~\cite{baillon2017bayesian} for an alternate design.~\cite{wolfers2006prediction} study when the prediction market price can be interpreted as the average of traders' beliefs, while~\cite{ottaviani2015price} study when the market price under-reacts to new information.~\cite{dai2020wisdom} show how to infer the state of the world from prediction-market trading data using these theoretical ideas. There is also a large literature on pari-mutuel markets, and on adapting them to aggregate information---see, e.g.,~\cite{pennock2004dynamic}.

As discussed in the Introduction, our approach has some similarities with Theorem~1.4 of~\cite{prelec2017solution}. This approach requires elicitation of $p(s_i\mid s_j)$, i.e., the conditional probability that a different agent has seen signal $s_i$ given that you have seen $s_j$. In particular, this requires that the subjects understand the meaning of signals (H and T in their example). In contrast, PMBA questions are all in terms of primitives (beliefs and higher-order beliefs) of an agent. Our position is that this distinction is important because if the subjects really understand the model/prior, including the meaning of signals, the easiest question is to simply ask them for the state-contingent population signal distribution (i.e., the bias of the various possible coins in the terminology of PSM).

A different approach to ours for aggregating information with more than two states can be found in~\cite{libgober2021hypothetical}. The principal does not know the agent's prior, but can elicit, after an unknown Blackwell experiment, their posterior and their contingent hypothetical beliefs (their posteriors contingent on seeing any signal). He shows that the prior and experiment can be recovered from these contingent hypothetical beliefs, via a linear-regression method, even if some signals are not sampled (i.e., $\mu$ need not be a square matrix). In contrast,~\cite{prelec2022general}, following~\cite{samet1998iterated}, propose identifying contingent hypothetical beliefs with a transition matrix on signals. An invariant distribution of the matrix corresponds to the ex-ante population signal distribution and thereby delivers the unknown information structure. At a technical level, our ``expected population average beliefs'' serve a similar role to the ``contingent hypothetical beliefs'' in these papers. However, the contingent hypothetical beliefs include information about higher-order beliefs well beyond the expected population mean; moreover, both~\cite{libgober2021hypothetical} and~\cite{prelec2022general} require that the agents share a common prior and thereby rule out misspecified information settings (Procedure~\ref{proc:lipmba}).


\bibliographystyle{econometrica}
\bibliography{tci}

@String{Computer = "{IEEE} Computer" }

@String{Springer = "Springer-Verlag" }

@article{ngangoue2021learning,
Author = {Ngangoué, M. Kathleen and Weizsäcker, Georg},
Title = {Learning from Unrealized versus Realized Prices},
Journal = {American Economic Journal: Microeconomics},
Volume = {13},
Number = {2},
Year = {2021},
Month = {May},
Pages = {174-201},
DOI = {10.1257/mic.20180268},
URL = {https://www.aeaweb.org/articles?id=10.1257/mic.20180268}}

@article{arieli2017crowd,
  title={When is the Crowd Wise?},
  author={Arieli, Itai and Babichenko, Yakov and Smorodinsky, Rann},
  journal={Available at SSRN 3083608},
  year={2017}
}

@article{esponda2014hypothetical,
Author = {Esponda, Ignacio and Vespa, Emanuel},
Title = {Hypothetical Thinking and Information Extraction in the Laboratory},
Journal = {American Economic Journal: Microeconomics},
Volume = {6},
Number = {4},
Year = {2014},
Month = {November},
Pages = {180-202},
DOI = {10.1257/mic.6.4.180},
URL = {https://www.aeaweb.org/articles?id=10.1257/mic.6.4.180}}

@article{charness2009winnerscurse,
Author = {Charness, Gary and Levin, Dan},
Title = {The Origin of the Winner's Curse: A Laboratory Study},
Journal = {American Economic Journal: Microeconomics},
Volume = {1},
Number = {1},
Year = {2009},
Month = {February},
Pages = {207-36},
DOI = {10.1257/mic.1.1.207},
URL = {https://www.aeaweb.org/articles?id=10.1257/mic.1.1.207}}

@article{vespa2019contingent,
Author = {Martínez-Marquina, Alejandro and Niederle, Muriel and Vespa, Emanuel},
Title = {Failures in Contingent Reasoning: The Role of Uncertainty},
Journal = {American Economic Review},
Volume = {109},
Number = {10},
Year = {2019},
Month = {October},
Pages = {3437-74},
DOI = {10.1257/aer.20171764},
URL = {https://www.aeaweb.org/articles?id=10.1257/aer.20171764}}

@article{li2017obviously,
Author = {Li, Shengwu},
Title = {Obviously Strategy-Proof Mechanisms},
Journal = {American Economic Review},
Volume = {107},
Number = {11},
Year = {2017},
Month = {November},
Pages = {3257-87},
DOI = {10.1257/aer.20160425},
URL = {https://www.aeaweb.org/articles?id=10.1257/aer.20160425}}

@misc{libgober2021hypothetical,
      title={Hypothetical Beliefs Identify Information}, 
      author={Jonathan Libgober},
      year={2021},
      eprint={2105.07097},
      archivePrefix={arXiv},
      primaryClass={econ.TH},
      note={arXiv preprint arXiv:2105.07097}
}

@article{palley2019extracting,
  title={Extracting the wisdom of crowds when information is shared},
  author={Palley, Asa B and Soll, Jack B},
  journal={Management Science},
  volume={65},
  number={5},
  pages={2291--2309},
  year={2019},
  publisher={INFORMS}
}

@article{palley2020boosting,
  title={Boosting the wisdom of crowds within a single judgment problem: Selective averaging based on peer predictions},
  author={Palley, Asa and Satop{\"a}{\"a}, Ville},
  journal={Available at SSRN 3504286},
  year={2020}
}

@article{prelec2017solution,
  title={A solution to the single-question crowd wisdom problem},
  author={Prelec, Dra{\v{z}}en and Seung, H Sebastian and McCoy, John},
  journal={Nature},
  volume={541},
  number={7638},
  pages={532},
  year={2017},
  publisher={Nature Publishing Group}
}

@article{feddersen1997voting,
  title={Voting behavior and information aggregation in elections with private information},
  author={Feddersen, Timothy and Pesendorfer, Wolfgang},
  journal={Econometrica: Journal of the Econometric Society},
  pages={1029--1058},
  year={1997},
  publisher={JSTOR}
}

@article{wolfers2004prediction,
  title={Prediction markets},
  author={Wolfers, Justin and Zitzewitz, Eric},
  journal={Journal of economic perspectives},
  volume={18},
  number={2},
  pages={107--126},
  year={2004}
}

@article{baillon2017bayesian,
  title={Bayesian markets to elicit private information},
  author={Baillon, Aur{\'e}lien},
  journal={Proceedings of the National Academy of Sciences},
  volume={114},
  number={30},
  pages={7958--7962},
  year={2017},
  publisher={National Acad Sciences}
}

@techreport{wolfers2006prediction,
  title={Prediction markets in theory and practice},
  author={Wolfers, Justin and Zitzewitz, Eric},
  year={2006},
  institution={national bureau of economic research}
}

@article{ottaviani2015price,
  title={Price reaction to information with heterogeneous beliefs and wealth effects: Underreaction, momentum, and reversal},
  author={Ottaviani, Marco and S{\o}rensen, Peter Norman},
  journal={American Economic Review},
  volume={105},
  number={1},
  pages={1--34},
  year={2015}
}

@article{prelec2006algorithm,
  title={An algorithm that finds truth even if most people are wrong},
  author={Prelec, Drazen and Seung, Sebastian},
  journal={Unpublished manuscript},
  volume={69},
  year={2006}
}

@article{cvitanic2019honesty,
  title={Honesty via choice-matching},
  author={Cvitani{\'c}, Jak{\v{s}}a and Prelec, Dra{\v{z}}en and Riley, Blake and Tereick, Benjamin},
  journal={American Economic Review: Insights},
  volume={1},
  number={2},
  pages={179--92},
  year={2019}
}

@inproceedings{radanovic2013robust,
  title={A robust bayesian truth serum for non-binary signals},
  author={Radanovic, Goran and Faltings, Boi},
  booktitle={Proceedings of the AAAI Conference on Artificial Intelligence},
  volume={27},
  year={2013}
}

@inproceedings{radanovic2014incentives,
  title={Incentives for truthful information elicitation of continuous signals},
  author={Radanovic, Goran and Faltings, Boi},
  booktitle={Proceedings of the AAAI Conference on Artificial Intelligence},
  volume={28},
  year={2014}
}

@article{brier1950verification,
  title={Verification of forecasts expressed in terms of probability},
  author={Brier, Glenn W},
  journal={Monthly weather review},
  volume={78},
  number={1},
  pages={1--3},
  year={1950}
}

@article{karni2009mechanism,
  title={A mechanism for eliciting probabilities},
  author={Karni, Edi},
  journal={Econometrica},
  volume={77},
  number={2},
  pages={603--606},
  year={2009},
  publisher={Wiley Online Library}
}

@inproceedings{witkowski2012robust,
  title={A robust bayesian truth serum for small populations},
  author={Witkowski, Jens and Parkes, David},
  booktitle={Proceedings of the AAAI Conference on Artificial Intelligence},
  volume={26},
  year={2012}
}

@article{prelec2004bayesian,
  title={A Bayesian truth serum for subjective data},
  author={Prelec, Dra{\v{z}}en},
  journal={science},
  volume={306},
  number={5695},
  pages={462--466},
  year={2004},
  publisher={American Association for the Advancement of Science}
}

@inproceedings{pennock2004dynamic,
  title={A dynamic pari-mutuel market for hedging, wagering, and information aggregation},
  author={Pennock, David M},
  booktitle={Proceedings of the 5th ACM conference on Electronic commerce},
  pages={170--179},
  year={2004}
}

@article{dai2020wisdom,
  title={The wisdom of the crowd and prediction markets},
  author={Dai, Min and Jia, Yanwei and Kou, Steven},
  journal={Journal of Econometrics},
  year={2020},
  publisher={Elsevier}
}

@article{grossman1980impossibility,
  title={On the impossibility of informationally efficient markets},
  author={Grossman, Sanford J and Stiglitz, Joseph E},
  journal={The American economic review},
  volume={70},
  number={3},
  pages={393--408},
  year={1980},
  publisher={JSTOR}
}

@article{milgrom1982theory,
  title={A theory of auctions and competitive bidding},
  author={Milgrom, Paul R and Weber, Robert J},
  journal={Econometrica: Journal of the Econometric Society},
  pages={1089--1122},
  year={1982},
  publisher={JSTOR}
}

@book{mertens2015repeated,
  title={Repeated games},
  author={Mertens, Jean-Fran{\c{c}}ois and Sorin, Sylvain and Zamir, Shmuel},
  volume={55},
  year={2015},
  publisher={Cambridge University Press}
}

@article{lipman2003finite,
  title={Finite order implications of common priors},
  author={Lipman, Barton L},
  journal={Econometrica},
  volume={71},
  number={4},
  pages={1255--1267},
  year={2003},
  publisher={Wiley Online Library}
}

@inproceedings{shaw2011designing,
  title={Designing incentives for inexpert human raters},
  author={Shaw, Aaron D and Horton, John J and Chen, Daniel L},
  booktitle={Proceedings of the ACM 2011 conference on Computer supported cooperative work},
  pages={275--284},
  year={2011}
}

@article{rigol2017targeting,
  title={Targeting high ability entrepreneurs using community information: Mechanism design in the field},
  author={Rigol, Natalia and Hussam, Reshmaan and Roth, Benjamin},
  journal={American Economic Review, forthcoming},
  year={2021}
}

@ARTICLE{mertens1985formulation,
  author = {Mertens, Jean-Francois and Zamir, Shmuel},
  title = {Formulation of Bayesian analysis for games with incomplete information},
  journal = {International Journal of Game Theory},
  year = {1985},
  volume = {14},
  pages = {1--29},
  number = {1},
  publisher = {Springer}
}

@article{golub2010naive,
  title={Naive learning in social networks and the wisdom of crowds},
  author={Golub, Benjamin and Jackson, Matthew O},
  journal={American Economic Journal: Microeconomics},
  volume={2},
  number={1},
  pages={112--49},
  year={2010}
}

@article{austen1996information,
  title={Information aggregation, rationality, and the Condorcet jury theorem},
  author={Austen-Smith, David and Banks, Jeffrey S},
  journal={American political science review},
  volume={90},
  number={1},
  pages={34--45},
  year={1996},
  publisher={Cambridge University Press}
}

@article{acemoglu2011bayesian,
  title={Bayesian learning in social networks},
  author={Acemoglu, Daron and Dahleh, Munther A and Lobel, Ilan and Ozdaglar, Asuman},
  journal={The Review of Economic Studies},
  volume={78},
  number={4},
  pages={1201--1236},
  year={2011},
  publisher={Oxford University Press}
}

@article{fu2021full,
  title={Full surplus extraction from samples},
  author={Fu, Hu and Haghpanah, Nima and Hartline, Jason and Kleinberg, Robert},
  journal={Journal of Economic Theory},
  volume={193},
  pages={105230},
  year={2021},
  publisher={Elsevier}
}

@article{samet1998iterated,
  title={Iterated expectations and common priors},
  author={Samet, Dov},
  journal={Games and economic Behavior},
  volume={24},
  number={1-2},
  pages={131--141},
  year={1998},
  publisher={Elsevier}
}

@article{doi:10.1287/mnsc.2020.3919,
author = {Wilkening, Tom and Martinie, Marcellin and Howe, Piers D. L.},
title = {Hidden Experts in the Crowd: Using Meta-Predictions to Leverage Expertise in Single-Question Prediction Problems},
journal = {Management Science, Forthcoming},
year = {2021},
doi = {10.1287/mnsc.2020.3919},

URL = { 
        https://doi.org/10.1287/mnsc.2020.3919
    
},
eprint = { 
        https://doi.org/10.1287/mnsc.2020.3919
    
}
}

@article{10.1371/journal.pone.0232058,
    doi = {10.1371/journal.pone.0232058},
    author = {Martinie, Marcellin AND Wilkening, Tom AND Howe, Piers D. L.},
    journal = {PLOS ONE},
    publisher = {Public Library of Science},
    title = {Using meta-predictions to identify experts in the crowd when past performance is unknown},
    year = {2020},
    month = {04},
    volume = {15},
    url = {https://doi.org/10.1371/journal.pone.0232058},
    pages = {1-11},
    number = {4},

}

@article{gul1998comment,
  title={A comment on Aumann's Bayesian view},
  author={Gul, Faruk},
  journal={Econometrica},
  volume={66},
  number={4},
  pages={923--927},
  year={1998},
  publisher={JSTOR}
}

@misc{prelec2022general,
  title={General identifiability of possible world models for crowd wisdom},
  author={Prelec, Drazen and McCoy, John},
  note={PsyArXiv preprint},
  year={2022},
  publisher={PsyArXiv}
}

@inproceedings{ostrovsky2009information,
  title={Information aggregation in dynamic markets with strategic traders},
  author={Ostrovsky, Michael},
  booktitle={Proceedings of the 10th ACM conference on Electronic commerce},
  pages={253--254},
  year={2009}
}

@inproceedings{kong2020information,
  title={Information elicitation mechanisms for statistical estimation},
  author={Kong, Yuqing and Schoenebeck, Grant and Tao, Biaoshuai and Yu, Fang-Yi},
  booktitle={Proceedings of the AAAI Conference on Artificial Intelligence},
  volume={34},
  pages={2095--2102},
  year={2020}
}

@article{doi:10.1287/mnsc.2023.4955,
author = {McCoy, John and Prelec, Drazen},
title = {A Bayesian Hierarchical Model of Crowd Wisdom Based on Predicting Opinions of Others},
journal = {Management Science},
year = {2023},
doi = {10.1287/mnsc.2023.4955},
note = {Published online October 19, 2023},

URL = { 
    
        https://doi.org/10.1287/mnsc.2023.4955
    
    

},
eprint = { 
    
        https://doi.org/10.1287/mnsc.2023.4955
    
    

}
}

@article{MMFJET2014,
  title={Does one Bayesian make a Difference?},
  author={Mueller-Frank, Manuel},
  journal={Journal of Economic Theory},
  volume={154},
  pages={423--452},
  year={2014},
  publisher={Elsevier}
}

@book{vandervaartwellner1996,
  title={Weak Convergence and Empirical Processes},
  author={van der Vaart, Aad W. and Wellner, Jon A.},
  year={1996},
  publisher={Springer}
}

@book{vershynin2018highdim,
  title={High-Dimensional Probability: An Introduction with Applications in Data Science},
  author={Vershynin, Roman},
  year={2018},
  publisher={Cambridge University Press}
}

\appendix

\section{Proofs from Section~\ref{sec:pmba}}\label{app:proofofThm1}

First we prove a general property:
\begin{lemma}\label{lemma:beliefsconcentrate}
Suppose Part (2) of Assumption~\ref{ass:1} is satisfied. Then for any state $\omega$,
\[
\lim_{n \rightarrow \infty} \frac1n \sum_{i=1}^n \tmu_i \to_P \bar{\mu}(\omega).
\]
\end{lemma}

\begin{proof}
Fix a realized state $\omega$ and coordinate $\omega' \in \Omega$. Define
\[
A_n:=\frac1n\sum_{i=1}^n \tmu_{i,\omega'},
\qquad
m_n:=\frac1n\sum_{i=1}^n \mu_{i,\omega'}(\omega).
\]
By definition of $\mu_{i,\omega'}(\omega)=\mathbb E[\tmu_{i,\omega'}\mid \omega]$, we have
$\mathbb E[A_n\mid \omega]=m_n$. It is therefore enough to show
$A_n-m_n\to_P 0$.

We have:
\begin{align*}
\text{Var}(A_n\mid \omega) &= \frac{1}{n^2} \sum_{i=1}^n \sum_{j=1}^n \text{Cov}(\tmu_{i,\omega'}, \tmu_{j,\omega'}|\omega).
\end{align*}

By Assumption~\ref{ass:1}(2), for each agent $i$ and any $\epsilon>0$, there exists a set $N_i \subseteq N$ with $|N_i| \leq n(\epsilon)$ such that for all $j \notin N_i$, the signals of agents $i$ and $j$ are $\epsilon$-independent conditional on $\omega$. Since beliefs are bounded measurable functions of signals, this event-wise bound implies (by approximation of bounded functions by simple functions) that for $j \notin N_i$,
\[
|\text{Cov}(\tmu_{i,\omega'}, \tmu_{j,\omega'}\mid\omega)| \leq C\epsilon
\]
for a universal constant $C$.

Split the double sum into pairs where $j \in N_i$ (correlated) and $j \notin N_i$ ($\epsilon$-independent):
\begin{align*}
\text{Var}(A_n\mid \omega) &\leq \frac{1}{n^2} \left[ \sum_{i=1}^n \sum_{j \in N_i} |\text{Cov}(\tmu_{i,\omega'}, \tmu_{j,\omega'}|\omega)| + \sum_{i=1}^n \sum_{j \notin N_i} |\text{Cov}(\tmu_{i,\omega'}, \tmu_{j,\omega'}|\omega)| \right]\\
&\leq \frac{1}{n^2} \left[ n \cdot n(\epsilon) \cdot 1 + n^2 \cdot C\epsilon \right]\\
&= \frac{n(\epsilon)}{n} + C\epsilon.
\end{align*}

As $n \to \infty$, the first term $\frac{n(\epsilon)}{n} \to 0$ since $n(\epsilon)$ is finite. Since $\epsilon$ is arbitrary, this implies $\text{Var}(A_n\mid \omega)\to 0$. By Chebyshev's inequality,
\[
A_n-m_n \to_P 0.
\]
By definition of $\bar\mu_{\omega'}(\omega)$ (Section~\ref{sec:pmba}), $m_n\to \bar\mu_{\omega'}(\omega)$, hence
\[
\frac1n\sum_{i=1}^n \tmu_{i,\omega'} \to_P \bar\mu_{\omega'}(\omega).
\]
Since $|\Omega|<\infty$, the result extends componentwise to the full vector:
\[
\frac1n\sum_{i=1}^n \tmu_i \to_P \bar\mu(\omega).
\]
\end{proof}

\paragraph{Claim~\ref{cl:fullrank} (restated).}
Suppose that signals are conditionally i.i.d.\ draws from a finite set $S$, and Assumption~\ref{ass:simple} is satisfied. If each agent gets $L-1$ independent draws from the distribution $P_{\omega }^{S}$, then Assumption~\ref{ass:fullrank} is satisfied with the signal space $S^{L-1}$.

\begin{proof}
At each $\omega$, the distribution of the $L-1$ signals can be summarized by
\[
\big(P_{\omega}^{S}\big)^{\otimes (L-1)}.
\]
Since Assumption~\ref{ass:simple} is satisfied, the dimension of the linear space spanned by $\left\{
P_{\omega }^{S}\right\} _{\omega \in \Omega }$ is at least $2$. Hence, by
Lemma 4 in Fu et al.~\citeyearpar{fu2021full}, $\left\{(P_{\omega }^{S})^{\otimes (L-1)}\right\}_{\omega \in \Omega }$ are
linearly independent and hence $\left\{ (P_{\omega}^{S})^{\otimes (L-1)}P\left( \omega \right) \right\}_{\omega \in \Omega }$
are also linearly independent ($P\left( \omega \right) $ denotes the prior
probability of $\omega $). Identify $\left\{ (P_{\omega}^{S})^{\otimes (L-1)}P\left( \omega \right) \right\} _{\omega \in \Omega }$
with a $\left\vert S\right\vert ^{L-1}\times L$ matrix which we denote by $M$%
. Hence, rank$\left( M\right) =L$. For $\mathbf{s}=\left(
s^{1},\ldots,s^{L-1}\right) $ and $\omega \in \Omega $, observe that $M_{%
\mathbf{s},\omega}=P\left( \mathbf{s},\omega \right) $. Since rank$\left(
M\right) =L$, the conditional probability matrix $\left\{ P\left( \omega |%
\mathbf{s}\right) \right\}$ also has rank $L$. Hence, Assumption~\ref{ass:fullrank} is satisfied with
the signal space $S^{L-1}$.
\end{proof}

\subsection{Proof of Theorem~\ref{thm:pmba}}

\begin{proof}
The proof follows from the structure of the PMBA procedure and the established lemmas. We show that each step of Procedure~\ref{proc:pmba2} works as intended under the given assumptions.

\textbf{Step 1:} Under Assumption~\ref{ass:1}, the law of large numbers ensures that $\hat{\mu} = \frac1n \sum_{i=1}^n \mu_i$ converges in probability to $\bar{\mu}(\omega)$ as $n \to \infty$.
By our assumption, the limit $\bar{\mu}(\omega)$ exists almost surely (see Footnote~6 on page~5). Therefore $\hat{\mu}$ converges to $\bar{\mu}(\omega)$ almost surely.

\textbf{Step 2:} By Assumption~\ref{ass:fullrank}, we can select $L$ agents whose belief vectors form a full-rank matrix $\bm{\mu}$. This is guaranteed to exist since the assumption requires that the convex hull of the support of beliefs has an interior relative to $\Delta^L$.

\textbf{Step 3:} For each selected agent $i$, define the finite-$n$ second-order object
\[
\alpha_i^{(n)}:= \mathbb{E}\!\left[\frac1n \sum_{j=1}^n \tmu_j \,\bigg|\, s_i\right].
\]
By the law of iterated expectations, for each fixed $n$:
\begin{align*}
\alpha_i^{(n)} &= \mathbb{E}\left[\frac1n \sum_{j=1}^n \tmu_j \bigg| s_i\right] \\
&= \sum_{\omega' \in \Omega} \mathbb{E}\left[\frac1n \sum_{j=1}^n \tmu_j \bigg| s_i, \omega'\right] P(\omega' \mid s_i).
\end{align*}
For a fixed $\omega'$, one can decompose the inner average into agents in the finite dependence neighborhood of $i$ and all remaining agents; the first part is $O(1/n)$, and the second part converges (by Lemma~\ref{lemma:beliefsconcentrate}) to $\bar\mu(\omega')$. Hence, conditional on each $\omega'$, the inner expectation converges to $\bar\mu(\omega')$, and therefore
\[
\alpha_i^{(n)} \to \alpha_i^\ast:=\sum_{\omega' \in \Omega}\bar{\mu}(\omega')\,\mu_{i,\omega'}
\]
in probability as $n\to\infty$.

\textbf{Step 4:} In the limit, the system of equations $\bm{\alpha}^{\ast} = \bm{\mu} \bar{\bm{\mu}}$ has a unique solution $\bar{\bm{\mu}} = \bm{\mu}^{-1} \bm{\alpha}^{\ast}$ because $\bm{\mu}$ is invertible by Step 2. For large finite $n$, the elicited $\bm{\alpha}^{(n)}$ is close to $\bm{\alpha}^{\ast}$.

\textbf{State Identification:} By Assumption~\ref{ass:simple}, distinct states lead to different population average beliefs, so $\bar{\mu}(\omega) \neq \bar{\mu}(\omega')$ for any $\omega \neq \omega'$. Since $\hat{\mu}$ converges to $\bar{\mu}(\omega)$ in probability, the true state $\omega$ is identified as the row of $\bar{\bm{\mu}}$ that minimizes the distance to $\hat{\mu}$.

The procedure recovers the true state $\omega$ in probability because all the convergence results hold in probability.
\end{proof}
 
\subsection{Proof of Lemma~\ref{lemma:Genericity}}
\begin{proof} We first establish Lemma~\ref{lemma:Genericity}(1):
\noindent The probability mass of $q(s)$ is equal to
\begin{equation*}
    \Pr[q(s)|\omega]=p_{s\omega},
\end{equation*}
and the conditional probability of $q(s)\in b$ in state $\omega$ conditional on partition cell $b$ is equal to
\begin{equation*}
    \frac{p_{s\omega}}{\displaystyle \sum_{s':q(s')\in b}p_{s'\omega}}
\end{equation*}
This implies that, in state $\omega$, the conditional probability-weighted average of posterior beliefs in partition cell $b \in B$ is equal to
\begin{equation*}
    \sum_{s:q(s) \in b}\frac{p_{s\omega}}{\displaystyle \sum_{s':q(s')\in b}p_{s'\omega}}\left( \frac{p_{s\omega_1}}{\sum_{j=1}^Lp_{s\omega_j}},\ldots, \frac{p_{s\omega_L}}{\sum_{j=1}^Lp_{s\omega_j}}\right)
\end{equation*}
Fix a coordinate index $j\in\{1,\ldots,L\}$ and write $p_s:=\sum_{\ell=1}^L p_{s\omega_\ell}$ for the unconditional probability of signal $s$. Then the $j$-th coordinate of the previous vector expression is equivalent to
\begin{equation}\label{eq:I}
  \left(
    \frac{1}{\displaystyle \sum_{s'\;:\; q(s')\in b} p_{s'\omega}}
  \right)
  \left(
    \frac{1}{\displaystyle \prod_{s\;:\; q(s)\in b} p_s}
  \right)
  \sum_{s\;:\; q(s)\in b}
  p_{s\omega}\, p_{sj}\!
  \prod_{\hat{s}\neq s} p_{\hat{s}}.
  \tag{I}
\end{equation}
\noindent
For any realized state $\omega\in\Omega$ and information structure $p$, we can now construct an
$L\times L$ matrix that corresponds to a given partition $B$, where row $i$ corresponds to the
conditional probability-weighted average of posterior beliefs in partition cell $i$. Denote this matrix as
$\tilde q^{\,\omega,B}(p)$. Equation~\eqref{eq:I} corresponds to the entry
$\tilde q^{\,\omega,B}_{ij}(p)$ of the matrix, and note that all entries in row $i$ have the same
multiplying factor
\[
  \left(
    \frac{1}{\displaystyle \sum_{s\;:\; q(s)\in b} p_{s\omega}}
  \right)
  \left(\frac{1}{\displaystyle \prod_{s\;:\; q(s)\in b} p_s}
  \right).
\]
The matrix $\tilde q^{\,\omega,B}(p)$ has full rank if its determinant is not equal to zero.
Consider the $L\times L$ matrix $\check q^{\,\omega B}(p)$ where each entry $\check q^{\,\omega B}_{ij}(p)$ is constructed as follows
\[
  \check q^{\,\omega B}_{ij}(p)
  \;=\;
  \tilde q^{\,\omega B}_{ij}(p)\,
  \left(\sum_{s \,:\, q(s)\in b} p_{s\omega}\right)
  \left(\prod_{s \,:\, q(s)\in b} p_s\right).
\]
\noindent
The matrix $\check q^{\,\omega B}(p)$ has a determinant unequal to zero if and only if the matrix
$\tilde q^{\,\omega B}(p)$ also does. Note that each entry of the matrix $\check q^{\,\omega B}(p)$
is a polynomial function of the joint probability measure $p$. As the determinant of a matrix is
also a polynomial of its entries, the determinant of $\check q^{\,\omega B}(p)$ is a polynomial
function. The set of joint probability measures in $\Delta^{K\times L}$ for which the matrix
$\check q^{\,\omega B}(p)$ does not have full rank is given by the zero set $Z^{\omega B}$ of its
determinant:
\[
  Z^{\omega B} \;=\; \{\, p \in \Delta^{KL} : \det\nolimits^{\,\omega B}\!\big(\check q^{\,\omega B}(p)\big)=0 \,\}.
\]
\noindent
In order to establish the claim of the lemma, we need to show that the set of information
structures $p$, for which no partition induces a matrix of conditional expected means with full
rank, has Lebesgue measure zero. In other words, we need to establish that the zero set of the
following function has Lebesgue measure zero in $\Delta^{KL}$:
\[
  f^{\omega}(p)
  \;=\;
  \sum_{B \in \mathcal{B}}
  \Big(\det\nolimits^{\,\omega B}\!\big(\check q^{\,\omega B}(p)\big)\Big)^{2}.
\]
\noindent
Note that $f^{\omega} : \Delta^{KL} \to \mathbb{R}$ is a non-trivial polynomial function (for example, at information structures where signals are highly state revealing, one can choose a partition with nonzero determinant, so $f^\omega(p)>0$).
It then follows from Lemma~2 in~\cite{MMFJET2014} that its zero set has
$\lambda_{\Delta^{KL}}$ measure zero. Since the state space $\Omega$ is finite and the above holds for each $\omega \in \Omega$, it follows that for generic information structures $p$ one can always construct a partition of the posterior beliefs $Q^p$ such that the matrix of probability-weighted average beliefs for each partition cell has full rank.
\end{proof}
\begin{proof} We now establish Lemma~\ref{lemma:Genericity}(2):
\noindent We have that
\[
\bar{\mu}^{\,p}_j(\omega') \;=\; \sum_{h=1}^{K} \frac{p_{h\omega'}\, p_{hj}}{p_{\omega'}\, p_{s_h}}.
\]
For two states $\omega',\omega''$ the mean beliefs of state $j$ coincide if
\[
\sum_{h=1,\ldots,K} \frac{p_{h\omega'}\,p_{hj}}{p_{\omega'}\,p_{s_h}}
\;-\;
\sum_{h=1,\ldots,K} \frac{p_{h\omega''}\,p_{hj}}{p_{\omega''}\,p_{s_h}}
=0
\]
\[
\Longleftrightarrow\qquad
\sum_{h=1,\ldots,K} \Big(\frac{p_{h\omega'}}{p_{\omega'}}-\frac{p_{h\omega''}}{p_{\omega''}}\Big)\,
\frac{p_{hj}}{p_{s_h}} = 0
\]
\[
\Longleftrightarrow\qquad
\frac{1}{p_{\omega'}p_{\omega''}}
\sum_{h=1,\ldots,K}
\big(p_{h\omega'}p_{\omega''}-p_{h\omega''}p_{\omega'}\big)\,
\frac{p_{hj}}{p_{s_h}} = 0
\]
\[
\Longleftrightarrow\qquad
\frac{1}{p_{\omega'}p_{\omega''}} \frac{1}{\displaystyle \prod_{h=1,\ldots,K} p_{s_h}}
\\
\sum_{h=1,\ldots,K}
\big(p_{h\omega'}p_{\omega''}-p_{h\omega''}p_{\omega'}\big)\,
p_{hj}\!
\left(\prod_{g\ne h} p_{s_g}\right)=0.
\]

This expression is equivalent to
\[
f^{\omega'\omega'' j}(p)
\;=\;
\sum_{h=1,\ldots,K}
\big(p_{h\omega'}p_{\omega''}-p_{h\omega''}p_{\omega'}\big)\,p_{hj}\!
\left(\prod_{g\ne h} p_{s_g}\right)=0.
\]

In order to show that for $\lambda_{\Delta^{KL}}$-almost every information structure the mean belief of state $j$
is different in states $\omega'$ and $\omega''$, it is sufficient to show that the zero set of
$f^{\omega'\omega'' j}:\Delta^{KL}\to\mathbb{R}$ has
$\lambda_{\Delta^{K\times L}}$ measure zero. Note that $f^{\omega'\omega'' j}$ is a nontrivial
polynomial function (e.g., if state-conditioned signal distributions differ on at least one signal, the expression is nonzero for some $j$). Again invoking Lemma~2 from~\cite{MMFJET2014}, it follows that the
zero set of $f^{\omega'\omega'' j}$ has $\lambda_{\Delta^{KL}}$ measure zero. Let
$Z^{\omega'\omega'' j}$ denote the zero set of $f^{\omega'\omega'' j}$. Since each set
$Z^{\omega'\omega'' j}$ has Lebesgue measure zero and since there are only finitely many states,
it follows that
\[
\bigcup_{j=1,\ldots,L}\ \ \bigcup_{(\omega',\omega'')\in\Omega\times\Omega}
Z^{\omega'\omega'' j}
\]
has Lebesgue measure zero, concluding the proof.
\end{proof}

\subsection{Proofs of Section~\ref{sec:misspecified}}
See below for the proof of Lemma~\ref{lem:misspec2}.
\begin{proof}
Let $x_{iN_k}$ denote the indicator that agent $i$ belongs to group $N_k$. Write
\[
\hat{\mu}_{N_k}^n=\frac{\sum_{i\le n}x_{iN_k}\mu_i}{\sum_{i\le n}x_{iN_k}}
=\frac{K_n^k}{L_n^k},
\quad
L_n^k:=\frac1n\sum_{i\le n}x_{iN_k},
\quad
K_n^k:=\frac1n\sum_{i\le n}x_{iN_k}\mu_i.
\]
Under Assumption~\ref{ass:1Mis}, the sequences $\{x_{iN_k}\}_i$ and $\{x_{iN_k}\mu_i\}_i$ are conditionally i.i.d.\ (given the realized state), so by the strong law,
\[
L_n^k \xrightarrow{a.s.} L^k:=\mathbb{E}[x_{1N_k}],
\qquad
K_n^k \xrightarrow{a.s.} K^k:=\mathbb{E}[x_{1N_k}\mu_1].
\]
By construction we restrict attention to partitions with positive asymptotic group mass, so $L^k>0$. Hence, by the Continuous Mapping Theorem,
\[
\hat{\mu}_{N_k}^n=\frac{K_n^k}{L_n^k}\xrightarrow{a.s.}\frac{K^k}{L^k}
=\mathbb{E}\!\left[\mu_1\mid 1\in N_k\right].
\]
This limit is deterministic conditional on the state, proving the claim.
\end{proof}
See below for the proof of Lemma~\ref{lem:misspec}.
\begin{proof}
    Since private signals are conditional i.i.d.\ by Assumption~\ref{ass:1Mis}, a strong law of large numbers holds for the population mean beliefs. It follows that the second-order report $\alpha_i$ of each agent $i\in N$ is given by
\begin{equation*}
    \alpha_i=\sum_{\omega \in \Omega}\mu_i(\omega) \alpha _{i}^{\omega } =\sum_{\omega \in \Omega}\mu_i(\omega) \left(\bar{\mu}_{\omega_1}(\omega )+\zeta _{i}^{\omega } \right).
\end{equation*} 
For each group $N_k$, combining~\eqref{eq1:2ndordermisspec} with the display above yields
 \begin{equation*}
        \hat{\alpha}_{N_k}^n=\frac{1}{\displaystyle\sum_{i\leq n}x_{iN_k}}\sum_{i\leq n}x_{iN_k} \left( \sum_{\omega \in \Omega}\mu_i(\omega) \left(\bar{\mu}_{\omega_1}(\omega )+\zeta_{i}^{\omega } \right)\right)
    \end{equation*}
    which is equivalent to
\begin{equation*}
        \hat{\alpha}_{N_k}^n=\left(\sum_{\omega \in \Omega}\bar{\mu}_{\omega_1}(\omega )\left(\frac{1}{\displaystyle\sum_{i \leq n}x_{iN_k}}\sum_{i \leq  n}x_{iN_k}\mu_i(\omega) \right)\right)+\sum_{\omega \in \Omega}\left(\frac{1}{\displaystyle\sum_{i\leq n}x_{iN_k}}\sum_{i \leq n}x_{iN_k}\mu_i(\omega)\zeta_{i}^{\omega } \right)
    \end{equation*}
We proceed in two steps.
\begin{enumerate}
    \item The first summand is $\displaystyle\sum_{\omega \in \Omega}\hat{\mu}_{N_k}^n(\omega)\,\bar{\mu}_{\omega_1}(\omega )$. By Lemma~\ref{lem:misspec2}, $\hat{\mu}_{N_k}^n(\omega)$ converges almost surely to a deterministic limit for each $\omega$, so this summand converges almost surely to $\sum_{\omega\in\Omega}\mathbb{E}[\hat{\mu}_{N_k}(\omega)]\,\bar{\mu}_{\omega_1}(\omega)$.
    \item Consider the second summand
    \begin{equation*}
        \left(\frac{1}{\displaystyle\sum_{i\leq n}x_{iN_k}}\sum_{\omega \in \Omega}\sum_{i\leq n}x_{iN_k}\mu_i(\omega)\zeta_{i}^{\omega } \right)
    \end{equation*}
    which is equivalent to
    \begin{equation*}
        \left(\frac{1}{\frac1n\displaystyle\sum_{i\leq n}x_{iN_k}}\sum_{\omega \in \Omega}\frac1n \sum_{i\leq n}x_{iN_k}\mu_i(\omega)\zeta_{i}^{\omega } \right)
    \end{equation*}
    For each fixed state index $\bar\omega\in\Omega$, define
    \[
    L_n^k:=\frac1n\sum_{i\le n}x_{iN_k},
    \qquad
    K_{n}^{k,\bar\omega}:=\frac1n\sum_{i\le n}x_{iN_k}\mu_i(\bar\omega)\zeta_i^{\bar\omega}.
    \]
    By Assumptions~\ref{ass:misspecified} and~\ref{ass:1Mis}, conditional on the realized state, both sequences satisfy a strong law:
    \[
    L_n^k \xrightarrow{a.s.} L^k:=\mathbb E[x_{1N_k}]>0,
    \qquad
    K_n^{k,\bar\omega}\xrightarrow{a.s.}K^{k,\bar\omega}:=\mathbb E[x_{1N_k}\mu_1(\bar\omega)\zeta_1^{\bar\omega}].
    \]
    Using the conditional independence of $\zeta_i^{\bar\omega}$ and $\mu_i$ in Assumption~\ref{ass:misspecified},
    \[
    K^{k,\bar\omega}
    =\mathbb E[\zeta^{\bar\omega}]\,\mathbb E[x_{1N_k}\mu_1(\bar\omega)].
    \]
    Hence, by continuous mapping,
    \[
    \frac{K_n^{k,\bar\omega}}{L_n^k}
    \xrightarrow{a.s.}
    \mathbb E[\zeta^{\bar\omega}]
    \frac{\mathbb E[x_{1N_k}\mu_1(\bar\omega)]}{\mathbb E[x_{1N_k}]}
    =
    \mathbb E[\zeta^{\bar\omega}]\,\mathbb E[\hat\mu_{N_k}(\bar\omega)],
    \]
    where the last equality uses Lemma~\ref{lem:misspec2}. Summing over $\bar\omega\in\Omega$ gives
    \[
    \sum_{\omega\in\Omega}\left(\frac{1}{\sum_{i\le n}x_{iN_k}}\sum_{i\le n}x_{iN_k}\mu_i(\omega)\zeta_i^\omega\right)
    \xrightarrow{a.s.}
    \sum_{\omega\in\Omega}\mathbb E[\hat\mu_{N_k}(\omega)]\,\mathbb E[\zeta^\omega],
    \]
    concluding the proof.
    \end{enumerate} 
    \end{proof}

\newpage

\begin{center}
\vspace*{3cm}
{\LARGE\bfseries SUPPLEMENTARY ONLINE APPENDIX}
\vspace*{2cm}
\end{center}

\section{Bayesian Aggregation with Fixed Finite Populations}\label{sec:finite}
So far, we have discussed the possibility of aggregating information with an infinite or arbitrarily large population. In practice, of course, populations are finite. The procedures given above will correctly identify the true state with high probability in large populations. For instance, in our baseline model, even with a large but finite population, the average belief in the population will concentrate around the expected belief conditional on the true state. As long as the expected beliefs conditional on state are sufficiently different across the states, an appropriately modified procedure will recover the true state of the world with high probability.

Nevertheless, at a theoretical level, one may wonder whether, with a finite population, there is any value to eliciting higher-order beliefs, ignoring the difficulties of such elicitation in practice. To make this potentially valuable, we need a more exacting benchmark, since, as we argued above, PMBA will already aggregate information with appropriately high probability. The benchmark we use therefore is one of a ``full information posterior'', i.e., can we, without knowledge of the underlying information structure $P$, and eliciting solely agents' beliefs (and higher-order beliefs), nevertheless reach the same posterior beliefs as an omniscient agent who knew $P$ and directly observed all the agents' signals? Our previous results answered this in the positive and showed that for the case of an infinite population (under the maintained assumptions), the output of PMBA could identify the degenerate belief on the true state.

In this section, we answer these questions primarily in the negative. First we show that with a finite population, and with elicitation of the entire hierarchy of beliefs, an agnostic procedure can learn the prior and signals of each agent. This result is essentially a straightforward corollary of the results of~\cite{mertens2015repeated}. However, elicitation of the entire hierarchy is obviously impractical in real-world applications. 

By contrast, we show that, for elicitation up to any finite-order of beliefs, there is an identification problem: there exist information structures where the exact same finite hierarchy can be realized among agents in both states.

Consider a finite set of agents $\left\{ 1,2,\ldots,N\right\} $ with a common
prior $P$ defined on a finite set $\Omega \times S$ (where $S=\times
_{i=1}^{N}S_{i}$ is the set of their signal profiles). The aggregator does not
know $P$ but can ask the agents to report their higher-order beliefs. In this section, we show that the aggregator can effectively elicit the full information posterior, provided that they can ask the agents to report their entire hierarchy of beliefs and that each hierarchy of beliefs uniquely identifies a signal of an agent. 

To this end, we recall the standard formulation of higher-order beliefs by~\cite{mertens1985formulation}. Denote by $s_{i}^{k}$ the $k$th-order belief of agent $i$ over $\Omega$.
For instance, $s_{i}^{1}=$marg$_{\Omega }P \left( \cdot |s_{i}\right) $, 
\[s_{i}^{2}\left( \omega ,s_{-i}^{1}\right) =P \left( \left\{ \left( \omega
^{\prime },s_{-i}^{\prime }\right) :\omega ^{\prime }=\omega \text{ and }%
s_{-i}^{\prime 1}=s_{-i}^{1}\right\} |s_{i}\right) \text{,}
\]%
and so on. Denote by $\tilde{s}_{i}$ the hierarchy of beliefs induced by
signal $s_{i}$, i.e., $\tilde{s}_{i}=\left( s_{i}^{1},s_{i}^{2},\ldots\right) $%
. Let $\tilde{S}_{i}$ be the set of hierarchies of beliefs induced from $%
S_{i}$ and $\tilde{P}$ the distribution induced by $P$ on $\Omega
\times \tilde{S}$. Moreover, by~\cite{mertens1985formulation}, each $\tilde{s}_{i}
$ induces a belief $\tilde{\pi}_{i}\left( \tilde{s}_{i}\right) $ over $%
\Omega \times \tilde{S}_{-i}$, where $\tilde{S}_{-i}=\times_{j\neq i}\tilde{S}_{j}$.

\begin{restatable}{theorem}{thminfinite}\label{thm:infinite}
Let $\tilde{\mu}$ denote the induced probability measure on $\Omega \times \tilde{S}$. Suppose that $\tilde{\mu}$ has finite support and that the agents report $\tilde{s}=(\tilde{s}_i)_i$. Then there exists a procedure that recovers the ``pooled information'' posterior on the states, i.e., $\tilde{P}(\cdot \mid \tilde{s})$.
\end{restatable}

\begin{proof}
Assume that $\text{marg}_{\tilde{S}}\tilde{\mu}\left( \tilde{s}\right) >0$. It follows from Theorem
III.2.7 of~\cite{mertens2015repeated} that the aggregator can derive:
\begin{enumerate}
    \item  The set $E\left( \tilde{s}\right) $
which is the smallest set $Y\subseteq \Omega \times \tilde{S}$ of
state-hierarchy profiles satisfying
\begin{enumerate}
    \item  $\left( \omega ,\tilde{s}\right) \in
Y$ for some $\omega \in \Omega$, and,
    \item  for each $\left( \omega ,\tilde{t}\right) \in
Y $, we have $\left\{ \tilde{t}_{i}\right\} \times $supp $\tilde{\pi}%
_{i}\left( \tilde{t}_{i}\right) \subseteq Y$.
\end{enumerate}
\item The unique consistent
probability $\tilde{\mu}$ on $E\left( \tilde{s}\right) $.
\end{enumerate}
With (1) and (2),
the aggregator can compute $\tilde{\mu}\left( \cdot |\tilde{s}\right) $. We
recap and illustrate both (1) and (2) for the ease of reference.

To construct (1), define $C_{i}^{1}\left( \omega ,\tilde{t}\right) =\left\{ \tilde{t}%
_{i}\right\} \times $supp$\tilde{\pi}_{i}\left( \tilde{t}_{i}\right) $ for
each $\left( \omega ,\tilde{t}\right) \in \Omega \times \tilde{S}$; and
inductively, for every $l\geq 1$ and $\tilde{t}\in \tilde{S}$, define 
\[
C_{i}^{l+1}\left( \omega ,\tilde{t}\right) =C_{i}^{l}\left( \omega ,\tilde{t}%
\right) \bigcup \bigcup\limits_{\left( \omega ^{\prime },\tilde{t}^{\prime
}\right) \in C_{i}^{l}\left( \omega ,\tilde{t}\right)
}\bigcup\limits_{j=1}^{N}C_{j}^{1}\left( \omega ^{\prime },\tilde{t}^{\prime
}\right) \text{.} 
\]%
Then, let%
\[
C_{i}\left( \omega ,\tilde{t}\right) =\bigcup\limits_{l=1}^{\infty
}C_{i}^{l}\left( \omega ,\tilde{t}\right) \text{.} 
\]%
These include $\left( \omega ,\tilde{t}\right) $, the state-hierarchy
profiles which agent $i$ regards as possible (i.e., $C_{i}^{1}\left( \omega ,%
\tilde{t}\right) $) at $\left( \omega ,\tilde{t}\right) $, the
state-hierarchy profiles which some agent regards at some state-hierarchy
profile in $C_{i}^{1}\left( \omega ,\tilde{t}\right) $ (i.e., $%
C_{i}^{2}\left( \omega ,\tilde{t}\right) $), and so on. 

Consider any $\omega ^{\ast }$ such that $\tilde{\mu}\left( \omega ^{\ast },\tilde{s}%
\right) >0$. Note that $%
C_{i}\left( \omega ^{\ast },\tilde{s}\right) $ satisfies properties (a) and
(b) above by construction and hence $C_{i}\left( \omega ^{\ast },\tilde{s}%
\right) \supseteq E\left( \tilde{s}\right) $.\footnote{%
Property (b) is by construction and property (a) follows from the assumption
that $\tilde{\mu}\left( \omega ^{\ast },\tilde{s}\right) >0$.} Also we can
argue that $C_{i}\left( \omega^* ,\tilde{s}\right) \subseteq E\left( \tilde{s}%
\right) $ inductively.

For (2), it follows from Bayes' rule that for every $\left( \omega ^{\prime
},\tilde{t}^{\prime }\right) \in C_{i}^{1}\left( \omega ,\tilde{t}\right) $
and $\tilde{\mu}\left( \omega ,\tilde{t}\right) >0$, we must have $\tilde{t}%
_{i}=\tilde{t}_{i}^{\prime }$ and hence 
\[
\frac{\tilde{\mu}\left( \omega ^{\prime },\tilde{t}^{\prime }\right) }{%
\tilde{\mu}\left( \omega ,\tilde{t}\right) }=\frac{\tilde{\pi}_{i}\left( 
\tilde{t}_{i}^{\prime }\right) \left( \omega ^{\prime },\tilde{t}^{\prime
}\right) \tilde{\mu}_{i}\left( \tilde{t}_{i}^{\prime }\right) }{\tilde{\pi}%
_{i}\left( \tilde{t}_{i}\right) \left( \omega ,\tilde{t}_{-i}\right) \tilde{%
\mu}_{i}\left( \tilde{t}_{i}\right) }=\frac{\tilde{\pi}_{i}\left( \tilde{t}%
_{i}^{\prime }\right) \left( \omega ^{\prime },\tilde{t}^{\prime }\right) }{%
\tilde{\pi}_{i}\left( \tilde{t}_{i}\right) \left( \omega ,\tilde{t}%
_{-i}\right) }>0\text{.} 
\]%
The aggregator knows $\tilde{\pi}_{i}\left( \tilde{t}_{i}\right) $ and hence
can express $\tilde{\mu}\left( \omega ^{\prime },\tilde{t}^{\prime }\right) $
as a multiple of $\tilde{\mu}\left( \omega ,\tilde{t}\right) $. Inductively,
for every $\left( \omega ^{\prime },\tilde{t}^{\prime }\right) \in
C_{i}^{l}\left( \omega ,\tilde{t}\right) $, $\tilde{\mu}\left( \omega
^{\prime },\tilde{t}^{\prime }\right) $ can also be expressed as a multiple
of $\tilde{\mu}\left( \omega ,\tilde{t}\right) $. 

Hence, $\tilde{\mu}$ is uniquely determined on $E\left( \tilde{s}\right)=C_{i}\left( \omega^* ,\tilde{s}\right)$ since $\tilde{\mu}\left( \omega ^{\ast },\tilde{s}\right) >0$. $\tilde{\mu}(\cdot|\tilde{s})$ can thus be calculated.
\end{proof}

One might ask whether such a result can be obtained while only asking agents for their higher-order beliefs up to some finite order (i.e., as opposed to the full infinite hierarchy). It is possible to show that in general, without further assumptions, this is not the case: the result of Theorem~\ref{thm:infinite} cannot be achieved if the aggregator only knows the reported beliefs up to order $m$. Our argument, in the following section, is adapted from the leading example of~\cite{lipman2003finite}.

\subsection{Eliciting Higher-Order Beliefs up to a Finite Order \texorpdfstring{$m$}{m}}\label{sec:lipman}

We adopt his notation here for ease of comparison. Suppose there are two players and two states of nature $\left\{ \omega _{1},\omega _{2}\right\} $ as in our paper. The construction is easier to explain in terms of the standard partitional model of knowledge.

The model below considers $8$ extended states: $\{\left(
\sigma_{l},k\right): l \in \{1,2\} \text{ and } k \in \{1,2,3,4\}\}$. Interpret state $(\sigma_l,k)$ in this model as corresponding to $\omega _{l}$ realized for $l=1,2$
and any $k$. In this example, Lipman considers the hierarchy of belief
induced by common knowledge that player $1$ assigns probability $2/3$ to $%
\omega _{1}$ and player $2$ assigns probability $1/3$ to $\omega _{1}$. Such
a type does not admit a common prior. However, Lipman shows that the model below, which has a uniform common prior, admits a state $\left(
\sigma_{1},1\right) $ where the players have the same belief as that of the common knowledge type described above, up to any finite order $m$.

The prior for the players is the uniform distribution over the extended states, i.e.: 
\[
\begin{tabular}{@{}lcccccccc@{}}
\toprule
& $\left( \sigma_{1},4\right) $ & $\left( \sigma_{1},3\right) $ & $\left(
\sigma_{2},2\right) $ & $\left( \sigma_{1},1\right) $ & $\left( \sigma_{2},1\right) $ & $%
\left( \sigma_{1},2\right) $ & $\left( \sigma_{2},3\right) $ & $\left( \sigma_{2},4\right) $
\\
\midrule
$\mu $ & $\frac{1}{8}$ & $\frac{1}{8}$ & $\frac{1}{8}$ & $\frac{1}{8}$ & $%
\frac{1}{8}$ & $\frac{1}{8}$ & $\frac{1}{8}$ & $\frac{1}{8}$ \\
\bottomrule
\end{tabular}%
\]
Agents' information is identified with the following partitions:
\begin{eqnarray*}
\Pi _{1} &=&\left\{ \left\{ \left( \sigma_{1},4\right) ,\left( \sigma_{1},3\right)
,\left( \sigma_{2},2\right) \right\} \left\{ \left( \sigma_{1},1\right) ,\left(
\sigma_{2},1\right) ,\left( \sigma_{1},2\right) \right\} \text{ }\left\{ \left(
\sigma_{2},3\right) ,\left( \sigma_{2},4\right) \right\} \right\} \\
\Pi _{2} &=&\left\{ \left\{ \left( \sigma_{1},4\right) ,\left( \sigma_{1},3\right)
\right\} \left\{ \left( \sigma_{2},2\right) ,\left( \sigma_{1},1\right) ,\left(
\sigma_{2},1\right) \right\} \left\{ \left( \sigma_{1},2\right) ,\left( \sigma_{2},3\right)
,\left( \sigma_{2},4\right) \right\} \right\} \text{.}
\end{eqnarray*}%

Note that each $\Pi _{i}\left( \sigma',k\right) $ is identified with a signal of
player $i$. Observe that player 1 assigns probability one to $\omega _{2}$
at $\Pi _{1}\left( \sigma_{2},4\right) $ and player $2$ assigns probability one
to $\omega _{1}$ at $\Pi _{2}\left( \sigma_{1},4\right) $. Hence, each partition
cell of each player induces a different second-order belief. In particular, $%
\theta:=\left( \Pi _{1}\left( \sigma_{1},1\right) ,\Pi _{2}\left( \sigma_{1},1\right) \right) $
is the only partition profile at which the second-order beliefs of both players are identical to $t$. Conditional on the reported second-order belief at $\theta$ (denoted $\theta^2$), we have
\[
\mu \left( \omega _{1}|\theta^2 \right) =\frac{1}{2}\text{.} 
\]

Now consider another model modified from the previous one by adding one additional state with the following prior:
\[
\begin{tabular}{@{}lccccccccc@{}}
\toprule
& $\left( \sigma_{1},4\right) ^{\prime }$ & $\left( \sigma_{1},3\right) ^{\prime }$ & $%
\left( \sigma_{2},2\right) ^{\prime }$ & $\left( \sigma_{1},1\right) ^{\prime }$ & $%
\left( \sigma_{1},1\right) $ & $\left( \sigma_{2},1\right) $ & $\left( \sigma_{1},2\right) $
& $\left( \sigma_{2},3\right) $ & $\left( \sigma_{2},4\right) $ \\
\midrule
$\mu ^{\prime }$ & $\frac{1}{20}$ & $\frac{1}{20}$ & $\frac{1}{10}$ & $\frac{%
1}{10}$ & $0$ & $\frac{1}{10}$ & $\frac{1}{5}$ & $\frac{1}{5}$ & $\frac{1}{5}
$ \\
\bottomrule
\end{tabular}%
\]%
where $^{\prime }$ is to indicate that a state is to the left of $\left(
\sigma_{1},1\right) $ and it will be useful in generalizing the idea to eliciting $m$ orders of beliefs, for $m\geq 3$ in what follows. Agents' information is now given by the partitions:
\begin{eqnarray*}
\Pi _{1}^{\prime } &=&\left\{ \left\{ \left( \sigma_{1},4\right) ^{\prime
},\left( \sigma_{1},3\right) ^{\prime },\left( \sigma_{2},2\right) ^{\prime },\left(
\sigma_{1},1\right) ^{\prime }\right\} ,\left\{ \left( \sigma_{1},1\right) ,\left(
\sigma_{2},1\right) ,\left( \sigma_{1},2\right) \right\} \text{ }\left\{ \left(
\sigma_{2},3\right) ,\left( \sigma_{2},4\right) \right\} \right\} \text{;} \\
\Pi _{2}^{\prime } &=&\left\{ \left\{ \left( \sigma_{1},4\right) ^{\prime
},\left( \sigma_{1},3\right) ^{\prime }\right\} ,\left\{ \left( \sigma_{2},2\right)
^{\prime },\left( \sigma_{1},1\right) ^{\prime },\left( \sigma_{1},1\right) ,\left(
\sigma_{2},1\right) \right\} ,\left\{ \left( \sigma_{1},2\right) ,\left(
\sigma_{2},3\right) ,\left( \sigma_{2},4\right) \right\} \right\} \text{.}
\end{eqnarray*}%

Define $\theta' := \left( \Pi
_{1}^{\prime }\left( \sigma_{1},1\right) ,\Pi _{2}^{\prime }\left( \sigma_{1},1\right)
\right)$. Conditional on the reported second-order belief at $\theta'$ (denoted $(\theta')^2$), we have
\[
\mu ^{\prime }\left( \omega _{1}|(\theta')^2\right)
=0\text{.}
\]%
Second, observe that 
\[
(\theta')^{2}=\theta^{2}\text{.}
\]
This means that the reported second-order beliefs at $\theta'$ under $\mu^{\prime}$ are the same as at $\theta$ under $\mu$. The
basic idea is that in $\mu ^{\prime }$, we \textquotedblleft
shift\textquotedblright\ the probability which $\mu $ assigns to $\left(
\sigma_{1},1\right) $ to the additional state $\left( \sigma_{1},1\right) ^{\prime }$.
This additional state helps us preserve the first-order belief of player $2$
at $\Pi _{2}^{\prime }\left( \sigma_{1},1\right) $ at $1/3$, while we
decrease the probability $\left( \sigma_{1},4\right) $ to $0$ to preserve the
first-order belief of player $1$ at $\Pi _{1}^{\prime }\left( \left(
\sigma_{1},4\right) ^{\prime }\right) $. This takes care of the states to the
left of $\left( \sigma_{1},1\right) $. For states to the right, we double the
probability of $\left( \sigma_{2},2\right) $, $\left( \sigma_{2},3\right) $, and $%
\left( \sigma_{2},4\right) $, so that the first-order belief of player $1$ at $\Pi
_{1}^{\prime }\left( \sigma_{1},1\right) $ and the first-order belief of player $2
$ at $\Pi_{2}^{\prime }\left( \sigma_{1},1\right) $ are also preserved. Therefore, just eliciting the first two orders of beliefs, we cannot distinguish between the model corresponding to $\mu$ (under which the posterior would be that both states are equally likely), and $\mu'$ (under which the posterior would be that the state is $\omega_2$ for sure). 

The construction for $m>2$ is similar but more involved, see below.

\subsection{Construction for \texorpdfstring{$m>2$}{m>2}}\label{app:higher-order}

Similarly, we can generalize the construction in Section~\ref{sec:lipman} to the case with $m\geq
3 $ as follows: 
\begingroup\small
\begin{eqnarray*}
\Pi _{1}=\left\{ \left\{ \left\{ \left( \sigma_{1},2k-1\right) ,\left(
\sigma_{1},2k\right) ,\left( \sigma_{2},k\right) \right\} :k=1,\ldots,2^{m-1}\right\}
,\left\{ \left( \sigma_{2},2^{m-1}+1\right) ,\ldots,\left( \sigma_{2},2^{m}\right)
\right\} \right\} \text{;} \\
\Pi _{2}=\left\{ \left\{ \left\{ \left( \sigma_{2},2k-1\right) ,\left(
\sigma_{2},2k\right) ,\left( \sigma_{1},k\right) \right\} :k=1,\ldots,2^{m-1}\right\}
,\left\{ \left( \sigma_{1},2^{m-1}+1\right) ,\ldots,\left( \sigma_{1},2^{m}\right)
\right\} \right\} \text{.}
\end{eqnarray*}%
\[
\mu \left( \sigma_{l},k\right) =\frac{1}{2^{m+1}},\forall l=1,2,\forall
k=1,2,\ldots,2^{m}\text{.} 
\]%
Without loss of generality, assume that $m\geq 3$ is odd. Construct the new
model as follows: 
\[
\Pi'_{1}=\left\{ 
\begin{array}{c}
\left\{ \left( \sigma_{1},1\right) ,\left( \sigma_{2},1\right) ,\left( \sigma_{1},2\right)
\right\} ,\left\{ \left( \sigma_{1},1\right) ^{\prime },\left( \sigma_{2},2\right)
^{\prime },\left( \sigma_{1},3\right) ^{\prime },\left( \sigma_{1},4\right) ^{\prime
}\right\} , \\ 
\left\{ \left\{ \left( \sigma_{1},2k-1\right) ^{\prime },\left( \sigma_{1},2k\right)
^{\prime },\left( \sigma_{2},k\right) ^{\prime }\right\} :k=2^{n-1}+1,\ldots,2^{n}%
\text{ and }n=3,5,\ldots,m-2\right\} , \\ 
\left\{ \left\{ \left( \sigma_{1},2k-1\right) ,\left( \sigma_{1},2k\right) ,\left(
\sigma_{2},k\right) \right\} :k=2^{n-1}+1,\ldots,2^{n}\text{ and }%
n=2,4,\ldots,m-1\right\} , \\ 
\left\{ \left( \sigma_{2},2^{m-1}+1\right)' ,\ldots,\left( \sigma_{2},2^{m}\right)'
\right\}%
\end{array}%
\right\} \text{;} 
\]%
\[
\Pi' _{2}=\left\{ 
\begin{array}{c}
\left\{ \left( \sigma_{1},1\right) ,\left( \sigma_{2},1\right) ,\left( \sigma_{1},1\right)
^{\prime },\left( \sigma_{2},2\right) ^{\prime }\right\} , \\ 
\left\{ \left\{ \left( \sigma_{2},2k-1\right) ^{\prime },\left( \sigma_{2},2k\right)
^{\prime },\left( \sigma_{1},k\right) ^{\prime }\right\} :k=2^{n-1}+1,\ldots,2^{n}%
\text{ and }n=2,4,\ldots,m-1\right\} , \\ 
\left\{ \left\{ \left( \sigma_{2},2k-1\right) ,\left( \sigma_{2},2k\right) ,\left(
\sigma_{1},k\right) \right\} :k=2^{n-1}+1,\ldots,2^{n}\text{ and }%
n=1,3,\ldots,m-2\right\} , \\ 
\left\{ \left( \sigma_{1},2^{m-1}+1\right),\ldots,\left(
\sigma_{1},2^{m}\right) \right\}%
\end{array}%
\right\} \text{.} 
\]
\endgroup

As in the case of $m=2$, we start from the partition cells $\left\{ \left(
\sigma_{1},1\right) ,\left( \sigma_{2},1\right) ,\left( \sigma_{1},2\right) \right\} $ and $%
\left\{ \left( \sigma_{1},1\right) ,\left( \sigma_{2},1\right) ,\left( \sigma_{1},1\right)
^{\prime },\left( \sigma_{2},2\right) ^{\prime }\right\} $ which contain $\left(
\sigma_{1},1\right) $. The states with $^{\prime }$ are those \textquotedblleft
to the left\textquotedblright\ of $\left( \sigma_{1},1\right) $, whereas the
states without $^{\prime }$ are like those \textquotedblleft to the
right\textquotedblright\ of $\left( \sigma_{1},1\right) $. We can then mimic the
idea for $m=2$ to solve for a prior $\mu ^{\prime }$ with the desired properties as follows:
\begin{enumerate}
\item Again, for $x>0$, set 
\[
\mu ^{\prime }\left( \sigma_{1},1\right) =0\text{ and }\mu ^{\prime }\left(
\sigma_{2},1\right) =\mu ^{\prime }\left( \sigma_{1},1\right) ^{\prime }= \mu
^{\prime }\left( \sigma_{2},2\right) ^{\prime }=x\text{.} 
\]

\item The number of states with $^{\prime }$, excluding $\left(
\sigma_{1},1\right) ^{\prime }$ and $\left( \sigma_{2},2\right) ^{\prime }$, is:
\begin{align*}
&y \equiv 3\times \left( 2^{1}+2^{3}+\cdots +2^{m-2}\right),\\    
\text{i.e., }& y = 2^m -2.
\end{align*}

We assign probability $\tfrac{x}{2}$ to each of these \textquotedblleft
left\textquotedblright\ states.

\item The number of states without $^{\prime }$ excluding $\left(
\sigma_{1},1\right) $ and $\left( \sigma_{2},1\right) $ (i.e., $\left( \sigma_{2},2\right) $
and states in $\Pi _{1}\left( \sigma_{2},k\right) $ where $k=2^{n-1}+1,\ldots,2^{n}$
and $n=2,4,\ldots,m-1$) is: 
\[
1+3\times \left( 2^{1}+2^{3}+\cdots +2^{m-2}\right)= y +1 \text{.} 
\]
We assign probability $2x$ to each of the \textquotedblleft
right\textquotedblright\ states.

\item Hence for the total probability to sum to $1$, we must have:
\begin{align*}
&3x + y \frac{x}{2} + (y+1) 2x = 1,\\
\implies & x = \frac{2}{10 + 5y} ,\\
\implies &x=\frac{1}{5\times2^{m}}\text{.} 
\end{align*}
\end{enumerate}
 In summary, observe that as desired, we have that
\begin{eqnarray*}
\mu \left( \omega _{1}|\left( \Pi _{1}\left( \sigma_{1},1\right) ,\Pi _{2}\left(
\sigma_{1},1\right) \right) ^{m}\right) &=&\frac{1}{2}\text{.} \\
\mu ^{\prime }\left( \omega _{1}|\left( \Pi _{1}^{\prime }\left(
\sigma_{1},1\right) ,\Pi _{2}^{\prime }\left( \sigma_{1},1\right) \right) ^{m}\right)
&=&0\text{.} \\
\left( \Pi _{1}^{\prime }\left( \sigma_{1},1\right) ,\Pi _{2}^{\prime }\left(
\sigma_{1},1\right) \right) ^{m} &=&\left( \Pi _{1}\left( \sigma_{1},1\right) ,\Pi
_{2}\left( \sigma_{1},1\right) \right) ^{m}\text{.}
\end{eqnarray*}

Clearly, we can flip left and right to construct another model\ $\Pi
_{1}^{\prime \prime }$ and $\Pi _{2}^{\prime \prime }$ with 
\begin{eqnarray*}
\mu ^{\prime }\left( \omega _{1}|\left( \Pi _{1}^{\prime \prime }\left(
\sigma_{1},1\right) ,\Pi _{2}^{\prime \prime }\left( \sigma_{1},1\right) \right)
^{m}\right) &=&1\text{;} \\
\left( \Pi _{1}^{\prime \prime }\left( \sigma_{1},1\right) ,\Pi _{2}^{\prime
\prime }\left( \sigma_{1},1\right) \right) ^{m} &=&\left( \Pi _{1}\left(
\sigma_{1},1\right) ,\Pi _{2}\left( \sigma_{1},1\right) \right) ^{m}\text{.}
\end{eqnarray*}

The construction here makes use of the feature that all higher-order
beliefs of the common knowledge hierarchy $t$ are degenerate. More
precisely, this feature ensures that in \textquotedblleft
shifting\textquotedblright\ out the probability which $\mu $ assigns to $%
\left( \sigma_{1},1\right) $, we can preserve the higher-order beliefs of $\Pi
_{1}\left( \sigma_{1},1\right) $ and $\Pi _{2}\left( \sigma_{1},1\right) $ as long as
we can preserve the first-order belief at every partition cell except for $%
\Pi _{1}\left( \sigma_{2},2^{m}\right) $ and $\Pi _{2}\left( \sigma_{1},2^{m}\right)
^{\prime }$. 

\section{Counterexamples to SP/ SC for multiple states}\label{app:counterex}

\begin{example}\label{ex}
Suppose there are three states $\Omega = \{\omega_1, \omega_2, \omega_3\}$, and three signals $S = \{s_1, s_2, s_3\}$. All agents have the initial uniform prior. Agents have a uniform prior over signals and receive conditionally i.i.d.\ signals, conditional on the state, so that their posteriors can be described as:
\begin{align*}
    \begin{array}{c|ccc}
         & P(\omega_1|\cdot) & P(\omega_2|\cdot) & P(\omega_3|\cdot) \\\hline
        s_1  & 0.4 & 0.21 & 0.39 \\
        s_2 & 0.45 & 0.54 & 0.01 \\
        s_3 & 0.44 & 0.06 & 0.5
    \end{array}
\end{align*}
i.e., the $i^{\text{th}}$ row and $j^{\text{th}}$ column is the posterior the agent places on $\omega_j$ upon signal $s_i$, $P(\omega_j\mid s_i)$. Note that this violates the assumption of ``diagonal dominance''---$P(\omega_1\mid s_2) > P(\omega_1\mid s_1)$, i.e., an agent who sees $s_2$ (or, indeed, $s_3$) places a higher posterior on $\omega_1$ than an agent who sees signal $s_1$. However, note that an agent who sees signal $s_i$ believes that $\omega_i$ is the most likely state.

To be clear, note that this can be achieved by the following distribution of signals conditional on the state (numbers rounded):
 \begin{align*}
     \begin{array}{c|ccc}
         & P(\cdot|\omega_1) & P(\cdot|\omega_2) & P(\cdot|\omega_3) \\\hline
        s_1  & 0.31 & 0.259 & 0.433 \\
        s_2 & 0.349 & 0.667 & 0.011 \\
        s_3 & 0.341 & 0.074 & 0.556
    \end{array}
 \end{align*}   
This results, by mechanical calculation, in the following average beliefs in the population:
 \begin{align*}
     \begin{array}{c|ccc}
         & \bar{\mu}(\omega_1) & \bar{\mu}(\omega_2) & \bar{\mu}(\omega_3) \\\hline
        \omega_1  & 0.431 & 0.436 & 0.422 \\
        \omega_2 & 0.274 & 0.419 & 0.130 \\
        \omega_3 & 0.295 & 0.145 & 0.447
    \end{array}
 \end{align*}  
Here the $i^\text{th}$ row and $j^\text{th}$ column represent the average belief in the population that the state is $\omega_i$ when the true state is $\omega_j$.

Note that the resulting expected population average belief is:
\begin{align*}
     \begin{array}{c|ccc}
         & \alpha(s_1) & \alpha(s_2) & \alpha(s_3) \\\hline
        \omega_1  & 0.429 & 0.434 & 0.427 \\
        \omega_2 & 0.248 & 0.351 & 0.211 \\
        \omega_3 & 0.323 & 0.215 & 0.362
    \end{array}
 \end{align*}  
i.e., the column labeled $\alpha(s_i)$ is the expected population average belief of an agent seeing signal $s_i$.

More generally, we can summarize the set of ``surprisingly popular'' state(s) in each case as:
\begin{align*}
    \begin{array}{c|ccc}
         & \omega_1 & \omega_2 & \omega_3 \\\hline
        \alpha(s_1)  & \omega_1, \underline{\omega_2} & \omega_1, \underline{\omega_2} & \omega_3 \\
        \alpha(s_2) & \omega_3 & \omega_1, \underline{\omega_2} & \omega_3 \\
        \alpha(s_3) & \omega_1, \underline{\omega_2} & \omega_1, \underline{\omega_2} & \omega_3
    \end{array}
\end{align*}
i.e., the entry in row $i$ corresponding to $\alpha(s_i)$ and column $j$ corresponding to state $\omega_j$ is the set of states that are surprisingly popular, in true state $\omega_j$, relative to the expected population average beliefs of an agent who received signal $s_i$. The underlined entry, if there are multiple, is the one that is surprising by the largest magnitude. For example, relative to an agent seeing $s_1$ in true state $\omega_1$, both $\omega_1$ and $\omega_2$ are surprisingly popular, but $\omega_2$ is the most surprisingly popular: $(0.274 - 0.248 = 0.028 > 0.02 = 0.431 - 0.429)$.

Note that in our example, the ``most surprisingly popular'' procedure fails to identify the true state when the state is $\omega_1$, regardless of the signal of who is polled for their expectation about the population average beliefs. Indeed, if an agent who saw signal $s_2$ is polled, the true state is not even in the set of surprisingly popular states. Note also the failure of first-order stochastic dominance of population average beliefs in this 3-state example---as we argued earlier, that was key to why SP works with binary states, but this can be violated with three or more states. 

Another closely related approach proposed is that of ``prediction normalized votes'' (see Section~1.3 of the Supplementary Appendix of~\cite{prelec2017solution}): each agent votes for the state they believe is more likely (which, in this example, was the signal they saw, i.e., an agent votes for state $\omega_i$ if they see signal $s_i$). The fraction of votes each state $\omega_j$ receives is normalized by the sum, over all states $\omega_k$, of the ratio of predicted vote fraction for $\omega_k$ by an $\omega_j$ voter to the predicted vote fraction for $\omega_j$ by an $\omega_k$ voter. A simple calculation shows that, given the specific numerical assumption we constructed, $\omega_2$ will have a higher prediction-normalized vote than $\omega_1$ when the true state is $\omega_1$.
\end{example}

\section{More Details about the Experiment}\label{app:experiment}

\subsection{Data and Code Availability}
The paper source files, the raw data used in the experiment, and the full open-source replication code for all reported empirical results are available at \url{https://github.com/malleshpai/robust-aggregation}. In particular, the repository includes the manuscript (\texttt{paper/}), raw data files (\texttt{Replication/data/raw/}), and the replication package with scripts, tests, and generated-table pipeline (\texttt{Replication/open/}).

\subsection{Consent and Ethical Approval}
The studies were all approved by the Department Ethics Review Committee (DERC) of NUS Department of Economics. For all studies, informed consent was obtained from the respondents using a consent form approved by DERC at the beginning of the online survey. Eligibility criteria required participants to be at least 18 years of age. Participants were informed of their rights to voluntary participation, confidentiality and anonymization of their responses, and minimal risks associated with participation.

\subsection{Detailed Survey Design}

\subsubsection{Knowledge Test}
Prior to eliciting participants' beliefs about NFT prices, we designed a brief knowledge assessment comprising six multiple-choice questions. The questions were selected according to two criteria: (1) they do not contain or imply any information about specific NFT prices or the price distribution in the NFT market, and (2) they only cover general concepts accessible to laypersons, avoiding any technical content that would require domain expertise. The full set of questions is provided below.

\textbf{Knowledge Assessment Questions:}

[1/6] What is ``minting'' an NFT?
\begin{enumerate}
    \item Save it to your phone
    \item \textit{Creating it and recording it on the blockchain}
    \item Uploading it to a gallery
    \item Selling it to the highest bidder
\end{enumerate}

[2/6] What is a common platform used for NFT trading and collecting?
\begin{enumerate}
    \item Google Photos
    \item Dropbox
    \item \textit{OpenSea}
    \item MetaMask
\end{enumerate}

[3/6] What are ``royalties'' in the NFT world?
\begin{enumerate}
    \item A reward for being an early buyer
    \item \textit{A way for artists to earn each time their NFT is resold}
    \item A fee paid to the blockchain network
    \item A code used to unlock NFTs
\end{enumerate}

[4/6] Which of the following best describes a ``utility NFT''?
\begin{enumerate}
    \item An NFT with no image or data
    \item An NFT used for mining cryptocurrency
    \item \textit{An NFT that provides additional functions, like access or benefits}
    \item An NFT that expires after 30 days
\end{enumerate}

[5/6] What is a ``PFP'' NFT often used for?
\begin{enumerate}
    \item \textit{Profile pictures and online identity}
    \item Online payments and banking
    \item Legal contracts
    \item Video streaming
\end{enumerate}

[6/6] Which of the following assets is fungible, unlike NFTs?
\begin{enumerate}
    \item A concert ticket
    \item A piece of digital art
    \item \textit{1 ETH (Ether)}
    \item A gaming skin with a unique serial number
\end{enumerate}

\subsubsection{Pedagogical Example}
Before participants begin evaluating NFTs, we present a brief pedagogical example to clarify the belief elicitation process. This is intended to mitigate potential ``learning'' effects. In the example, participants are guided through a simple scenario involving beliefs about ``\textit{Rain}'' versus ``\textit{No Rain}''. We then illustrate how the group average belief would be computed in the case where the population consists of two belief types. This example is intentionally designed to be unrelated to NFTs or price concepts, in order to avoid unintended priming effects or conceptual interference.

\begin{figure}[H]
\centering
\caption{Pedagogical Example on Self Belief}
\includegraphics[width=0.9\linewidth]{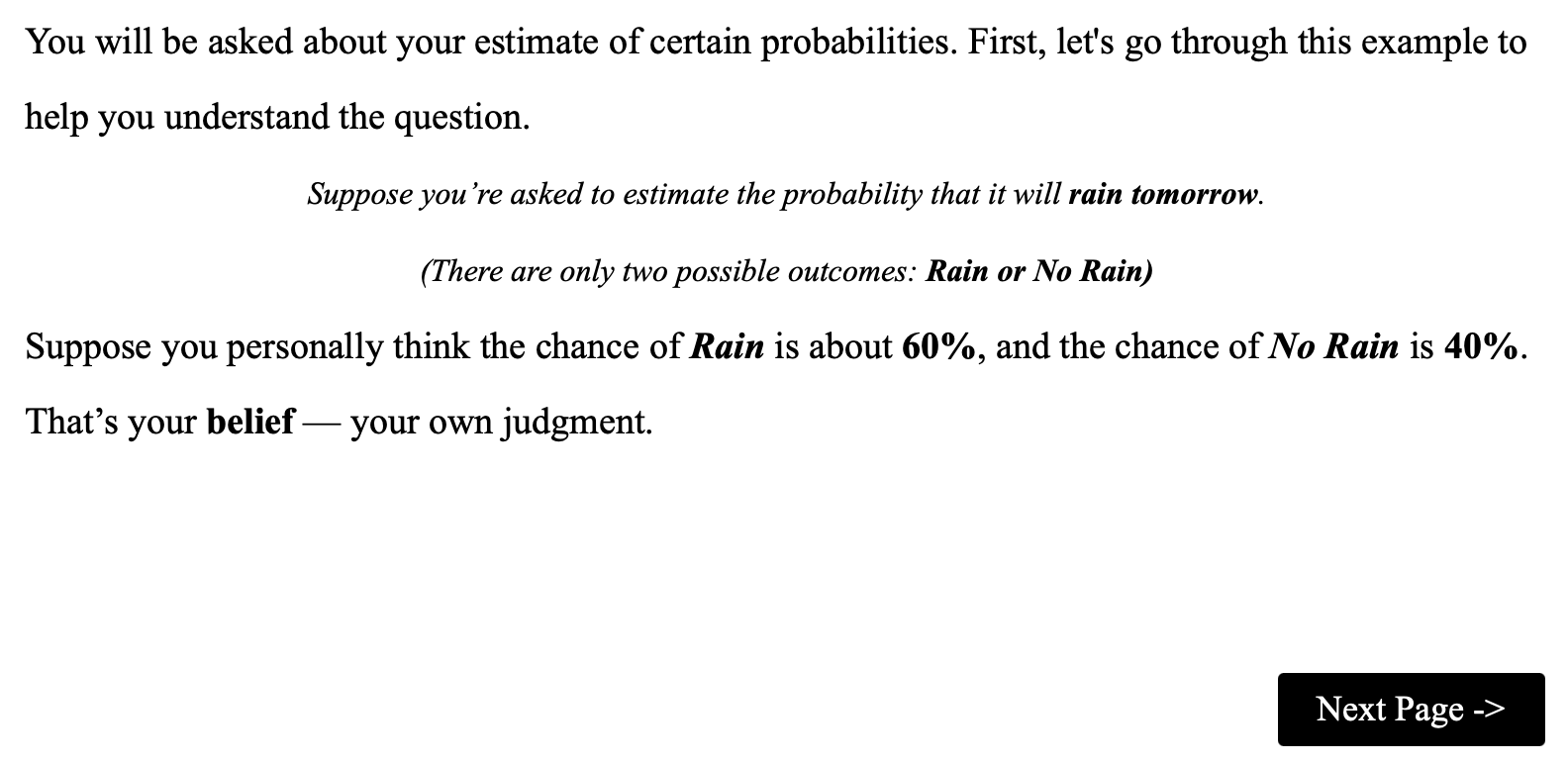} 
\includegraphics[width=0.9\linewidth]{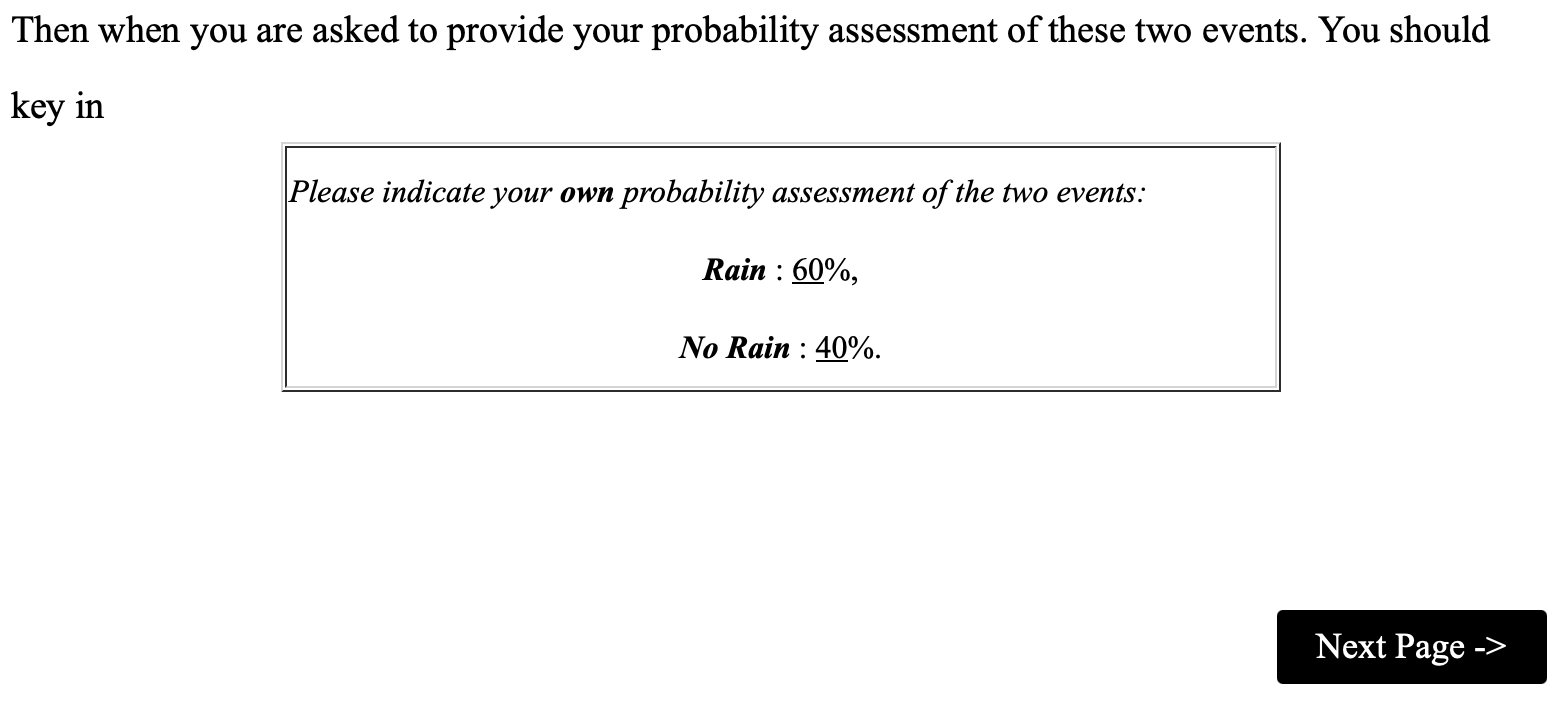}
\end{figure}

\begin{figure}[H]
\centering
\caption{Pedagogical Example on Group Average Belief}

\vspace{0.5em}
\includegraphics[width=0.9\linewidth]{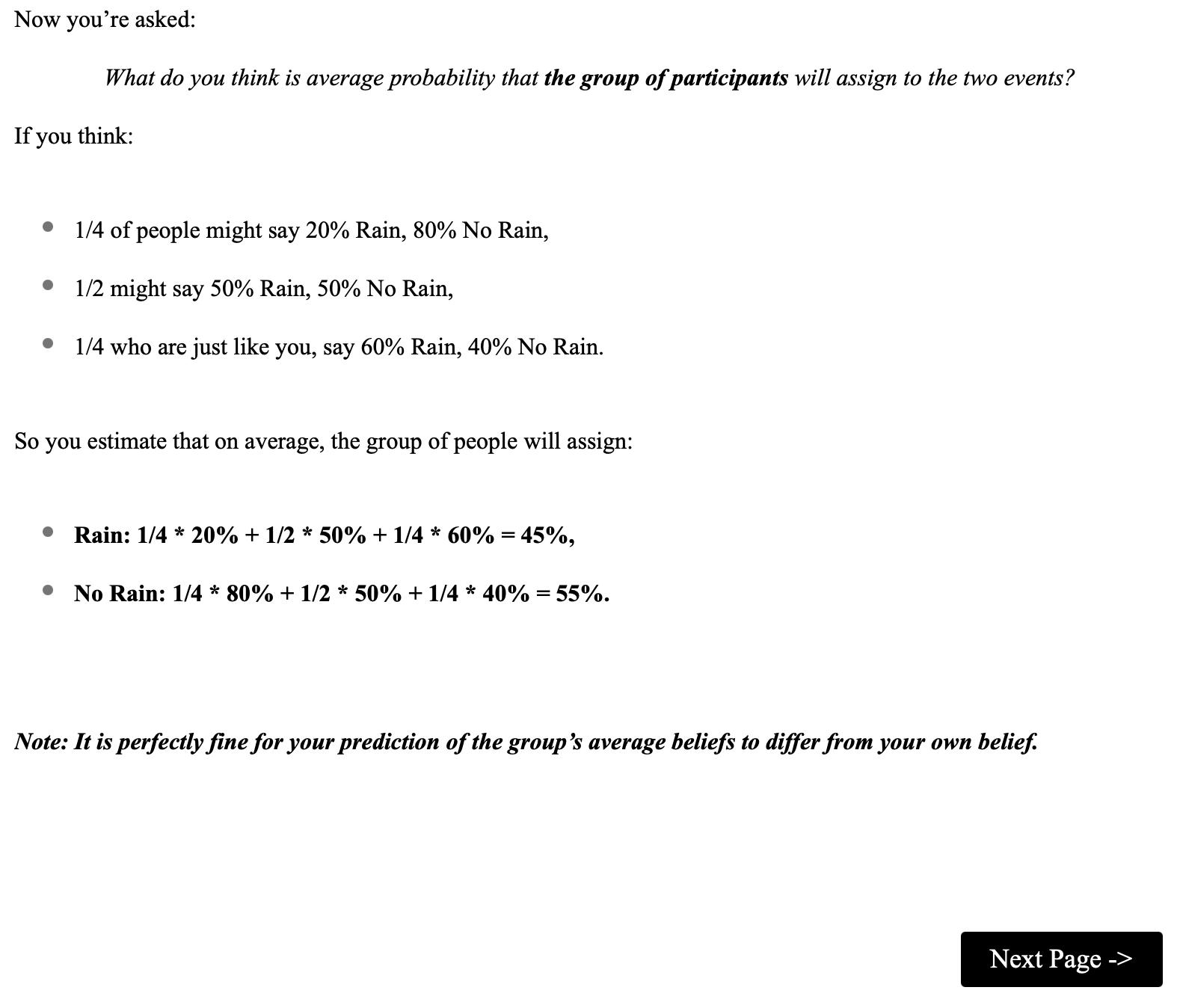} 
\includegraphics[width=0.9\linewidth]{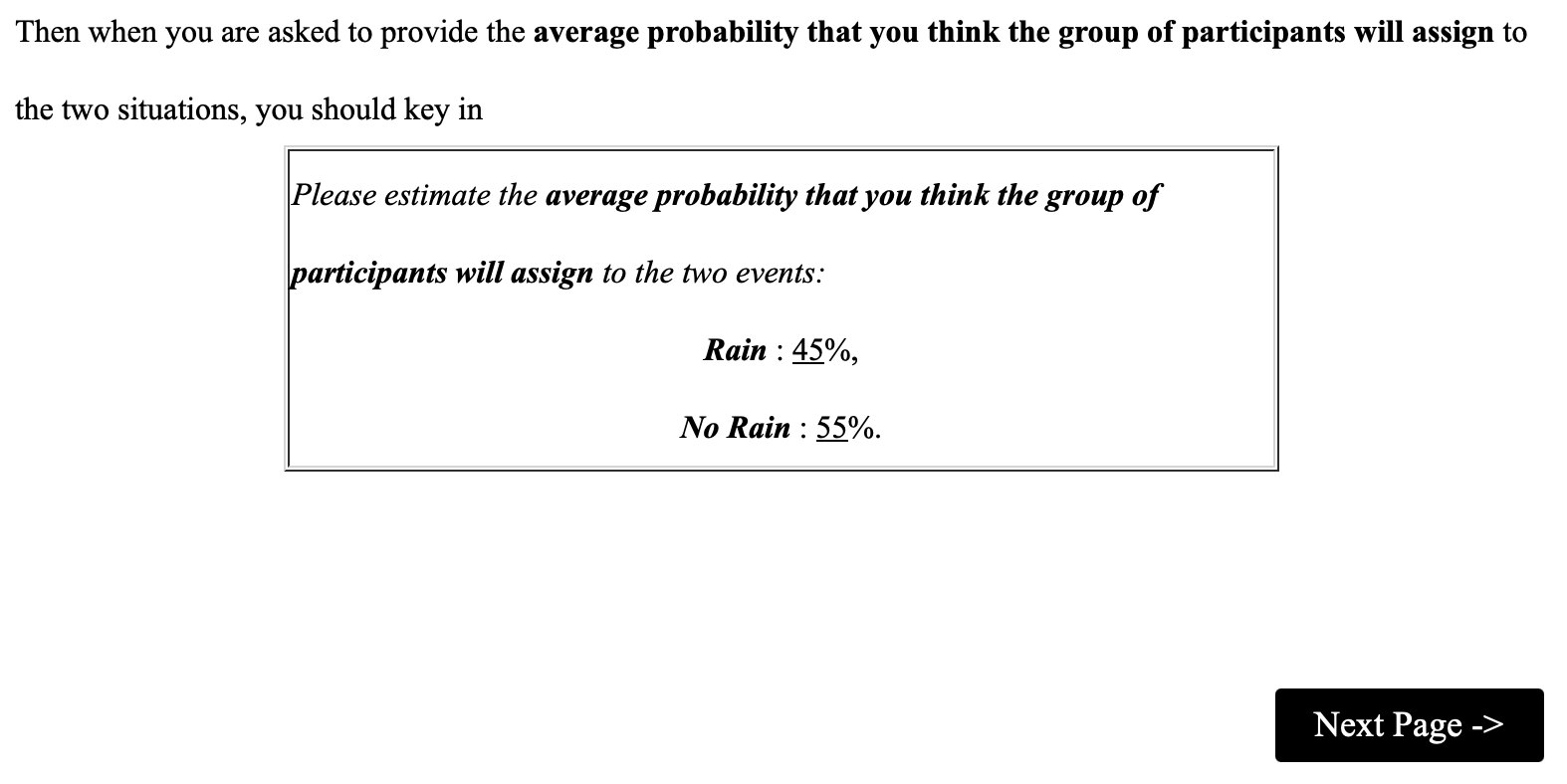}
\end{figure}

\subsubsection{Belief Elicitation Interface}
In the main task, upon presentation of each NFT image, the belief elicitation interface first displays the self-belief question immediately below the image. Participants may freely enter numeric values into the four input fields corresponding to the predefined price ranges. The sum of these inputs is automatically computed and displayed in the ``Total'' row. When participants click the ``Next'' button to proceed to the subsequent question on group belief, the system verifies whether the inputs sum to exactly 100. If not, a warning message is triggered, and participants are required to revise their entries before continuing.

After participants successfully submit their self-belief responses, the NFT image remains on screen while the question below transitions to ask beliefs about the group average. To reduce the risk of inattention to this change, the group belief question is visually highlighted, and a brief explanation is provided to clarify that one's belief about the group may differ from one's own belief.

\begin{figure}[H]
\centering
\caption{Self Belief Elicitation Question} 

\includegraphics[width=0.9\linewidth]{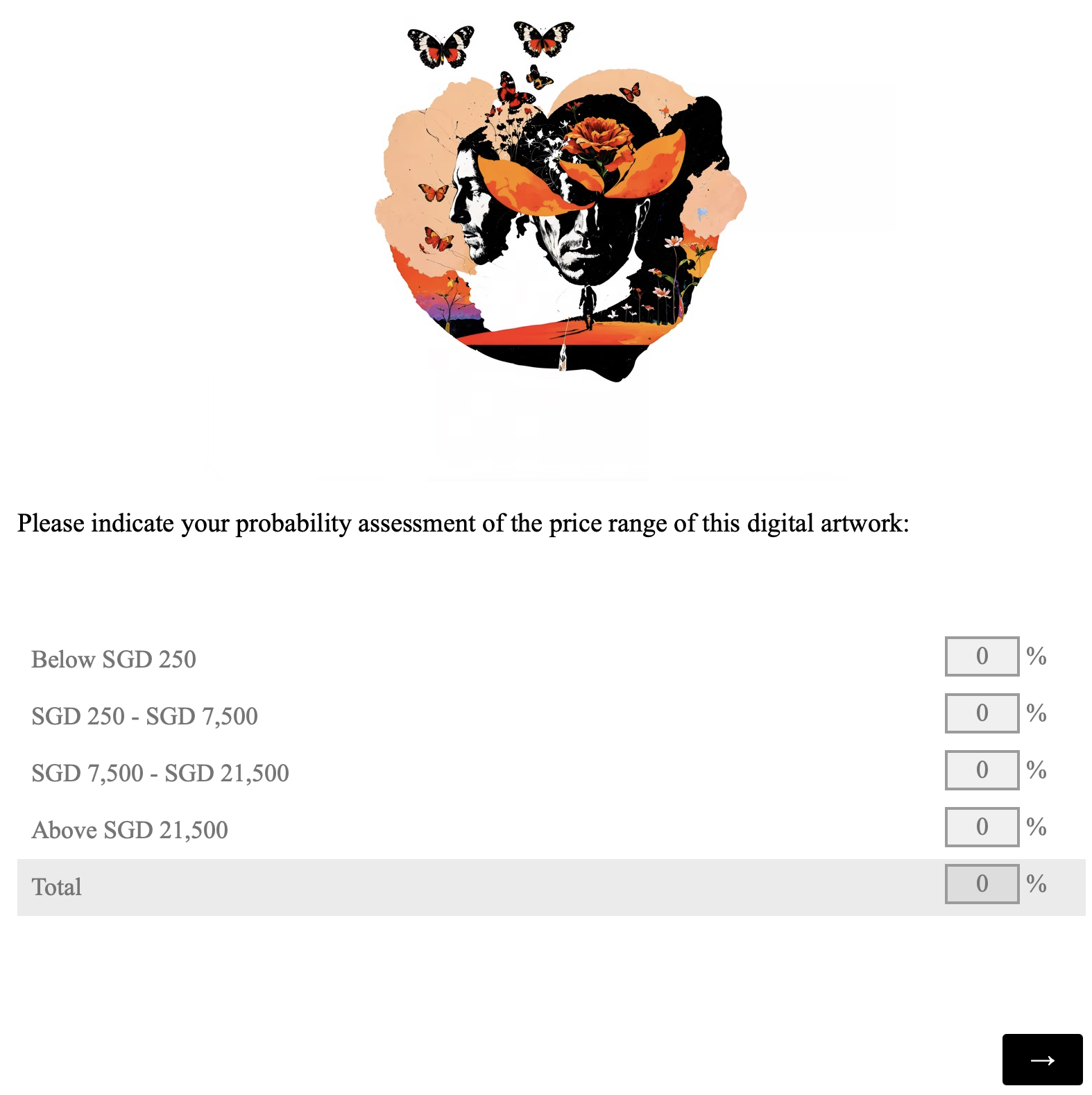}
\end{figure}
\begin{figure}[H]
\centering
\caption{Group Belief Elicitation Question} 

\includegraphics[width=0.9\linewidth]{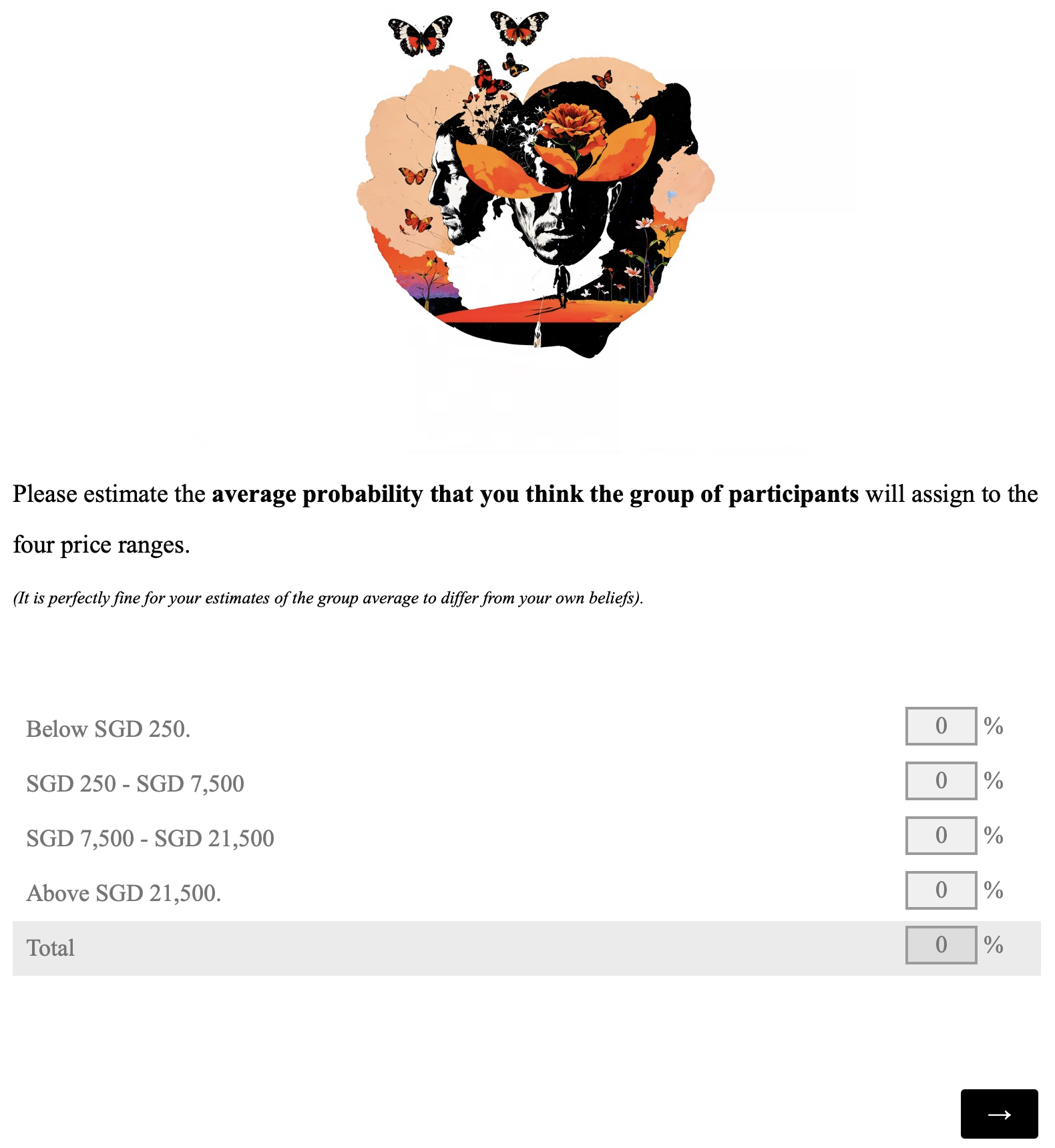}
\end{figure}

\subsection{Robustness Checks}

This section presents alternative specifications to verify the robustness of our main PMBA results. Table~\ref{tab:rob1} compares PMBA predictions across different implementation methods, including specifications with and without instrumental variables.

\subsubsection{Alternative Specifications}
Our baseline specification, reported in the main text, is estimated by OLS-IV with an intercept in both stages. As robustness checks, we vary (i) with/without intercept, (ii) unconstrained (OLS) vs.\ component-wise constrained (CLS), and (iii) IV vs.\ no-IV implementations. These exercises test whether the qualitative findings are sensitive to alternative structural assumptions.

First, omitting the intercept addresses the perfect collinearity implied by the fact that the four independent variables $\widehat{\mu}^i_k$ sum to one. This specification is substantively appropriate and plausible if we do not expect systematic upward or downward bias in second-order belief reports.

Second, to ensure structural interpretability of the coefficients as conditional probability beliefs, we impose the restriction $\bar{\mu}^{l}_k \in [0,1]$. For the no-intercept version, this is written as
\[
\alpha^l_i = \sum_{k=1}^4 \bar{\mu}^{l}_k \widehat{\mu}_i^k + \epsilon_i^l,
\quad \text{subject to} \quad \bar{\mu}^{l}_k \in [0,1].
\]

While constrained least squares (CLS), per se, could in principle be applied, there is no widely used or practically convenient implementation in common econometric toolkits, so we adopt a reparameterization approach. Specifically, we apply a logistic transformation,
\[
\bar{\mu}^{l}_k = \frac{\exp(z^{l}_k)}{1 + \exp(z^{l}_k)},
\]
which maps unconstrained parameters $z^{l}_k \in \mathbf{R}$ into the unit interval.\footnote{This reparameterization enforces component-wise bounds in $[0,1]$, but by itself it does not impose simplex summation. In the with-intercept CLS column, we use the omitted-category linear form (three regressors plus intercept in the four-state case) and recover the omitted coefficient by linear normalization. In the no-intercept CLS column, we estimate four constrained coefficients directly as in the displayed equation above.}

\begin{table}[htbp]
\footnotesize
\centering
\tiny
\begin{threeparttable}
\caption{Robustness Check: PMBA Performance Across Different Implementation Methods}\label{tab:rob1}
\begin{tabular}{cccccccccc}
\toprule
\textbf{NFT} & \textbf{True State} & \multicolumn{2}{c}{Regression With Intercept} & \multicolumn{2}{c}{Regression Without Intercept} & \multicolumn{2}{c}{Regression Without IV} & Majority \\
\cmidrule(lr){3-4} \cmidrule(lr){5-6} \cmidrule(lr){7-8}
& & \textit{OLS} & \textit{CLS} & \textit{OLS} & \textit{CLS} & \textit{OLS} & \textit{CLS} & \\
\midrule
2  & 1 & 1 & 2 & 1 & 2 & 2 & 2 & 1 \\
11 & 1 & 1 & 1 & 1 & 1 & 1 & 4 & 1 \\
17 & 1 & 1 & 1 & 1 & 1 & 1 & 1 & 1 \\
18 & 1 & 2 & 2 & 2 & 2 & 3 & 4 & 3 \\
7  & 2 & 1 & 1 & 1 & 1 & 3 & 1 & 1 \\
8  & 2 & 3 & 2 & 3 & 2 & 3 & 1 & 1 \\
9  & 2 & 1 & 4 & 1 & 1 & 4 & 4 & 1 \\
10 & 2 & 4 & 4 & 4 & 4 & 2 & 2 & 1 \\
13 & 2 & 2 & 2 & 2 & 2 & 3 & 4 & 1 \\
15 & 2 & 1 & 1 & 1 & 1 & 1 & 1 & 1 \\
19 & 2 & 1 & 1 & 1 & 1 & 1 & 1 & 1 \\
1  & 3 & 3 & 4 & 3 & 3 & 2 & 4 & 1 \\
4  & 3 & 1 & 1 & 1 & 4 & 1 & 1 & 1 \\
6  & 3 & 4 & 4 & 4 & 4 & 2 & 2 & 1 \\
12 & 3 & 4 & 4 & 4 & 4 & 2 & 4 & 2 \\
14 & 3 & 1 & 1 & 1 & 1 & 1 & 1 & 1 \\
3  & 4 & 2 & 3 & 2 & 3 & 3 & 4 & 1 \\
5  & 4 & 1 & 1 & 1 & 1 & 1 & 1 & 1 \\
16 & 4 & 1 & 4 & 1 & 3 & 2 & 4 & 1 \\
20 & 4 & 1 & 4 & 1 & 1 & 2 & 2 & 1 \\
\midrule
\multicolumn{2}{r}{Average absolute error by range} & 0.25 & 0.50 & 0.25 & 0.50 & 0.75 & 1.75 & 0.50 \\
& & 1.00 & 1.00 & 1.00 & 0.86 & 1.00 & 1.14 & 1.00 \\
& & 1.20 & 1.40 & 1.20 & 1.00 & 1.40 & 1.40 & 1.80 \\
& & 2.75 & 1.00 & 2.75 & 2.00 & 2.00 & 1.25 & 3.00 \\
\midrule
Paired $t$-test & $p$-value:  & 0.0481 & 0.0331 & 0.0481 & 0.0176 & 0.0675 & 0.3123 & \\
\bottomrule
\end{tabular}

\begin{tablenotes}[flushleft]
\footnotesize
\item \textit{Notes:} The four rows under ``Average absolute error by range'' report the mean absolute error for each true price range (Ranges 1--4). The first four columns use IV-based implementations (with/without intercept, OLS/CLS); the next two columns report no-IV implementations. OLS-IV with intercept and OLS-IV without intercept produce identical state predictions in this sample. Relative to majority voting, several PMBA implementations remain significantly better, while no-IV variants are directionally favorable but statistically weaker.
\end{tablenotes}
\end{threeparttable}
\end{table}

\subsection{NFTs Used in the Experiment}

This section provides the complete list of experimental materials. Table~\ref{tab:nft_list} displays all 20 NFTs used in the experiment, including their collection names, individual identifiers, actual market prices (in both ETH and SGD), and the price ranges they were assigned to for the belief elicitation task.

\begin{table}[htbp]
\centering
\caption{NFTs Used in the Experiment}\label{tab:nft_list}

\rotatebox{90}{%
\begin{minipage}{0.9\textheight} 
\centering
\begin{tabular}{@{}rllrrr@{}}
\toprule
\textbf{ID} & \textbf{Collection} & \textbf{Number} & \textbf{Price (ETH)} & \textbf{Price (SGD)} & \textbf{Range} \\
\midrule
1 & Azuki & 2765 & 4 & 10102.96 & 3 \\
2 & BBKC & 5272 & 0.0003 & 0.76 & 1 \\
3 & Bored Ape & 4491 & 13 & 32834.62 & 4 \\
4 & Citizen & 3769 & 6 & 15154.44 & 3 \\
5 & CryptoPunk & 4297 & 39 & 98503.86 & 4 \\
6 & Doodle & 9355 & 3.4899 & 8814.58 & 3 \\
7 & Enchanter Arabella of the Villa & -- & 0.31 & 782.98 & 2 \\
8 & Human Unreadable & 73 & 1.65 & 4167.47 & 2 \\
9 & One Gravity & 1770 & 1.27 & 3207.69 & 2 \\
10 & Lil Pudgy & 20828 & 1.089 & 2750.53 & 2 \\
11 & Memory & 343 & 0.0758 & 191.45 & 1 \\
12 & Meridian & 942 & 3.5 & 8840.09 & 3 \\
13 & Mutant Ape & 7734 & 2.159 & 5453.07 & 2 \\
14 & Milady & 9026 & 4.49 & 11340.57 & 3 \\
15 & Predator, First Draft & 291 & 1.05 & 2652.03 & 2 \\
16 & Pudgy Penguin & 8693 & 9.88 & 24954.31 & 4 \\
17 & Supduck & 6639 & 0.017 & 42.94 & 1 \\
18 & The Forbidden Museum Hall & -- & 0.0455 & 114.92 & 1 \\
19 & Wealthy Hypio Babies & 5042 & 1.62 & 4091.70 & 2 \\
20 & GigaChad & 95 & 25 & 63143.50 & 4 \\
\bottomrule
\end{tabular}

\vspace{1em}
\noindent\textit{Note:} This table shows all 20 NFTs used in our experiment. Prices are listed in both ETH (Ethereum cryptocurrency) and SGD (Singapore Dollars). The \textbf{Range} column shows which of our four price categories each NFT falls into: 1 (SGD 0--250), 2 (SGD 250--7,500), 3 (SGD 7,500--21,500), or 4 (SGD 21,500+). We verified that all NFTs stayed in their assigned ranges during the experiment.
\end{minipage}
} 

\end{table}

\subsection{Replication of Surprisingly Popular Procedure}

This section replicates the Surprisingly Popular (SP) procedure using our experimental data. Table~\ref{tab:sp_replication} shows the SP predictions for each NFT, comparing the average first-order beliefs with the average second-order beliefs to determine which states are ``surprisingly popular'' and providing statistical tests for the significance of these differences.

\begin{table}[htbp]\centering
\scriptsize
\begin{threeparttable}
\caption{Replication of Surprisingly Popular (SP) Procedure Results}\label{tab:sp_replication}
\sisetup{table-format=2.2, table-number-alignment=center}
\setlength{\tabcolsep}{4pt}
\renewcommand{\arraystretch}{0.95}
\begin{tabular}{ccccc}
\toprule
\textbf{NFT} & \textbf{True State (Range)} & $\hat{\mu}_{(1\&2)}$ & \textbf{Mean of }$\alpha_{(1\&2)}$ & \textbf{SP Prediction}\\
\midrule
2 & 1\&2 (1) & 68.73\% & 69.93\% & 3\&4 \\
   \textit{(N=103)} & & & (0.3139) & \\
  
11 & 1\&2 (1) & 63.32\% & 62.52\% & 1\&2 \\
    \textit{(N=99)} & & & (0.3973) & \\
   
17 & 1\&2 (1) & 80.10\% & 77.34\% & 1\&2 \\
   \textit{(N=94)} & & & (0.1251) & \\
   
18 & 1\&2 (1) & 48.74\% & 45.84\% & 1\&2 \\ 
   \textit{(N=99)} & & & (0.1577) & \\
\midrule
7 & 1\&2 (2) & 77.27\% & 78.16\% & 3\&4 \\
  \textit{(N=97)} & & & (0.3506) & \\

8 & 1\&2 (2) & 65.39\% & 66.99\% & 3\&4 \\
   \textit{(N=96)} & & & (0.2805) & \\

9 & 1\&2 (2) & 65.03\% & 60.72\% & 1\&2 \\
   \textit{(N=93)} & & & (0.0864) & \\

10 & 1\&2 (2) & 75.25\% & 75.73\% & 3\&4 \\
    \textit{(N=97)} & & & (0.4268) & \\

13 & 1\&2 (2) & 59.95\% & 61.73\% & 3\&4 \\
    \textit{(N=102)} & & & (0.2943) & \\

15 & 1\&2 (2) & 75.48\% & 70.10\% & 1\&2 \\
    \textit{(N=100)} & & & (0.0254) & \\

19 & 1\&2 (2) & 77.39\% & 73.41\% & 1\&2 \\
    \textit{(N=104)} & & & (0.0713) & \\
\midrule
1 & 3\&4 (3) & 65.93\% & 64.62\% & 1\&2 \\
   \textit{(N=102)} & & & (0.3261) & \\

4 & 3\&4 (3) & 73.74\% & 75.15\% & 3\&4 \\
   \textit{(N=100)} & & & (0.2875) & \\

6 & 3\&4 (3) & 72.64\% & 72.36\% & 1\&2 \\
   \textit{(N=100)} & & & (0.4562) & \\

12 & 3\&4 (3) & 56.52\% & 51.82\% & 1\&2 \\
    \textit{(N=101)} & & & (0.0586) & \\

14 & 3\&4 (3) & 80.86\% & 80.41\% & 1\&2 \\
   \textit{(N=98)} & & & (0.4291) & \\
\midrule
3 & 3\&4 (4) & 54.24\% & 57.20\% & 3\&4 \\
  \textit{(N=101)} & & & (0.1874) & \\

5 & 3\&4 (4) & 77.07\% & 77.17\% & 3\&4 \\
  \textit{(N=101)} & & & (0.4839) & \\

16 & 3\&4 (4) & 72.01\% & 74.99\% & 3\&4 \\
   \textit{(N=100)} & & & (0.1236) & \\

20 & 3\&4 (4) & 64.76\% & 62.74\% & 1\&2 \\
    \textit{(N=103)} & & & (0.2470) & \\
   
\bottomrule
\end{tabular}

\begin{tablenotes}[flushleft]\footnotesize
\item \textit{Notes:} This table shows how the SP procedure works with our NFT data. Since SP only works with 2 categories, we group our 4 price ranges into 2: ``low prices'' (Ranges 1 \& 2) and ``high prices'' (Ranges 3 \& 4). For each NFT, the number in parentheses is the one-sided $p$-value for testing whether $\hat{\mu}_{(1\&2)}-\text{Mean}(\alpha_{(1\&2)})=0$. Most tests are not statistically significant (high $p$-values), meaning SP struggles to distinguish between the two categories in our setting.
\end{tablenotes}
\end{threeparttable}
\end{table}

\end{document}